\documentclass[11pt]{article}

\usepackage{color}
\usepackage{graphicx}
\usepackage{amsmath}
\usepackage{hyperref}
\usepackage{verbatim}
\usepackage{natbib}
\usepackage{amssymb}
\usepackage{amsfonts}
\usepackage{mathtools}
\usepackage{amsthm}
\usepackage{calrsfs}
\usepackage{upgreek}
\usepackage{bm}
\usepackage{bbm}
\usepackage{color}
\usepackage{xspace}
\usepackage[margin=1 in]{geometry}
\usepackage[misc]{ifsym}
\usepackage{lscape,epsfig}
\usepackage{float}

\usepackage{multirow}

\usepackage{url}
\makeatletter
\g@addto@macro{\UrlBreaks}{\UrlOrds}
\makeatother

\newcommand{\eqnum}{\leavevmode\hfill\refstepcounter{equation}\textup{(\theequation)}}
\def\R{\mathbb{R}}
\newcommand{\iid}[0]{i.i.d.\xspace}

\newcommand{\norm}[1]{\lVert{#1}\rVert}

\newcommand{\PP}[1]{\mathbb{P}\left\{{#1}\right\}} 
\newcommand{\EE}[1]{\mathbb{E}\left[{#1}\right]} 

\def\R{\mathbb{R}}
\def\Z{\mathbf{Z}}
\def\W{\mathbf{W}}

\newcommand\eqd{\stackrel{\mathclap{\normalfont\mbox{d}}}{=}}

\newcommand{\X}{\mathbf{X}}
\newcommand{\Y}{\mathbf{Y}}

\newcommand{\bma}{\mathbf{a}}
\newcommand{\bmb}{\mathbf{b}}
\newcommand{\bms}{\mathbf{s}}
\newcommand{\bmbeta}{\bm{\beta}}
\newcommand{\bmmu}{\bm{\mu}}
\newcommand{\bmeta}{\bm{\eta}}
\newcommand{\bmSigma}{\bm{\Sigma}}
\newcommand{\bmepsilon}{\bm{\epsilon}}
\newcommand{\bmTheta}{\bm{\Theta}}
\newcommand{\bmvarepsilon}{\bm{\varepsilon}}

\DeclareMathOperator{\diag}{diag}

\newcommand{\Xp}{\widetilde{\X}}

\newcommand{\indep}{\perp \!\!\! \perp}
\newcommand{\iidsim}{\stackrel{\mathrm{iid}}{\sim}}

\newtheorem{theorem}{Theorem}

\newtheorem{lemma}{Lemma}

\newtheorem{definition}{Definition}

\title{FDR Controlled Multiple Testing for Union Null Hypotheses: A Knockoff-based Approach}

\author{Ran Dai$^{*}$\email{ran.dai@unmc.edu} and
Cheng Zheng$^{**}$\email{cheng.zheng@unmc.edu} \\
Department of Biostatistics, University of Nebraska Medical Center, Omaha, Nebraska, U.S.A.}

\author{Ran Dai\thanks{Department of Biostatistics, University of Nebraska Medical Center, Omaha, Nebraska, U.S.A., ran.dai@unmc.edu},  Cheng Zheng\thanks{Department of Biostatistics, University of Nebraska Medical Center, Omaha, Nebraska, U.S.A., cheng.zheng@unmc.edu}}

\begin{document}

\maketitle

\begin{abstract}
False discovery rate (FDR) controlling procedures provide important statistical guarantees for the replicability in signal identification based on multiple hypotheses testing. In many fields of study, FDR controlling procedures are used in high-dimensional (HD) analyses to discover features that are truly associated with the outcome. In some recent applications, data on the same set of candidate features are independently collected in multiple different studies. For example, gene expression data are collected at different facilities and with different cohorts, to identify the genetic biomarkers of multiple types of cancers. These studies provide us opportunities to identify signals by considering information from different sources (with potential heterogeneity) jointly. This paper is about how to provide FDR control guarantees for the tests of union null hypotheses of conditional independence. We present a knockoff-based variable selection method (\textit{Simultaneous knockoffs}) to identify mutual signals from multiple independent data sets, providing exact FDR control guarantees under finite sample settings. This method can work with very general model settings and test statistics. We demonstrate the performance of this method with extensive numerical studies and two real data examples.

\end{abstract}

\section{Introduction}
\label{sec:intro}

There is a pressing need in making discoveries by analyzing information from multiple sources jointly. With recent advances in scientific research, data on the same set of candidate features are often collected independently from multiple sources. For example, social scientists collect data on the economic and socioeconomic status of people from different community groups. In genome-wide association studies (GWAS), associations of genome features with multiple different outcomes of interest are studied in multiple experiments \citep{uffelmann2021}. These data motivate us to identify mutual signals from multiple experiments for purposes like reproducibility research or mediator identification. This paper focuses on how to identify mutual signals from multiple independent studies and provide variable selection accuracy guarantees with mild design and model assumptions.

Now we formulate the mutual signal identification problem in statistical terms. Assume we have data from $K$ independent experiments and denote $[K] = \{1,\cdots,K\}$. Within the $k$-th experiment, $(Y^k_{i},X^k_{i1},\cdots,X^k_{ip}) \iidsim \mathcal{D}_k$, $i=1,\cdots,n_k$. In our setting, the outcome variables $Y^1,\cdots, Y^K$ can be of different data types and $(X_1^k,\cdots, X_p^k)$ can have different distributions among the different experiments. Here we denote $Y^k$s and $X^k$s as continuous variables throughout the paper for the simplicity of notation. In practice, they can be of other data types (continuous/count/nominal/ordinal/mixed). For example, $Y^k$s can be different disease outcomes and $\X^k$s can be gene expression data measured on different scales. Define $H_{0j}^k$ as the null hypothesis indicating the $j$-th feature not being a signal in the $k$-th experiment (i.e. $X_j^k \indep Y^k | X^k_{-j}$ where $X^k_{-j} := \{X^k_1,\cdots,X^k_{p}\} \setminus X^k_j$), and denote $\mathcal{H}^k = \{j\in [p]: H_{0j}^k ~\text{is true}\}$, where $[p]:=\{1,\cdots,p\}$. 
Instead of testing the $H_{0j}^k$s, we are interested in testing the union null hypotheses 
\begin{equation}\label{eqn:H0j}
    H_{0j}=\cup_{k=1}^K H_{0j}^k, ~~\text{for}~ j \in [p].
\end{equation} 
\begin{equation}\label{eqn:H}
\text{We define ~}    \mathcal{S} = \{j\in [p]: H_{0j} ~\text{is false}\}~ \text{and  $\mathcal{H} = \mathcal{S}^c = \cup_{k=1}^K \mathcal{H}^k = \{j\in [p]: H_{0j} ~\text{is true}\}$}.
\end{equation} 
We aim at developing a selection procedure returning a selection set $\widehat{\mathcal{S}} \subseteq [p]$ with a controlled false discovery rate (FDR), which is the expected false discovery proportion (FDP):  \begin{equation}\label{eqn:fdr}
\text{FDR}(\widehat{\mathcal{S}}) = \EE{\text{FDP}(\widehat{\mathcal{S}})}= \EE{\frac{|\widehat{\mathcal{S}} \cap \mathcal{H}|}{|\widehat{\mathcal{S}}|\vee 1}}.
\end{equation}
To begin, we give some examples to motivate our method.

\subsection{Examples}\label{sec:example}

The problem of testing multiple union null hypotheses is related to many important scientific areas, for example, the reproducibility analysis in GWAS \citep{bogomolov2013,heller2014,heller2014b}, comparative research in genomics studies \citep{rittschof2014}, and mediation analysis \citep{sampson2018}. Below, we give several motivating examples that can be considered as problems of identifying the mutual signal set $\widehat{\mathcal{S}}$. 

\subsubsection{Repeatability research}
In some fields of biology, experimental results are required to agree with each other under conditions that include the same measurement procedure, same operators, same measuring system, same operating conditions, same location, and replicate measurements on the same or similar objects \citep{plant2020reproducibility, ioannidis2009repeatability}. 
It aims at identifying signals in repeated experiments.
 Mathematically, $K$ independent data sets $(\Y^1,\X^1),\cdots, (\Y^K,\X^K)$ are collected, where $\Y^{k}\in \R^{n_k}$ and $\X^k\in \R^{n_k\times p}$ for $k \in [K]$, and $(Y^k_{i},X^k_{i1},\cdots,X^k_{ip})\iidsim \mathcal{D}$ for all $k \in [K]$ and $i \in [n_k]$. To identify the $j$-th feature as a signal, we test the union null hypothesis $H_{0j}$ as defined in equation \eqref{eqn:H0j}. 

As a remark, for repeatability research, when we assume $(Y^k_{i},X^k_{i1},\cdots,X^k_{ip})$ are $\iid$ for all $k \in [K]$ and $i \in n_k$, the union null hypotheses set $\mathcal{H}$ is identical to the null hypotheses sets from the individual experiments for all $k \in [K]$. In this case, we can alternatively pool the data and do analyses on the pooled data set to improve power \citep{nci2007replicating}. However, in many cases, this is not practical because of privacy and data ownership issues.  

\subsubsection{Mutual features across populations with heterogeneity}

When identifying risk factors, designing new treatments, or making policies, there is a need to guarantee consistent results across heterogeneous sub-populations. For example, reproducibility analysis aims at identifying findings independently discovered across multiple experiments. These experiments could be slightly different because they are conducted in different institutions, by different experimenters, or at different times. For example, in genetic studies, the association between single nucleotide polymorphisms (SNPs) and phenotype is recognized as a scientific finding only if it has been discovered from different independent studies with the same features and different cohorts \citep{heller2014}.

To form the conditional independence tests for this case, we have $K$ independent studies, where 
$(Y^k_{i},X^k_{i1},\cdots,X^k_{ip}) \iidsim \mathcal{D}_k, ~\text{for}~ i \in [n_k]$, for each $k \in [K]$, and we test the conditional independence $Y^k \indep X^k_{j} | X^k_{-j}, ~\text{for}~ k \in [K] ~\text{and}~ j \in [p]$. The $j$-th feature is a mutual signal if and only if the union null hypothesis $H_{0j}$ does not hold. 

As a remark, \cite{heller2014} proposed a \textit{repFDR} method, which also provides the FDR control guarantees on testing multiple union null hypotheses. This method is based on the Benjamini–Hochberg (BHq) procedure \citep{benjamini1995}; and it assumes that the vector of test statistics for hypotheses in each study are jointly independent or are positive regression dependent (PRDS) on the subset of true null hypotheses. This assumption does not hold in general in our settings. There is a modification of this method that allows for an arbitrary dependence, however, it is known to be very conservative \citep{benjamini2001}.

\subsubsection{High-dimensional mediator selection}
In many scientific fields, it is important to identify features that are associated with multiple responses. In particular, mediators can be discovered from simultaneous feature-treatment and feature-outcome associations. For example, suppose we aim at identifying gene expression mediators that are both associated with the treatment and the risk of a certain disease. To do this, we jointly use information from two independent studies, one on the associations between the gene expressions and the treatment, the other on the association between the gene expressions and the risk of the disease with the treatment being fixed. The selection of mediators from high-dimensional gene expression features can be framed as a problem of testing the union null hypotheses with $K=2$. In particular, 
$(Y^k_{i},X^k_{i1},\cdots,X^k_{ip}) \iidsim \mathcal{D}_k, ~\text{for}~ k=1,2$,
where $Y^1$ and $Y^2$ are the treatment and the outcome respectively. Notice that in this example the true signal sets for $Y^1$ and $Y^2$ are not necessarily identical. We test the conditional independence
$Y^k \indep X^k_{j} | X^k_{-j}, ~\text{for}~ k=1,2, ~\text{and}~ j\in[p]$.
The $j$-th feature is a mediator if and only if the union null hypothesis $H_{0j}$ does not hold.

\subsection{Prior work}

\paragraph{Current advance in FDR control for identifying simultaneous signals}

For reproducibility research, Bogomolov et al. proposed methods based on the BH procedure by selecting features that are commonly selected among all the experiments \citep{heller2014, bogomolov2013, bogomolov2018}. There are multiple works based on computing the local FDR as the optimal scalar summary of the multivariate test statistics \citep{chi2008, heller2014b}. Recently, \citet{xiang2019} presented the signal classification problem for multiple sequences of multiple tests, where the identification of the simultaneous signal is a special case, and \citet{zhao2020} proposed a nonparametric method for asymptotic FDR control in identifying simultaneous signals. However, all methods above assume not only the independence of the experiments but also the independence (or PRDS) of the p-values for the features within each experiment, which is not realistic in many complex high-dimensional data applications, such as the GWAS and other omics data.

\paragraph{Knockoff-based methods}
For multiple testing problems within a single experiment, there are recent advances in relaxing the assumption of independence among the features. Powerful knockoff-based methods have been developed for exact FDR control in selecting features with conditional associations with the response \citep{barber2015, candes2018}. The original knockoff filter proposed by \citet{barber2015} works on linear models assuming no knowledge of the design of covariates, the signal amplitude, or the noise level. It achieves exact FDR control under finite sample settings. It is also extended to work with high-dimensional settings \citep{barber2019}. Later \citet{candes2018} proposed the Model-X knockoff method, extending the knockoff filter to achieve exact FDR control for nonlinear models. This method allows the conditional distribution of the response to be arbitrary and completely unknown but requires the distribution of $\X$ to be known. \citet{barber2020} further showed that the Model-X knockoff method is robust against errors in the estimation of the distribution of $\X$. In addition, \citet{huang2020} relaxed the assumptions of the Model-X knockoff method so that the FDR can be controlled as long as the distribution of $\X$ is known up to a parametric model. There are also abundant publications on the construction of knockoffs with an approximated distribution of $\X$. \citet{romano2019} developed a Deep knockoff machine using deep generative models. \citet{liu2019} developed a Model-X generating method using deep latent variable models. More recently, \citet{bates2020} proposed an efficient general metropolized knockoff sampler. \citet{spector2020} proposed to construct knockoffs by minimizing the reconstructability of the features. Knockoff-based methods have also been extended to test the intersection of null hypotheses. In this direction, group and multitask knockoff methods \citep{dai2016}, and prototype group knockoff methods \citep{chen2020} have been proposed. Variants of knockoff methods have become useful tools in scientific research. For example, to identify the variations across the whole genome associated with a disease, \citet{sesia2019} developed a hidden Markov model knockoff method for FDR control in GWAS.

\subsection{Our contributions}
{\color{black} In this paper, we propose a knockoff-based procedure to establish exact FDR control in selecting mutual signals from multiple conditional independence tests, assuming very general conditional models. The main contributions of this paper are summarized below: 
\begin{enumerate}
    \item We construct a knockoff-based procedure for testing the union null hypotheses for feature selection, namely the \textit{Simultaneous knockoffs}. This procedure can work on general conditional dependence models $Y|\X$ and data structures in $\X$.
    \item We prove that the \textit{Simultaneous knockoff} method can lead to exact FDR control in testing multiple union null hypotheses for feature selection under finite sample settings.
    \item We show that a broad class of filter statistics can be used for this method, and give general recipes for generating different powerful statistics.
    \item We demonstrate the FDR control property and the power of our method with extensive simulation settings. We also illustrate the application with two real data examples.
\end{enumerate}
} 

The rest of the paper is organized as follows. In Section \ref{sec:meth}, we present the \textit{Simultaneous knockoff} framework. In Section \ref{sec:theo}, we give the theoretical guarantees for exact FDR control of the \textit{Simultaneous knockoff} method in finite sample settings and the robustness result for the potential misspecification of the distribution of $\X$. In Section \ref{sec:simulation}, we show the empirical performance of the \textit{Simultaneous knockoff} method under different model assumptions and data structures. Finally, in Section \ref{sec:application}, we apply the \textit{Simultaneous knockoff} procedure to two real data examples.


\section{Methods}
\label{sec:meth}
In this section, we present the \textit{Simultaneous knockoff} procedure. This procedure can be paired with both the Fixed-X knockoffs \citep{barber2015} and the Model-X knockoffs \citep{candes2018} to allow for very general model settings and various data structures in real data applications. Before presenting the \textit{Simultaneous knockoff} method, we briefly review the Fixed-X and the Model-X knockoff methods.

\subsection{The Fixed-X and the Model-X knockoff procedures}

The high-level idea behind the knockoff methods is to construct a ``knockoff" copy of the covariates, retaining their inner structures. Unlike the ``true" covariates, the knockoff copies are created independent of the response. These knockoff variables are then mixed into the model to monitor the FDP during the selection. Heuristically speaking, if one variable is a true signal, it is more likely to be selected than its knockoff copy, otherwise, it is equally likely to be selected as its knockoff copy. Therefore, by counting the number of knockoff variables entering the selected set, the FDP can be (over) estimated. 

\paragraph{Model settings} For the Fixed-X knockoff method, the setup is a decentralized linear model, $\Y =\X\bmbeta + \bmvarepsilon, ~\text{where $\Y \in \R^{n}$, $\X \in \R^{n\times p}$, $\bmbeta \in \R^p$ and $\bmvarepsilon \in \R^{n} \sim \mathcal{N}(0,\sigma^2 I_n).$}$
This method has weak assumptions on the covariates $\X$,  the amplitudes of the
unknown regression coefficients $\beta$, and does not require the noise level ($\sigma^2$) to be known. The Model-X knockoff method works on more general conditional model settings. It does not require the dependence of $Y|\X$ to be known by assuming the knowledge of the distribution of $\X$ (or if the distribution of $\X$ can be well approximated). Therefore, it can work with many more models such as generalized linear models (GLMs) or nonlinear models.

\paragraph{Algorithm} There are four main steps in the knockoff procedure listed as below.

\begin{itemize}
    \item \textit{Knockoff construction.} A set of knockoff features $\Xp = [\widetilde{\X}_1 \ \dots \ \widetilde{\X}_p]$ are constructed in this step. For the Fixed-X knockoff construction, $\Xp$ needs to satisfy that, for some vector $\bms\geq 0$,
\begin{equation}\label{eqn:knockoff_condition}
\Xp^\top \Xp = \X^\top \X, \quad \Xp^\top \X = \X^\top \X - \diag\{\bms\}.\end{equation} 
For the Model-X knockoff construction, $\Xp$ needs to satisfy the pairwise exchangeability condition:
\begin{equation}\label{eqn:modelX_knockoff_condition}
    [ \X \ \Xp ]_{\text{Swap(j)}} \eqd [ \X \ \Xp ]~\text{and}~ \Xp \perp \!\!\! \perp \Y | \X \quad \text{for all}~ j \in [p], 
\end{equation}  where $\text{Swap(j)}$ stands for exchanging the $j$-th column and the $(j+p)$-th column of $[ \X \ \Xp ]$, and $A \eqd B$ indicates $A$ and $B$ are identical in distribution. $\Xp$ can be generated using various algorithms \citep{barber2015, romano2019,liu2019,bates2020,spector2020}. More knockoff construction details are reviewed in Web Appendix A.1.

\item \textit{Test statistics calculation.} Appropriate test statistics need to be calculated for the features $\X$ and their knockoff copies $\Xp$. For the Fixed-X knockoffs, the statistics $[\Z,\widetilde{\Z}] \in \R^{2p}$ needs to be a function of $([\X \widetilde{\X}]^\top[\X \widetilde{\X}], [\X \widetilde{\X}]^\top \Y)$. For the Model-X knockoffs, $[\Z,\widetilde{\Z}]$ needs to be a function of $([ \X \ \Xp ], \Y)$, such that if we swap features $\X_j$ with its corresponding knockoffs $\widetilde{\X}_j$, then the statistics $Z_j$ and $\widetilde{Z}_j$ get swapped. Examples of the $[\Z,\widetilde{\Z}]$ statistics are provided in Web Appendix A.2. 

\item \textit{Filter statistics calculation.} We construct the filter statistics $\W \in \R^p$ such that $W_j = f(Z_j, \widetilde{Z}_j)$ where $f$ is an antisymmetric function, i.e. $f(x,y) = - f(y,x)$. Without loss of generality, we further let $f(x,y)>0$ when $x>y$. If $X_j$ is a signal, we would expect $\PP{Z_j > \widetilde{Z}_j} > 0.5$, while if $X_j$ is not a signal, we would expect $Z_j$ and $\widetilde{Z}_j$ to have the same distribution. Thus, we expect $W_j$ to have a positive sign with $>0.5$ probability if $X_j$ is a signal and with $0.5$ probability if $X_j$ is not a signal. This allows us to estimate the FDP in $\widehat{S}(t):= \{j: W_j \geq t\}$ as 
\begin{equation}\label{eqn:fdp} \widehat{\text{FDP}}(t) = \frac{\# W_j\leq -t}{(\# W_j \geq t) \vee 1}.\end{equation}

\item \textit{Threshold calculation and feature selection.} With the knockoff filter, we select
\begin{equation}\label{eqn:knockoff}
   \widehat{S} = \{j: W_j \geq \tau\},~\text{where}~ \tau = \min\left\{t \in \mathcal{W}_+: \frac{\#\{j:W_j \leq -t\}}{\#\{j:W_j\geq t\}\vee 1}\leq q\right\} .
\end{equation}
With a more conservative knockoff+ filter, we select
\begin{equation}\label{eqn:knockoff+}
   \widehat{S}_+ = \{j: W_j \geq \tau_+\}, ~\text{where}~ \tau_+ = \min\left\{t \in \mathcal{W}_+: \frac{1+\#\{j:W_j \leq -t\}}{\#\{j:W_j\geq t\}\vee 1}\leq q\right\} .
\end{equation}
Here $q$ is the target FDR level and $\mathcal{W}_+ = \{|W_j|:|W_j|>0\}$. 
\end{itemize}

\subsection{Simultaneous knockoff framework}

In this section, we propose the general \textit{Simultaneous knockoff} framework, which enables us to use the knockoff approach for FDR control in testing the union null hypotheses of conditional independence. This approach enjoys very general model assumptions and exact FDR control guarantees in finite sample settings. 

\subsubsection{Preliminaries}
One na\"{i}ve idea to identify mutual signals in $K$ experiments is to select the intersection set of the variables selected from the individual experiments. However, this method cannot control the FDR (see more details in Section \ref{sec:simulation}). Therefore, we alternatively aim at constructing valid filter statistics $\W$ to allow the estimation of the FDP in our multiple testing of the union null hypotheses using equation \eqref{eqn:fdp}. We establish a general recipe to construct such $\W$s with only the summary statistics that can be calculated using the Fixed-X and the Model-X knockoff methods from single experiments. To begin, we give several definitions.

\begin{definition}
(Swapping) For a set $S \subseteq [p]$, and for a vector $\mathbf{V} = (V_1, \cdots, V_{2p}) \in \R^{2p}$, $\mathbf{V}_{\textnormal{Swap}(S)}$ indicates the swapping of $V_j$ with $V_{j+p}$ for all $j \in S$. 
\end{definition}

\begin{definition}\label{def:flip-sign}
(Flip sign function) A function $f: \R^{2p}\rightarrow \R^p$ is called a flip sign function if it satisfies that for all $S\subset [p]$, $f([\Z,\widetilde{\Z}]_{\textnormal{Swap}(S)})=f([\Z,\widetilde{\Z}])\odot \bmepsilon(S)$ where $\Z,\widetilde{\Z},\bmepsilon(S) \in \R^p$, and $\epsilon(S)_j=-1$ for all $j\in S$ and $\epsilon(S)_j=1$ otherwise. Here $\odot$ denotes the Hadamard product.
\end{definition}

An example of a flip sign function is $f([\Z,\widetilde{\Z}])=\Z-\widetilde{\Z}$. \eqnum\label{eqn:one-flip} 
More examples of flip sign functions and their relationships to antisymmetric functions are discussed in Web Appendix A.4.

\begin{definition}\label{def:one-swap}
(One swap combining function (OSCF)) A function $f: \R^{2pK}\rightarrow \R^{2p}$ is called a one swap combining function (OSCF) if it satisfies that for all $k\in [K]$ and all $S \subset [p]$, \[\left[f([\Z^1,\widetilde{\Z}^1],\cdots,[\Z^K,\widetilde{\Z}^K])\right]_{\textnormal{Swap}(S)}=f([\Z^1,\widetilde{\Z}^1],\cdots,[\Z^k,\widetilde{\Z}^k]_{\textnormal{Swap}(S)},\cdots,[\Z^K,\widetilde{\Z}^K]),\]
where $\Z^k,\widetilde{\Z}^k \in \R^p$ for all $k \in [K]$.
\end{definition}

As a remark, the definition of an OSCF implicitly requires 
that for any set $S \subset [p]$,
\begin{align*}
   &f([\Z^1,\widetilde{\Z}^1]_{\textnormal{Swap}(S)},\cdots,[\Z^K,\widetilde{\Z}^K])\\
    =&f([\Z^1,\widetilde{\Z}^1], [\Z^2,\widetilde{\Z}^2]_{\textnormal{Swap}(S)},\cdots,[\Z^K,\widetilde{\Z}^K])= \cdots =f([\Z^1,\widetilde{\Z}^1],\cdots,[\Z^K,\widetilde{\Z}^K]_{\textnormal{Swap}(S)}).
\end{align*}

An example of the OSCF can be defined as below: 
Let $\bma=(a_1,\cdots,a_K)\in A$ where $A=\{0,1\}^K$. We separate $A$ to two sets: the even set $A_e=\{\bma:\textnormal{mod}(||\bma||_1,2)=\textnormal{mod}(K,2)\}$; and the odd set $A_o=\{\bma:\textnormal{mod}(||\bma||_1,2)=\textnormal{mod}(K+1,2)\}$. Then we obtain an OSCF function $[\Z,\widetilde{\Z}]=f([\Z^1,\widetilde{\Z}^1],\cdots,[\Z^K,\widetilde{\Z}^K])$ as below:

\begin{align}\label{eqn:one-swap}
 Z_j=\sum_{\bma\in A_e}\prod_{k=1}^K Z_{jk}^{a_k}\widetilde{Z}_{jk}^{1-a_k} \quad \textnormal{and}\quad
\widetilde{Z}_j=\sum_{\bma\in A_o}\prod_{k=1}^K Z_{jk}^{a_k}\widetilde{Z}_{jk}^{1-a_k},    
\end{align}
where $Z_{jk}$ and $\widetilde{Z}_{jk}$ are the $j$-th entry of $\Z^k$ and $\widetilde{\Z}^k$ respectively.

In particular, when $K=2$, this construction can be written as:
\begin{align*}
Z_j=Z_{j1}Z_{j2}+\widetilde{Z}_{j1}\widetilde{Z}_{j2} \quad \textnormal{and}\quad \widetilde{Z}_j=Z_{j1}\widetilde{Z}_{j2}+\widetilde{Z}_{j1}Z_{j2} \quad \text{for}~j=1,\cdots,p.
\end{align*}

More OSCF examples are given in Web Appendix A.3.

\begin{definition}\label{def:one-swap-flip}
(One swap flip sign function (OSFF)) A function $f: \R^{2pK}\rightarrow \R^{p}$ is called a one swap flip sign function (OSFF) if it satisfies that for all $k\in [K]$ and all $S \subset [p]$, \[f([\Z^1,\widetilde{\Z}^1],\cdots,[\Z^k,\widetilde{\Z}^k]_{\textnormal{Swap}(S)},\cdots,[\Z^K,\widetilde{\Z}^K])=f([\Z^1,\widetilde{\Z}^1],\cdots,[\Z^k,\widetilde{\Z}^k],\cdots,[\Z^K,\widetilde{\Z}^K])\odot \bmepsilon(S),\]
where $\Z^k,\widetilde{\Z}^k \in \R^p$ for $k \in [K]$.
\end{definition}

 There are multiple ways to construct OSFFs. As shown in Lemma A1 in Web Appendix A.5, if $f_1:\R^{2pK}\rightarrow \R^{2p}$ is an OSCF and $f_2:\R^{2p}\rightarrow \R^p$ is a flip-sign function, then $f=f_2\circ f_1$ is an OSFF, where $\circ$ denotes the composition of functions. When using the OSCF as defined in \eqref{eqn:one-swap} and using the flip-sign function as defined in \eqref{eqn:one-flip}, we obtain an OSFF $f([\Z^1,\widetilde{\Z}^1],\cdots,[\Z^K,\widetilde{\Z}^K])=\odot_{k=1}^K (\Z^k-\widetilde{\Z}^k)$. Alternative ways to construct OSFFs and more examples are provided in Web Appendix A.5.

\subsubsection{Algorithm}
The \textit{Simultaneous knockoff} procedure is described below:
\begin{itemize}
    
    \item \textit{Step 1: Knockoff construction for the individual experiments.} Denote the knockoff matrices for $\X^1,\cdots,\X^K$ as $\widetilde{\X}^1, \cdots,\widetilde{\X}^K$. For each $k \in [K]$, select a knockoff construction method  (either Fixed or Model-X) as described in Web Appendix A.1 that is compatible with the model setting for the experiment $k$ to generate $\widetilde{\X}^k$.

    \item \textit{Step 2: Test statistics calculation for the individual experiments.} For each experiment $k \in [K]$, choose and calculate statistics $[\Z^k, \widetilde{\Z}^k] \in \R^{2p}$ that are compatible with the knockoff construction method for experiment $k$. Details on the choices of $[\Z^k, \widetilde{\Z}^k]$ can be found in Web Appendix A.2.
 
    \item \textit{Step 3: Calculation of the filter statistics $\W$.} Choose an arbitrary OSFF $f$ as defined in Definition \ref{def:one-swap-flip} and calculate $\W=f([\Z^1,\widetilde{\Z}^1],\cdots,[\Z^K,\widetilde{\Z}^K])$. Examples of OSFFs can be found in Web Appendix A.5.
  
    \item \textit{Step 4: Threshold calculation and feature selection.} Using the filter statistics $\W$ from Step 3, we apply the knockoff+ filter \eqref{eqn:knockoff+} to obtain the selection set $\widehat{S}_+$ under the \textit{Simultaneous knockoff}+ procedure; or apply the knockoff filter \eqref{eqn:knockoff} to obtain $\widehat{S}$ under the \textit{Simultaneous knockoff} procedure.
\end{itemize}

\section{Theoretical Results}
\label{sec:theo}

The main result for this paper is the theoretical guarantee that \textit{Simultaneous knockoff} and \textit{Simultaneous knockoff}+ procedures can control the modified FDR (as defined in \eqref{eqn:mfdr} in Theorem \ref{thm:1}) and FDR respectively. 

\begin{theorem}\label{thm:1}
With the individual experiments satisfying the Fixed-X or the Model-X knockoff model settings, the \textit{Simultaneous knockoff} procedure \eqref{eqn:knockoff} 
controls the modified FDR defined as
\begin{equation}\label{eqn:mfdr}
\textnormal{mFDR}=\EE{\frac{|\widehat{S}\cap \mathcal{H}|}{|\widehat{S}|+1/q}}\leq q, \end{equation}
and the \textit{Simultaneous knockoff}+ procedure \eqref{eqn:knockoff+} controls the usual FDR
\[\EE{\frac{|\widehat{S}_+\cap \mathcal{H}|}{|\widehat{S}_+|\vee 1}}\leq q, ~\text{where $\mathcal{H}$ is the union null set as defined in \eqref{eqn:H}}.\]
\end{theorem}

The Fixed-X and the Model-X knockoff model settings can be found in Web Appendix A.1. The definition of mFDR is close to the FDR, especially when the selection set is relatively large. Although the more conservative \textit{Simultaneous knockoff}+ procedure can achieve exact FDR control, in real data applications, the knockoff filter is more widely used \citep{barber2015,dai2016, candes2018,sesia2019,romano2019}. 

The key step for the proof of Theorem \ref{thm:1} is to show that the signs of the $W_j$s for the union nulls are i.i.d. following a Bernoulli$(\frac{1}{2})$ distribution, and independent of $|W_j|$ for all $j \in \mathcal{H}$. As for the knockoff-based methods, this property effectively guarantees that for all $j \in \mathcal{H}$, there are equal probabilities of selecting the feature and its knockoff copy, which allows the knockoff procedure to (over) estimate the FDP. We show this in the Lemma \ref{lem:1}. The details of the proof can be found in Web Appendix B.

\begin{lemma}\label{lem:1}
Let $\W = f([\Z^1,\widetilde{\Z}^1],\cdots,[\Z^K,\widetilde{\Z}^K])$ where $f$ is an OSFF. Let $\bmepsilon \in \{\pm 1\}^p$ be a sign sequence independent of $\W$, with $\epsilon_j = +1$ for all $j \in \mathcal{S}$ and $\epsilon_j \sim \{\pm 1\}$ for all $j \in \mathcal{H}$. Then
$(W_1, \cdots, W_p) \eqd (W_1\cdot \epsilon_1, \cdots, W_p\cdot\epsilon_p)$. 
\end{lemma}

For models beyond the linear models, we need to use the Model-X knockoffs in the individual experiments. In real applications, the distribution of the candidate features might not be known exactly. In \cite{candes2018}, the robustness against the misspecification of the $\X$ distribution is shown empirically. \cite{barber2020} and \cite{huang2020} further addressed this question theoretically. For the \textit{Simultaneous knockoff} procedure, it is also very important to establish the robustness results against the misspecification of the distribution of $\X$. The following theorem shows the result.

\begin{theorem}\label{thm:2}
Under the definitions in Section \ref{sec:meth}, for any $\epsilon \geq 0$, consider the null variables for which $\min_{k:j\in \mathcal{H}^{k}}\widehat{\textnormal{KL}}_{j}^k\leq \epsilon$, where $\widehat{\textnormal{KL}}_{j}^k = \sum_{i=1}^{n_k} \log\left( \frac{P_{kj}(X^k_{ij}|X^k_{i,-j})Q_{kj}(\widetilde{X}^k_{ij}|X^k_{i,-j})}{Q_{kj}(X^k_{ij}|X^k_{i,-j})P_{kj}(\widetilde{X}^k_{ij}|X^k_{i,-j})}\right)$ where $P$ denotes the true distribution and $Q$ denotes the misspecified distribution. If we use the knockoff+ filter, then the fraction of the rejections corresponding to such nulls obeys
\begin{equation}
    \EE{\frac{|\{j:j\in \widehat{S}\cap \mathcal{H}~\textnormal{and}~\min_{k:j\in \mathcal{H}^{k}}\widehat{\textnormal{KL}}_{j}^k\leq \epsilon\}|}{|\widehat{S}|\vee 1}} \leq q\cdot e^{\epsilon}.
\end{equation}
In particular, this implies that the FDR is bounded as
\begin{equation}
    \textnormal{FDR}\leq \min_{\epsilon \geq 0} \left\{q\cdot e^{\epsilon} + \PP{\max_{j\in H_0}\min_{k:j\in \mathcal{H}_{k}}\widehat{\textnormal{KL}}_{j}^k > \epsilon}\right\}.
\end{equation}
Similarly, if we use the knockoff filter, for any $\epsilon \geq 0$, a slightly modified fraction of the rejections corresponding to nulls with $\min_{k:j\in \mathcal{H}^{k}}\widehat{\textnormal{KL}}_{j}^k\leq \epsilon$ obeys
\begin{equation}
    \EE{\frac{|\{j:j\in \widehat{S}\cap \mathcal{H}~\textnormal{and}~\min_{k:j\in \mathcal{H}^{k}}\widehat{\textnormal{KL}}_{j}^k\leq \epsilon\}|}{|\widehat{S}|+q^{-1}}} \leq q\cdot e^{\epsilon}
\end{equation}
and from this, we obtain a bound on the modified FDR:
\begin{equation}
\EE{\frac{|\widehat{S}\cap \mathcal{H}|}{|\widehat{S}|+q^{-1}}} \leq \min_{\epsilon\geq 0} \left\{q\cdot e^{\epsilon} + \PP{\max_{j\in H}\min_{k:j\in \mathcal{H}^{k}}\widehat{\textnormal{KL}}_{j}^k > \epsilon}\right\}.
\end{equation}
\end{theorem}

In real applications, when we have additional samples of $\X$ (for estimating the distribution of $\X$), we will be able to achieve a small enough $\epsilon$. Otherwise, it has been proposed to evaluate the potential inflation of FDR using simulation \citep{romano2019}. In Theorem \ref{thm:2}, we show an FDR upper bound result for \textit{Simultaneous knockoffs} which is similar to the result in \citet{barber2020} for Model-X knockoffs, to build some statistical foundations for such simulation approach. Below we give an example to demonstrate its application. 

{\color{black} Consider the example of the Gaussian knockoffs in \citet{barber2020}, i.e., $\X^k$ is normally distributed with the mean zero and the variance-covariance matrix $(\bmTheta^k)^{-1}$ and we use the Gaussian knockoff construction method, i.e., sample $\widetilde{\X}^k|\X^k\sim \mathcal{N}(\mathbf{I}-\mathbf{D}^k\widetilde{\bmTheta}^k\X^k,2\mathbf{D}^k-\mathbf{D}^k\widetilde{\bmTheta}^k\mathbf{D}^k)$ where $\widetilde{\bmTheta}^k$ is an estimated version of $\bmTheta^k$ and $\mathbf{D}^k$ is a nonnegative diagonal matrix such that $2\mathbf{D}^k-\mathbf{D}^k\widetilde{\bmTheta}^k\mathbf{D}^k$ is positive definite, then as shown in \citet{barber2020}, we have with probability at least $1-p^{-1}$,
$
max_{j=1,\cdots,p}\widehat{\text{KL}}^k_j\leq 4\delta_{\bmTheta^k}\sqrt{n_k \log(p)}(1+\text{Rem}),
$
where $\delta_{\bmTheta}=\max_{1,\cdots,p}(\bmTheta_{jj})^{-1/2}||\bmTheta^{-1/2}(\widetilde{\bmTheta}-\bmTheta)||_2$, and $\text{Rem}$ is a vanishing term when $n_k^{-1}\log(p)\rightarrow 0$. The graphical Lasso estimator of $\bmTheta^k$ with (unlabeled) sample size $N_k$ satisfy that $||\widehat{\bmTheta}^k-\bmTheta^k||_{\infty}\lesssim \sqrt{\frac{\log(p)}{N_k}}$ and thus $\delta_{\bmTheta^k} \asymp O(\sqrt{\frac{s_{\bmTheta^k}\log (p)}{N_k}})$. So $4\delta_{\bmTheta^k}\sqrt{n_k \log(p)}$ will be small if the unlabeled sample size $N_k$ for each subsample is large enough in the sense that $N_k>>n_ks_{\bmTheta^k}[\log(p)]^2$. Under a special setting where there exists a subset $\Omega\subset [K]$ such that $\mathcal{H}=\cup_{k\in \Omega}\mathcal{H}_k$, then we will just need enough unlabeled sample size within those subsample with index from $\Omega$.  
}

Our theoretical guarantees focus on the control of FDR.  The power is a monotonically decreasing function of $K$ and a monotonically increasing function of $n$. Asymptotically, as $K$ is fixed, and $\frac{\log p}{n} \rightarrow 0$, the power converges to 1 as $n \rightarrow \infty$ (See details in Web Appendix C.3.5).

Since there are no theoretical results on the choice of $\W$ for the most powerful test, we compare the power with several choices of $\W$s numerically. To understand the power of the proposed statistics, we plot the empirical distributions of the filter statistics ($W_j$) for $j \in \mathcal{H}$ and $j \in \mathcal{S}$ assuming $Z_j$s for $j \in \mathcal{H}$ are \iid normally distributed (Figure \ref{fig:W}). We can see that for $j \in \mathcal{H}$, the filter statistics $W_j$ is symmetric around 0, whereas for $j \in \mathcal{S}$, $\PP{W_j>0}>1/2$.

\begin{figure}
    \centering
    \includegraphics[scale=0.9]{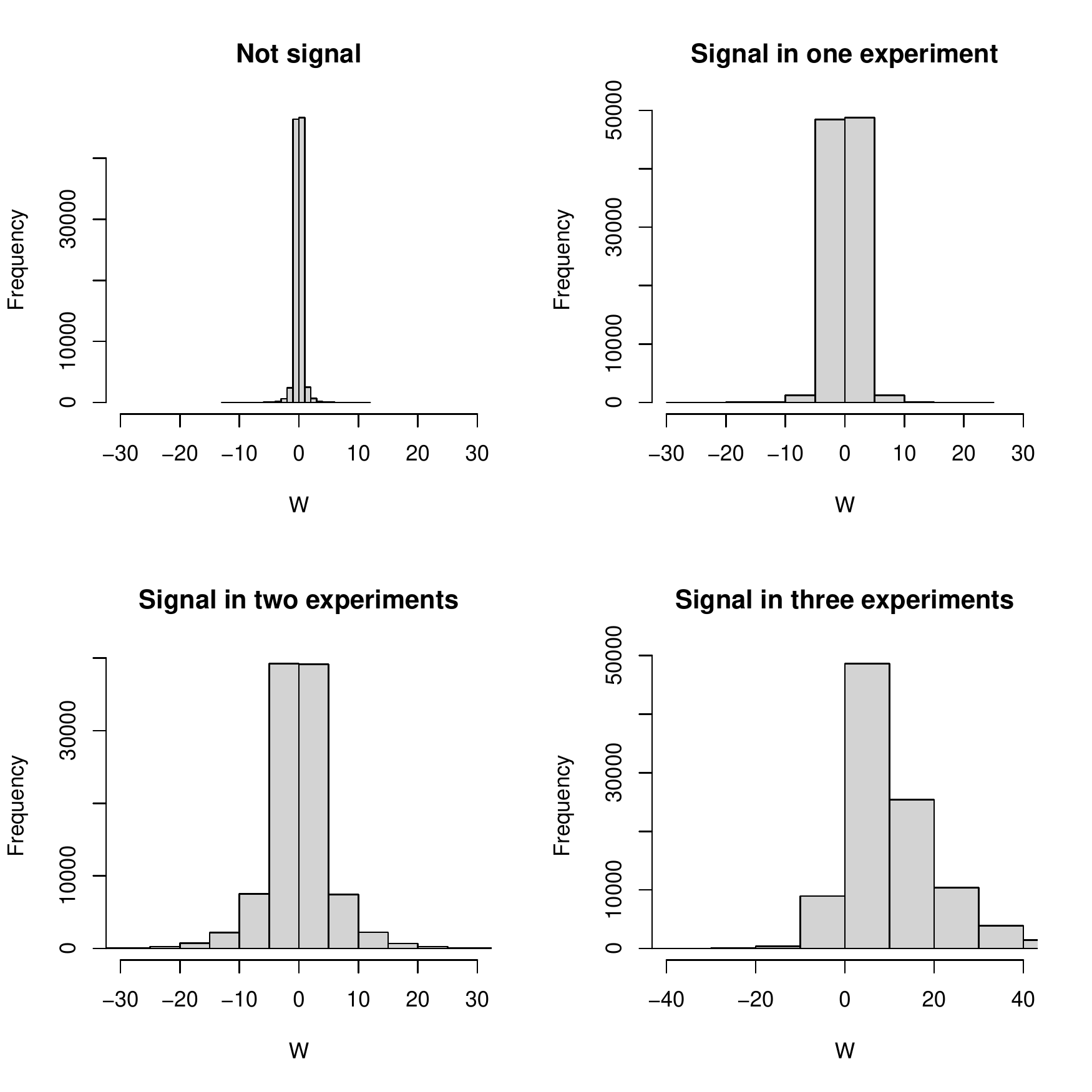}
    \caption{Distributions of the filter statistics $W_j = Z_j - \widetilde{Z}_j$, where $Z_j$ and $\widetilde{Z}_j$ are as defined in \eqref{eqn:one-swap} with $K=3$ for the cases where feature $j$ is not a signal in any of the experiments (null), $j$ is a signal in one experiment (null), two experiments (null) and three experiments (alternative).}
    \label{fig:W}
\end{figure}

\section{Simulation} 
\label{sec:simulation}
\paragraph{Simulation settings} We first consider the $K=2$ case with three data settings: 
\begin{enumerate}
    \item \textit{Continuous.} For both experiments, $Y_i^k$s are continuous and $Y_i^k|\X_i^k$s follow linear models.
    \item \textit{Binary.} For both experiments, $Y_i^k$s are binary and $Y_i^k|\X_i^k$s follow logistic models.
    \item \textit{Mixed.} $Y_i^1$ is continuous and $Y_i^1|\X_i^1$ follows a linear model; $Y_i^2$ is binary and $Y_i^2|\X_i^2$ follows a probit model.
\end{enumerate}
 We compare our proposed method (\textit{simultaneous}) with the two alternative methods below:
\begin{itemize}
    \item \textit{Pooling.} Data are pooled together and tests  of the conditional associations are performed using the knockoff methods for a single experiment.
    \item \textit{Intersection.} Knockoff methods for single experiments are used for the individual experiments and the intersection of the selected sets is used to construct the selection set of the mutual signals.
\end{itemize}

We first study the effect of the signal sparsity level of the mutual signals and the non-mutual signals. We use $s_0$ to denote the number of simultaneous signals among the $K$ experiments, and $s_k$ to denote the number of the signals that are only present in the $k$-th experiment. We study the three cases: 1. $s_1=s_2=0$; 2. $s_1=0, s_2\neq 0$; 3. $s_1=s_2\neq 0$. Next, we study the effect of the correlations among the covariates. Third, we study the effect of the difference in signal strengths between the two experiments. We consider two scenarios for the signal strengths: Scenario 1. the directions and the strengths of the mutual signals are identical among the $K$ experiments; Scenario 2. only the directions of the mutual signals are the same but the signal strengths are independent among the $K$ experiments. Data generation and algorithm implementation details can be found in Web Appendix C. Additional simulations for the $K=3$ case, the power comparison among different choices of the filter statistics $\W$, and the empirical distributions of $\W$ to show why the method has power are also provided in Web Appendix C. 

\paragraph{Results}

 Figure \ref{fig:s1s2} shows the power and the FDR for the three methods (\textit{simultaneous, pooling, and intersection}) on the three data settings (continuous, binary, and mixed) when we vary $s_1=s_2$. As $s_1=s_2$ increases, only the \textit{simultaneous} method controls the FDR. The \textit{simultaneous} method has slightly lower power than the \textit{pooling} method, and the power gap is still moderate when the signals in the two experiments have different strengths (Scenario 2, right panel). The simulation results are in agreement with our theoretical expectations. First, in terms of FDR control, the \textit{simultaneous} method we proposed always controls the FDR across all our designed settings. The \textit{pooling} method only controls the FDR when all samples from the two experiments are i.i.d. The \textit{intersection} method controls the FDR when $s_2=0$ but it fails when $s_1=s_2\neq 0$. In terms of power, there is some gap between the \textit{simultaneous} and the \textit{pooling} methods, because the tests of union null hypotheses are more stringent. However, the gap is moderate. More detailed simulation results can be found in the Web Appendix C. The simulation results for the $K=3$ case are similar to the $K=2$ case and are consistent with our theoretical expectations (see Figure C6 in Web Appendix C). The \textit{simultaneous} method controls the FDR and has good power. The \textit{pooling} method has high power but also has a very high FDR rate when there are signals that are shown in either one or two of the samples only. The \textit{intersection} method has similar power to the \textit{simultaneous} method, but it cannot control the FDR when a large number of features are signals in only two of the three samples.

The comparison among different $\W$ statistics suggests that the \textit{Max} and \textit{Diff} (see definitions in Web Appendix C.1.3) $\W$s have the best performance among the $\W$s we have explored. More simulation results can be found in Web Appendix C.

\begin{figure}
    \centering
    \includegraphics[scale=0.35]{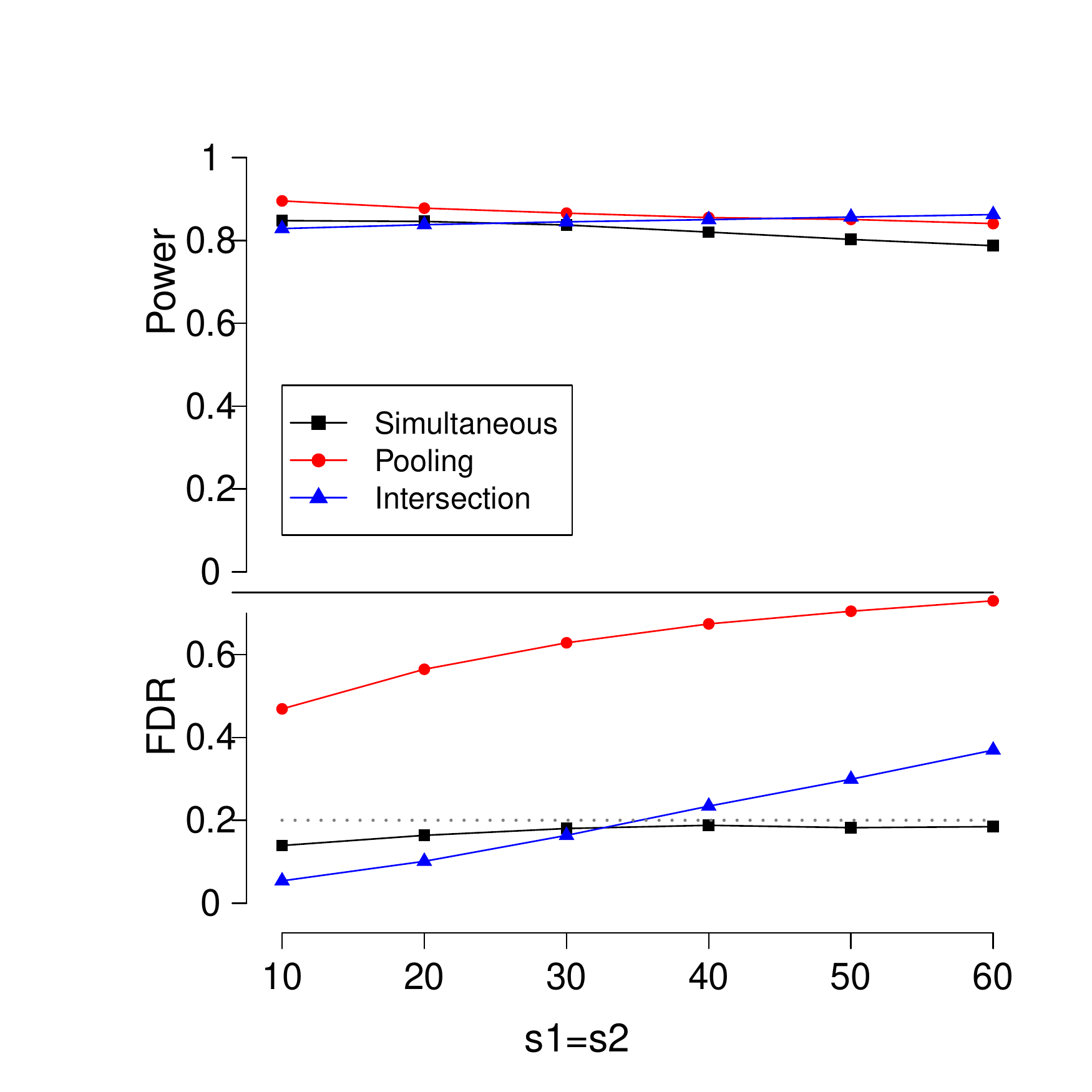}
    \includegraphics[scale=0.35]{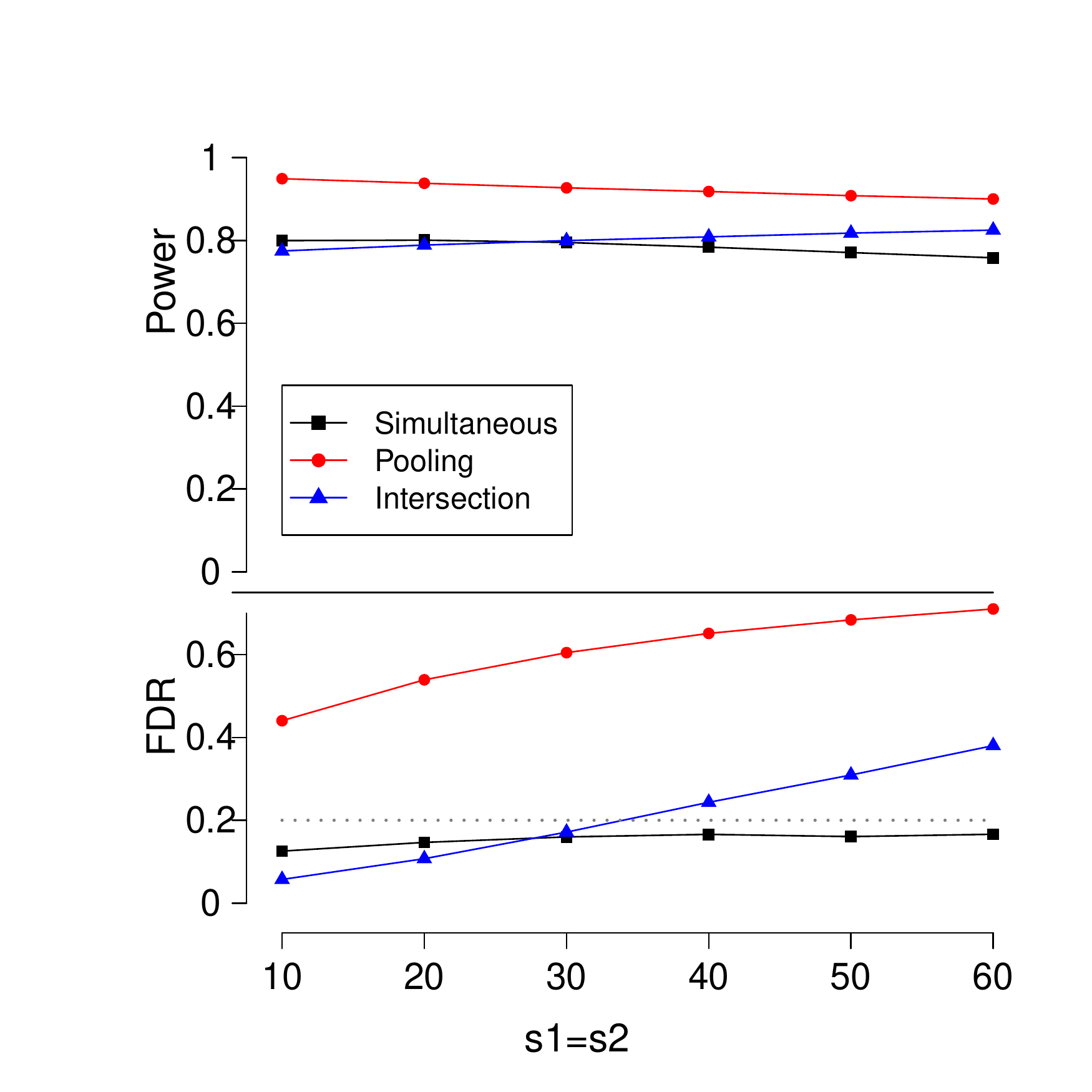}\\
    \includegraphics[scale=0.35]{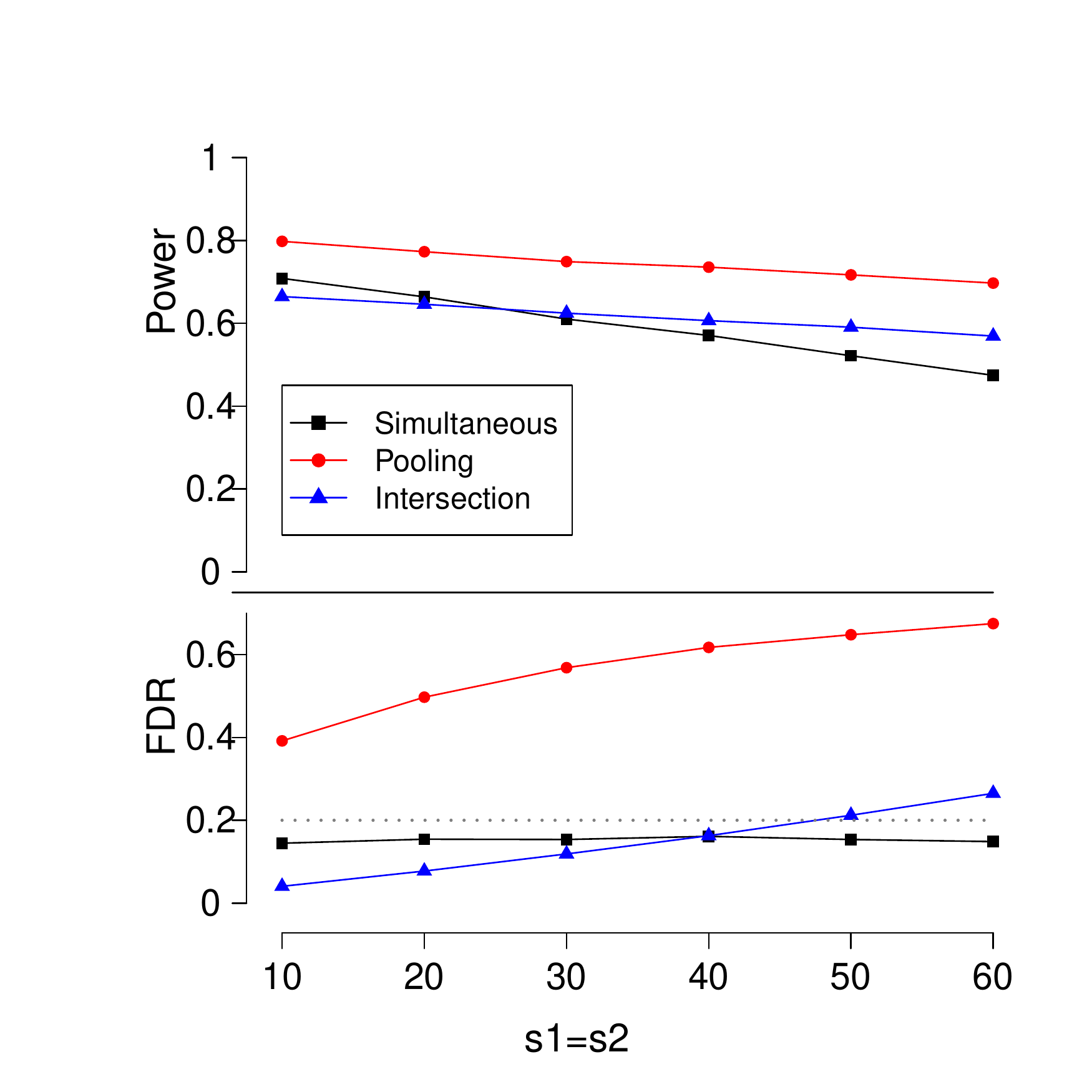}
    \includegraphics[scale=0.35]{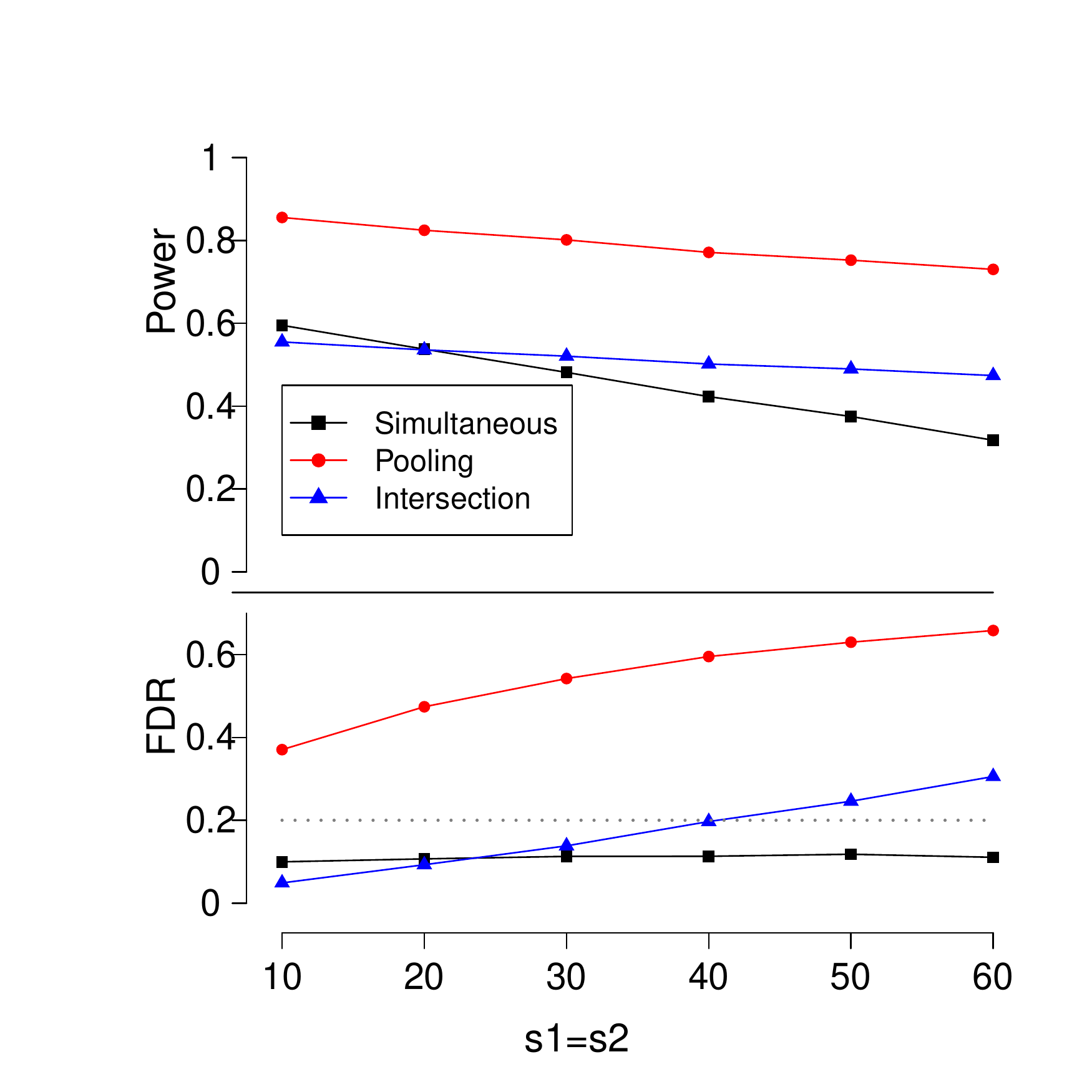}\\
    \includegraphics[scale=0.35]{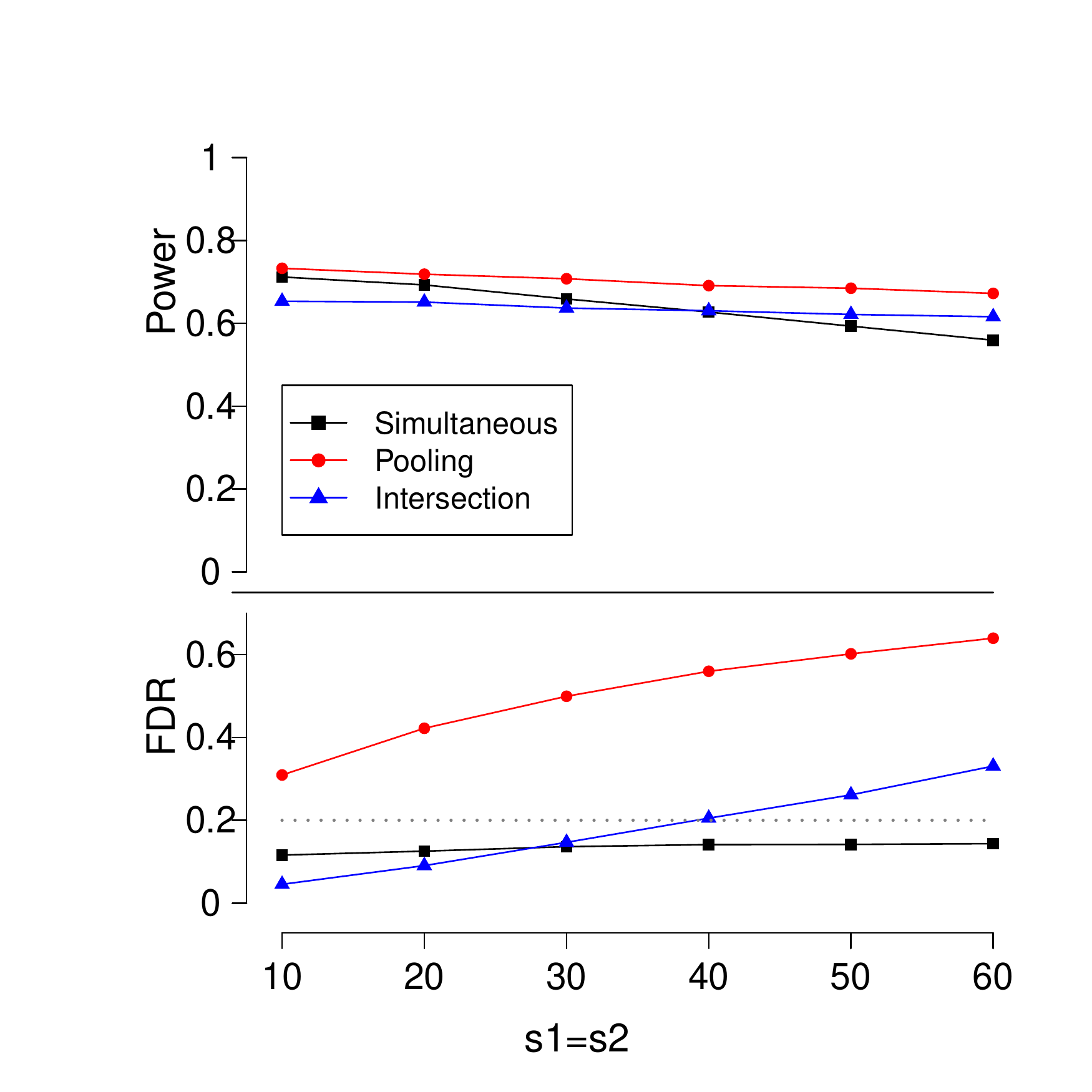}
    \includegraphics[scale=0.35]{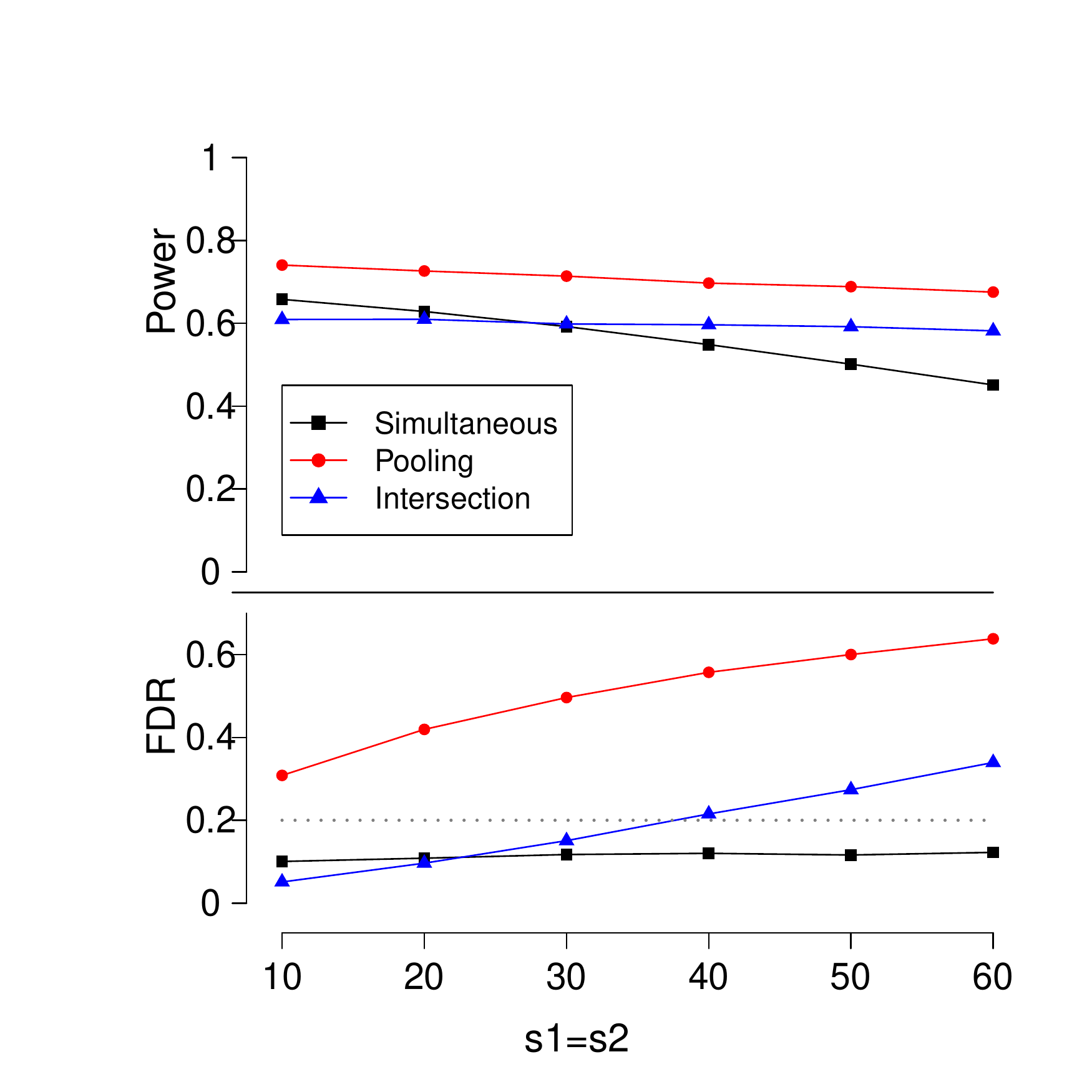}
    \caption{The power and the FDR for the continuous (upper), binary (middle) and mixed (lower) settings when varying $s_1=s_2$. Results for Scenario 1 are on the left and for Scenario 2 are on the right. A colored version of the figures can be found in the electronic version of the article.}
    \label{fig:s1s2}
\end{figure}

\section{Real data analysis}\label{sec:application}
In this section, we demonstrate the application of the \textit{Simultaneous knockoff} method on two real data examples. For the first data example, we use the Fixed-X knockoffs with linear models for the individual experiments; and for the second data example, we use the Model-X knockoffs with a penalized Cox regression model for each gene expression experiment.

\subsection{Application to the Communities and Crime data}

In a crime rate study, we aim to identify features that are universally associated with the community crime rate, regardless of race distribution in the community. This is potentially useful in guiding unbiased policy-making based on race-blinded findings. To achieve that, we select features that are simultaneously associated with the crime rate in different race distribution groups. 

We use the publicly available Communities and Crime data set from the University of California Irvine (UCI) machine learning repository. The data set contains crime information on $n=1994$ communities with different race distributions in the US. For the individual communities, it has information on the crime rate, as well as $122$ other variables that are potentially related to the crime rate. All continuous variables are normalized to the 0-1 range. Our primary outcome of interest is the normalized crime rate, our feature candidates are the $p=95$ features with no missing values that are not directly defined by race. We split the data into two subsets with approximately equal sample sizes based on the proportion of the Caucasian population (high/low) within the community. We fit a linear regression model to each subset of the data, aiming to identify mutual signals from both models with FDR control. We compare the three variable selection procedures (\textit{simultaneous}, \textit{pooling}, and \textit{intersection}) using the knockoff filter \eqref{eqn:knockoff}. We also compare our method with the \textit{repFDR} \citep{heller2014}. The \textit{repFDR} is developed for replicability studies, which requires that under the null the Z-scores are normally distributed. More details can be found in Web Appendix D.1.

Table \ref{tab:res_con} shows the results of identified features from different algorithms with a targeted FDR level of $q=0.1$. Our proposed \textit{simultaneous} method selected the following variables: ``the percentage of households with public assistance income in 1989", ``the percentage of kids born to never married", ``the percent of persons in dense housing". The \textit{pooling} method selected ``the percentage of kids born to never married", ``the percent of persons in dense housing", and ``the number of vacant households". The \textit{intersection} method selected ``the percentage of kids born to never married", and ``the percent of persons in dense housing". The \textit{repFDR} method selected ``the percentage of males who have never married".

\begin{table} 
\begin{center}
\begin{tabular}{c|ccc}
\hline
Analysis & Method & \# of features selected & \# of fake features selected\\
\hline
\multirow{4}{*}{Main}&Simultaneous& 3 & NA\\
         &pooling&3& NA\\
         &intersection&2& NA\\
         &repFDR & 1& NA\\
\hline
\multirow{4}{*}{Sensitivity}&Simultaneous& 3 & 0 \\
         &pooling&3&0 \\
         &intersection&1&0\\
         &repFDR &0&0\\
\hline
\end{tabular}
\end{center}
\vspace{0.2in}
\caption{Feature selection results for the primary analysis of the community crime data and the sensitivity analysis of the community crime data with added fake features from permutations. NA indicates not applicable.}
\label{tab:res_con}
\end{table}

To verify the robustness of our proposed method, we also added a set of $95$ fake features by permutation. The feature selection results are shown in Table \ref{tab:res_con} (Sensitivity). The variable selections with the \textit{simultaneous} method is relatively stable.

\subsection{Application to the TCGA data}

In this section, we demonstrate the usage of the \textit{Simultaneous knockoffs} to identify gene expressions that are associated with glioblastoma multiforme (GBM) for both male and female sub-populations. GBM is known as a hallmark of the malignant process, however, the molecular mechanisms that dictate the locally invasive progression remain an active research area. In this example, we use male and female sub-populations to demonstrate variable selection using our proposed \textit{simultaneous} method to identify mutual signals from heterogeneous data sets. In real applications, the sub-populations can be much more complicated (i.e. from different sources, collected at different places and times, and with different technologies). Therefore the fact that the \textit{simultaneous} method does not require the data to be pooled makes it useful. 

Our GBM gene expression data are from The Cancer Genome Atlas (TCGA). The data contains 501 subjects with the overall survival outcome (in days) and 17813 level-3 gene-level expression data. There are $71$ censored and $430$ death cases. We use the sure independence screening (SIS, see \cite{fan2008}) for a marginal screening, leaving $d=\lfloor n/log(n) \rfloor=79$ genes with the smallest p-values. The SIS method allows dimension reduction from exponentially growing $p$ to a relatively large scale $d <n$, while the reduced model still contains all the true signals with high probability. It has been widely used in other studies \citep{zhang2021, luo2020}. We apply the \textit{Simultaneous knockoffs} to identify genes associated with the survival time within both the male and female GBM patient cohorts. We also compare our method with the methods \textit{pooling}, \textit{intersection} and \textit{repFDR}. We perform sensitivity analysis by relaxing the screening procedure to include all genes with p-values smaller than 0.0002, which leads to $111$ candidate genes after the pre-screening step. For missing data, a complete case analysis was performed for the main analysis, while a single imputation was performed in the sensitivity analysis. 

With the \textit{Simultaneous knockoff} method, three genes (EID3, RNPS1, and VPS72) are selected. All these genes have been frequently studied for their relationships with cancer, including GBM \citep{kunadis2021,goyal2021diagnostic,heiland2016}. The \textit{pooling} method selected two genes (CROCC and FMR1NB), and the \textit{intersection} method selected none. The \textit{repFDR} method selected 2 genes (MAP2K4 and ZNF239). All the three genes selected by \textit{Simultaneous knockoff}s are also selected when using the threshold $p<0.0002$ for pre-screening, although one additional gene FMR1NB is also selected under the relaxed screening scenario. Sensitivity analysis also shows that EID3 and  VPS72 are selected when we use single imputation to treat the missing data.

\section{Discussion}
\label{sec:discussion}

The \textit{Simultaneous knockoff} method is a general framework for testing the union null hypotheses on conditional associations between candidate features and outcomes. It can work with very general conditional models and covariate structures, assuming the $K$ experiments are independent. This method provides opportunities to combine information from  the experiments with heterogeneous $\X$ structures, different dependencies of $Y|\X$, and different outcomes $Y$. The FDR control guarantee is exact for finite sample settings. 

This method has even broader applications beyond our motivating examples. For example, when working with the electronic health record (EHR) data from multiple data centers, some outcome variables and covariates are recorded differently among the centers (for example, for obesity, some centers record the body mass index (BMI) of patients, but others use yes/no); and the demographic distributions are different. The \textit{Simultaneous knockoff} method can be used to identify mutual signals to confirm the associations. This method also only requires very limited information (only the test statistics) to be shared among the data centers, which benefits data collaboration under privacy protections.

 One big limitation of the current method is that in practice it is hard to work with ultra high-dimensional data due to the limits of the computer memory for the knockoff construction. We use the SIS pre-screening step in our real data example to circumvent this problem. Although theoretically, the \textit{Simultaneous knockoff} method does not require the number of variables to be smaller than the number of observations when using the Model-X knockoff construction, the efficient construction of the knockoffs for ultra-high-dimensional features is still challenging and worth further research. Another limitation of the work is the lack of a theoretical analysis of power. This problem is difficult in general and the power of the Model-X knockoff method has just been studied \citep{wang2021power} recently. We expect the power of the \textit{Simultaneous knockoff} method to decrease monotonically as $K$ and $p$ increase. The exact power change with the growth of $n$, $p$, and $K$ is still a challenging open question.

There are some extensions of the knockoff methods to work with group-wise variable selection \citep{chen2020,dai2016} where we are interested in testing whether each specific group of variables is associated with the outcome conditioning on other groups of variables. The current version of \textit{Simultaneous knockoff} method focuses on the selection of individual features. The extension to work on the selection of groups of features is worth future explorations 

There are many open questions in multiple testing that are related to the hypothesis testing for the union null hypotheses. Although our \textit{Simultaneous knockoff} method provides solutions to the reproducibility of studies, feature selections across heterogeneous populations, and mediation analysis, there are still more challenges from the real applications. For example, we can further explore methods that will allow combining the information from different data sets with unidentified overlapping samples (like case-cohort study).

\section*{Acknowledgements}
This research is partly supported by the National Cancer Institute under grant R01 CA119171 and by the National Institute of General Medical Sciences under grant U54 GM115458.

\section*{Data Availability Statement}
The data that support the findings in this paper are available at \url{https://archive.ics.uci.edu/ml/datasets/communities+and+crime} (The Community and Crime data) and at \url{http://www.liuzlab.org/TCGA2STAT/#package-archive} (TCGA data) \cite{wan2015}.

  \bibliographystyle{apa}
 \bibliography{bib}

\appendix

\section*{Web Appendix A: Details for the choices of the functions in Algorithm 2.2.2}
\subsection*{Appendix A.1: Knockoff assumptions and construction methods}
\paragraph{The Fixed-X Knockoff assumptions and construction methods}
The Fixed-X knockoff method only works for the continuous outcomes that are normally distributed as it requires the decentralized linear model, $\Y =\X\bmbeta + \bmvarepsilon, ~\text{where $\Y \in \R^{n}$, $\X \in \R^{n\times p}$, $\bmvarepsilon \in \R^{n} \sim \mathcal{N}(0,\sigma^2 I_n)$}$, and $\bmbeta \in \R^{p}$. In addition, the Fixed-X knockoff method requires $n\geq 2p$. 

The constructed Fixed-X knockoff features $\Xp$ need to satisfy that, for some vector $\bms\geq 0$,
\begin{equation}
\Xp^\top \Xp = \X^\top \X, \quad \Xp^\top \X = \X^\top \X - \diag\{\bms\}.
\end{equation} 
The knockoffs $\Xp$ can be computed using an efficient semidefinite programming (SDP) algorithm with or without randomization \citep{barber2015}. Since $\X$ is treated as a fixed design, so this approach allows any type of $\X$, either continuous, categorical, count, or mixed (i.e., different types for different columns of $\X$). For categorical variables, dummy variables will be created and used in $\X$. One thing we would like to point out is that based on our current notation, it will perform a test for each dummy variable separately and compute the FDR by treating different dummy variables as multiple features. However, it will be scientifically more interesting to consider controlling the group FDR as in \citet{dai2016} so that the test will not depend on which reference level we choose when creating dummy variables. We use the R function \textit{create.fixed} within the R package \textit{knockoff} to implement this construction method. 

\paragraph{The Model-X Knockoff assumptions and construction methods}
The general sampling method as described in \citet{candes2018} can be applied to all kinds of data types (continuous, categorical, count, or mixed with or without missing data) as long as the joint distribution of $\X$ is known or can be estimated. However, for the implementation, their current R package only allows for continuous $\X$, so in simulation, we only consider the Model-X Knockoff construction methods for continuous $\X$. Specifically, two Model-X Knockoff construction methods are reviewed below:
\begin{itemize}
\item Gaussian: When the distribution of $\X$ is assumed to be Gaussian with the variance-covariance matrix $\bmSigma$, we can sample $\widetilde{\X}$ from $\widetilde{\X}|\X\sim N(\bmmu, \mathbf{V})$, where $\bmmu$ and $\mathbf{V}$ are given by,
\begin{eqnarray*}
\bmmu&=&\X-\X\bmSigma^{-1}\diag\{\bms\},\\
\mathbf{V}&=&2\diag\{\bms\}-\diag\{\bms\}\bmSigma^{-1}\diag\{\bms\}.
\end{eqnarray*}
We use the R function \textit{create.gaussian} within the R package \textit{knockoff} to implement this construction method.
\item Second order: The second-order Model-X knockoff construction method tries to sample $\widetilde{\X}$ such that 
\begin{equation*}
Cov([\X,\widetilde{\X}]) = \left(\begin{array}{cc}\bmSigma&\bmSigma-\diag\{\bms\}\\\bmSigma-\diag\{\bms\}&\bmSigma\end{array}\right).\end{equation*} 
where the requirements of $\bms$ are the same as for the Fixed-X knockoff and can be solved using the approximate semidefinite program (ASDP) algorithm as given in \citep{candes2018}. We use the R function \textit{create.second} within the R package \textit{knockoff} to implement this construction method.
\end{itemize}
The knockoffs $\Xp$ can also be generated using various advanced algorithms \citep{romano2019,liu2019,bates2020,spector2020}. However, we can only find the available implementations of these methods for continuous variables.

\subsection*{Appendix A.2: Statistics compatible with each knockoff construction method}
The statistics described in this section are for individual hypothesis testing for simplicity of notation. It can be extended to statistics for the group hypothesis testing similar to those used in \citet{dai2016}.
\subsection*{Appendix A.2.1: Statistics compatible with the Fixed-X knockoff construction method}
\begin{itemize}
\item Order of selection (Lasso): We can choose the lasso variable selection procedure and construct the statistics as $\widehat{\bmbeta}(\lambda)$, where 
\[\widehat{\bmbeta}(\lambda) = \arg\min_{\bmb\in\R^{2p}} \left\{\norm{\Y - [ \X \ \Xp ] \bmb}_2^2 +\lambda \norm{\bmb}_1\right\}.\]
We can run over a range of $\lambda$ values decreasing from $+\infty$ (a fully sparse model) to $0$ (a fully dense model) and define $Z_j$ as the maximum $\lambda$ such that $\widehat{\beta}_j(\lambda)\neq 0$. If there is no $\lambda$ such that $\widehat{\bmbeta}_j(\lambda)\neq 0$, then we will simply define $Z_j$ as 0.
\item Absolute coefficient (Lasso): We can use $|\widehat{\beta}_j(\lambda)|$ as defined above with a specific $\lambda$ value. 
\end{itemize}

\subsection*{Appendix A.2.2: Statistics compatible with the Model-X knockoff construction method}
For the Model-X knockoffs, very general conditional models (such as generalized linear models or nonlinear models) can be used. In addition to the statistics listed in Appendix A.2.1, we can also use the following statistics.

\begin{itemize}
\item Absolute coefficient (glmnet): We can use $|\widehat{\beta}_j(\lambda)|$ from the penalized generalized linear regression or the penalized Cox regression model of $Y$ on $[\X\  \widetilde{\X}]$ with either a specific $\lambda$ value or a $\widehat{\lambda}$ estimated from cross-validation. 
\item Standardized coefficient (glmnet): We can also use the standardized regression coefficients $|\widehat{\beta}_j(\lambda)|/\widehat{\text{SE}}(\widehat{\beta}_j(\lambda))$ from the penalized generalized linear regression or the penalized Cox regression model of $Y$ on $[\X\   \widetilde{\X}]$ with either a specific $\lambda$ value or a $\widehat{\lambda}$ estimated from cross-validation.
\item Order of selection (glmnet): We can use the minimum $\lambda$ such that the regression coefficient becomes 0 from the penalized generalized linear regression or the penalized Cox regression model of $Y$ on $[\X\ \widetilde{\X}]$, or the reciprocal of the order of each variable to be included in the model when increasing the number of variables to be selected.
\item Variable importance factor: We can use the variable importance factors from the random forest fitting of $Y$ on $[\X\ \widetilde{\X}]$ with either fixed tuning parameters or tuning parameters selected from cross-validation. 
\end{itemize}

\subsection*{Appendix A.3: Examples of OSCFs}
It is obvious that the example given in section 2.2.1 can be generalized to a broader class of OSCFs as below:
\begin{align} \label{eqn:OSCF}
 & Z_j=h_j\left(\sum_{\bma\in A_e}g_j(\prod_{k=1}^K f_{jk}(Z_{jk})^{a_k}f_{jk}(\widetilde{Z}_{jk})^{1-a_k}), \psi_j([\Z,\widetilde{\Z}])\right) \quad \textnormal{and}\quad \nonumber \\
& \widetilde{Z}_j=h_j\left(\sum_{\bma\in A_o}g_j(\prod_{k=1}^K f_{jk}(Z_{jk})^{a_k}f_{jk}(\widetilde{Z}_{jk})^{1-a_k}),\psi_j([\Z^1,\widetilde{\Z}^1],\cdots,[\Z^K,\widetilde{\Z}^K])\right),    
\end{align}
with arbitrary functions $f_{jk}(\cdot),g_{j}(\cdot), h_{j}(\cdot,\cdot)$, and a symmetric function $\psi_j(\cdot)$ in the sense that
\begin{multline*}
    \psi_j([\Z^1,\widetilde{\Z}^1],\cdots,[\Z^K,\widetilde{\Z}^K])=\psi_j([\Z^1,\widetilde{\Z}^1]_{\textnormal{Swap}(S)},\cdots,[\Z^K,\widetilde{\Z}^K])\\=\cdots=\psi_j([\Z^1,\widetilde{\Z}^1],\cdots,[\Z^K,\widetilde{\Z}^K]_{\textnormal{Swap}(S)}), ~\text{for $j\in [p]$ and $k \in [K]$.}
\end{multline*}  Here $f_{jk}(\cdot)$s can be viewed as the transformations of test statistics for features within individual experiments, allowing us to add weights to the features based on some prior knowledge. $g_j(\cdot)$s depict how information from the different experiments are pooled together, which can be different for different features. $h_j(\cdot,\cdot)$s are some final transformations of the statistics, which can also be different for different features. Examples of $\psi_j(\cdot)$s include maximum functions such as \begin{align*}
   & \psi_j([\Z^1,\widetilde{\Z}^1],\cdots,[\Z^K,\widetilde{\Z}^K])=\max_{k=1}^K\max(Z_{jk},\widetilde{Z}_{jk}),~\text{and}\\ & \psi_j([\Z^1,\widetilde{\Z}^1],\cdots,[\Z^K,\widetilde{\Z}^K])=\max_{k=1}^K\max_{j=1}^p\max(Z_{jk},\widetilde{Z}_{jk});
\end{align*} and sum functions such as \begin{multline*}
    \psi_j([\Z^1,\widetilde{\Z}^1],\cdots,[\Z^K,\widetilde{\Z}^K])=\sum_{j=1}^p\sum_{k=1}^Kw_{jk}(Z_{jk}+\widetilde{Z}_{jk})~ \text{with arbitrary weights $w_{jk}$s.}
\end{multline*}

Notice that this is a non-exhaustive list and there are still OSCFs beyond the class \eqref{eqn:OSCF}.

\subsection*{Appendix A.4: Examples of flip sign functions}
Example of antisymmetric functions include difference function $f(x,y)=x-y$ and signed max function $f(x,y)=max(x,y)\cdot(-1)^{\mathbbm{1}\{x>y\}}$. Actually, it can be the product of any symmetric function $g(x,y)=g(y,x)$ and the sign function, i.e., $f(x,y)=g(x,y)\cdot(-1)^{\mathbbm{1}\{x>y\}}$.
The flip sign function can be constructed entry-wise using the antisymmetric function as introduced in \cite{candes2018}. Examples of flip sign functions constructed this way include the difference function \begin{equation}\label{eqn:diff_fun}
f([\Z,\widetilde{\Z}])=\W \in \R^p, ~\text{where}~ W_j = Z_j - \widetilde{Z}_j,\end{equation} and the signed max function
\[f([\Z,\widetilde{\Z}])=\W \in \R^p, ~\text{where}~W_j= (Z_j\vee \widetilde{Z}_j) \cdot (-1)^{\mathbbm{1}\{Z_j<\widetilde{Z}_j\}} .\] 
However, there can be other function forms allowing $W_j$ to depend on not only $Z_j$ and $\widetilde{Z}_j$ but also $Z_{-j}$ and $\widetilde{Z}_{-j}$, for example
$W_j=\sum_{i=1}^p(Z_i+\widetilde{Z}_i)(Z_j-\widetilde{Z}_j)$. A broader class can be constructed by noticing that the Hadamard product of a flip sign function and a (pairwise) symmetric function $g(x_1,\cdots,x_p,y_1,\cdots,y_p)=g(x_1,\cdots,y_k,x_{k+1},\cdots,x_p,y_1,\cdots,x_{k},y_{k+1},\cdots,y_p)$ for any $k\in [p]$ is still a flip sign function.

\subsection*{Appendix A.5: Other examples for $W$ construction}
We first show that the composite function of a flip sign function and an OSCF can lead to an OSFF.
\begin{lemma}\label{lem:0}
Let $[\Z,\widetilde{\Z}] = f_1([\Z^1,\widetilde{\Z}^1],\cdots,[\Z^K,\widetilde{\Z}^K])$ be an OSCF, and $\W = f_2([\Z,\widetilde{\Z}])$ be a flip sign function. Then we have
$\W = f([\Z^1,\widetilde{\Z}^1],\cdots,[\Z^K,\widetilde{\Z}^K]):=f_2\circ f_1 ([\Z^1,\widetilde{\Z}^1],\cdots,[\Z^K,\widetilde{\Z}^K])$ is an OSFF.
\end{lemma}
\begin{proof}
We verify this by checking the definition:
\begin{eqnarray*}
&&f([\Z^1,\widetilde{\Z}^1],\cdots,[\Z^k,\widetilde{\Z}^k]_{\textnormal{Swap}(S)},\cdots,[\Z^K,\widetilde{\Z}^K])\\
&=&f_2\circ f_1([\Z^1,\widetilde{\Z}^1],\cdots,[\Z^k,\widetilde{\Z}^k]_{\textnormal{Swap}(S)},\cdots,[\Z^K,\widetilde{\Z}^K])\\
&=&f_2 ( [f_1([\Z^1,\widetilde{\Z}^1],\cdots,[\Z^k,\widetilde{\Z}^k],\cdots,[\Z^K,\widetilde{\Z}^K])]_{\textnormal{Swap}(S)})\\
&=&f_2 ( f_1([\Z^1,\widetilde{\Z}^1],\cdots,[\Z^k,\widetilde{\Z}^k],\cdots,[\Z^K,\widetilde{\Z}^K]))\odot \bmepsilon(S)\\
&=&f([\Z^1,\widetilde{\Z}^1],\cdots,[\Z^K,\widetilde{\Z}^K])\odot \bmepsilon(S).
\end{eqnarray*}
\end{proof}

Besides using the composite function of a flip sign function and an OSCF, we can also create $W$ by combining the flip sign functions within each dataset. 
This basically allows us to first decide which flip sign function shall be used within each dataset before combining information together. The results are summarized in the following lemma.

\begin{lemma}\label{lem:01}
Let $\W^k=h_k([\Z^k,\widetilde{\Z}^k]), k\in [K]$ are $K$ flip sign functions and $g_j: \R^K\rightarrow \R, j\in [K]$ are $K$ combining function such that $g_j(x_1,\cdots,x_K)=-g_j(x_1,\cdots,-x_k,\cdots,x_K)$ for all $k\in [K]$. If we define 
$f: R^{2Kp}\rightarrow R^p$ constructed elementwise as
\begin{eqnarray*}
f([\Z^1,\widetilde{\Z}^1],\cdots,[\Z^K,\widetilde{\Z}^K])_j=g_j(h_1([\Z^1,\widetilde{\Z}^1])_j,\cdots, h_K([\Z^K,\widetilde{\Z}^K])_j),
\end{eqnarray*}
then $f$ is an OSFF.
\end{lemma}
\begin{proof}
To verify this, notice that for each element, we have
\begin{eqnarray*}
&&f([\Z^1,\widetilde{\Z}^1],\cdots,[\Z^k,\widetilde{\Z}^k]_{\textnormal{Swap}(S)},\cdots,[\Z^K,\widetilde{\Z}^K])_j\\
&=&g_j(h_1([\Z^1,\widetilde{\Z}^1])_j,\cdots,h_k([\Z^k,\widetilde{\Z}^k]_{\textnormal{Swap}(S)})_j,\cdots, h_K([\Z^K,\widetilde{\Z}^K])_j)\\
&=&g_j(h_1([\Z^1,\widetilde{\Z}^1])_j,\cdots,h_k([\Z^k,\widetilde{\Z}^k])_j\epsilon(S)_j,\cdots, h_K([\Z^K,\widetilde{\Z}^K])_j)\\
&=&\epsilon(S)_jg_j(h_1([\Z^1,\widetilde{\Z}^1])_j,\cdots,h_k([\Z^k,\widetilde{\Z}^k])_j,\cdots, h_K([\Z^K,\widetilde{\Z}^K])_j)\\
&=&\epsilon(S)_jf([\Z^1,\widetilde{\Z}^1],\cdots,[\Z^K,\widetilde{\Z}^K])_j.
\end{eqnarray*}
So we have 
\[f([\Z^1,\widetilde{\Z}^1],\cdots,[\Z^k,\widetilde{\Z}^k]_{\textnormal{Swap}(S)},\cdots,[\Z^K,\widetilde{\Z}^K])=f([\Z^1,\widetilde{\Z}^1],\cdots,[\Z^K,\widetilde{\Z}^K])\odot\bmepsilon(S),\]
which finishes the proof.
\end{proof}

Still, there are other ways to construct the OSFF function for calculating the filter statistics $\W$ beyond the two classes we discussed above and a few examples are given below:
\begin{itemize}
\item Direct Max Sum: $\W=\sum_{k=1}^K(\Z^k\vee \widetilde{\Z}^k)\cdot(-1)^{\mathbbm{1}\{\Z<\widetilde{\Z}\}}$.
\item Direct Max Max: $\W=\max_{k=1}^K(\Z^k\vee \widetilde{\Z}^k)\cdot(-1)^{\mathbbm{1}\{\Z<\widetilde{\Z}\}}$.
\item Direct Diff Sum: $\W=\sum_{k=1}^K|\Z^k-\widetilde{\Z}^k|\cdot(-1)^{\mathbbm{1}\{\Z<\widetilde{\Z}\}}$.
\item Direct Diff Max: $\W=\max_{k=1}^K|\Z^k-\widetilde{\Z}^k|\cdot(-1)^{\mathbbm{1}\{\Z<\widetilde{\Z}\}}$.
\item Direct Sum Sum: $\W=\sum_{k=1}^K(\Z^k+\widetilde{\Z}^k)\cdot(-1)^{\mathbbm{1}\{\Z<\widetilde{\Z}\}}$.
\item Direct Sum Max: $\W=\max_{k=1}^K|\Z^k+\widetilde{\Z}^k|\cdot(-1)^{\mathbbm{1}\{\Z<\widetilde{\Z}\}}$.
\end{itemize}

\renewcommand{\thelemma}{B\arabic{lemma}}
\section*{Web Appendix B: Additional proof}
\subsection*{B.1: Proof of Theorem 1}\label{pf:thm1}

The proof of Theorem 1 follows the proof idea in \cite{barber2015}.
Let $m=\#\{j:W_j\neq 0\}$ and assume without loss of generality that $|W_1|\geq |W_2|\geq\cdots\geq |W_m|>0$. Define p-values $p_j=1$ if $W_j<0$ and $p_j=1/2$ if $W_j>0$, then Lemma 1 implies that null $p$-values are $i.i.d.$ with $p_j\geq \text{Unif}[0,1]$ and are independent from nonnulls. 

We first show the result for the knockoff+ threshold. Define $V=\#\{j\leq \widehat{k}_{+}: p_j\leq 1/2, j\in \mathcal{H}\}$ and $R=\#\{j\leq \widehat{k}_{+}:p_j\leq 1/2\}$ where $\widehat{k}_{+}$ satisfy that $|W_{\widehat{k}_{+}}|=\tau_{+}$ where $\tau_{+}$ is defined in theorem 1, we have
\begin{eqnarray*}
&&\EE{\frac{V}{R\vee 1}}=\EE{\frac{V}{R\vee 1}\mathbbm{1}_{\widehat{k}_{+}>0}}\\
&=&\EE{\frac{\#\{j\leq \widehat{k}_{+}: p_j\leq 1/2, j\in \mathcal{H}\}}{1+\#\{j\leq \widehat{k}_{+}: p_j> 1/2, j\in \mathcal{H}\}}\left(\frac{1+\#\{j\leq \widehat{k}_{+}: p_j> 1/2, j\in \mathcal{H}\}}{\#\{j\leq \widehat{k}_{+}:p_j\leq 1/2\}\vee 1}\right)\mathbbm{1}_{\widehat{k}_{+}>0}}\\
&\leq&\EE{\frac{\#\{j\leq \widehat{k}_{+}: p_j\leq 1/2, j\in \mathcal{H}\}}{1+\#\{ j\leq \widehat{k}_{+}: p_j> 1/2, j\in \mathcal{H}\}}}q\leq q,
\end{eqnarray*}
where the first inequality holds by the definition of $\widehat{k}_{+}$ and the second inequality holds by Lemma \ref{lem:mar}.

Similarly, for the knockoff threshold, we have $V=\#\{j\leq \widehat{k}_0: p_j\leq 1/2, j\in \mathcal{H}\}$ and $R=\#\{j\leq \widehat{k}_0:p_j\leq 1/2\}$ where $\widehat{k}_0$ satisfies that $|W_{\widehat{k}_0}|=\tau$ where $\tau$ is defined as in theorem 1, then
\begin{eqnarray*}
&& \EE{\frac{V}{R+q^{-1}}}\\
&=&\EE{\frac{\#\{ j\leq \widehat{k}_0: p_j\leq 1/2, j\in \mathcal{H}\}}{1+\#\{j\leq \widehat{k}_0: p_j> 1/2, j\in \mathcal{H}\}}\left(\frac{1+\#\{ j\leq \widehat{k}_0: p_j> 1/2, j\in \mathcal{H}\}}{\#\{j\leq \widehat{k}_0:p_j\leq 1/2\}+q^{-1}}\right)\mathbbm{1}_{\widehat{k}_0>0}}\\
&\leq&\EE{\frac{\#\{j\leq \widehat{k}_0: p_j\leq 1/2, j\in \mathcal{H}\}}{1+\#\{ j\leq \widehat{k}_0: p_j> 1/2, j\in \mathcal{H}\}}}q\leq q,
\end{eqnarray*}
where the first inequality holds by the definition of $\widehat{k}_0$ and the second inequality holds by Lemma \ref{lem:mar}.

\subsection*{B.2: Additional proofs for Lemmas related to Theorem 1} 
Here we first give the proof of Lemma 1 below:
\begin{proof}

For any $S\subseteq \mathcal{H}$, we can write it as the union of $K$ subsets $S=\cup_{k=1}^K S_k$, where $S_k\subseteq \mathcal{H}_k$ for $k = 1,\cdots,K$, and $S_{k1}\cap S_{k2} = \emptyset$ for all $k_1 \neq k_2$. In particular, we can let $S_k=S\cap \mathcal{H}_k\cap (\cup_{j=1}^{k-1}\mathcal{H}_j)^c$. 
Since $S_k\subseteq \mathcal{H}_k$, for $k\in [K]$, any statistics $[\Z^k, \widetilde{\Z}^k]=w([\X^k,\widetilde{\X}^k],\Y^k)$, by the construction of knockoffs, satisfy 
$[\Z^k, \widetilde{\Z}^k] \eqd [\Z^k, \widetilde{\Z}^k]_{\textnormal{Swap}(S_k)} $. By the mutually independence between $[Z^1,\widetilde{Z}^1],\cdots, [Z^K,\widetilde{Z}^K]$, we have 

\[f\left([\Z^1,\widetilde{\Z}^1]_{\textnormal{Swap}(S_1)},\cdots,[\Z^K,\widetilde{\Z}^K]_{\textnormal{Swap}(S_K)}\right)\eqd f\left([\Z^1,\widetilde{\Z}^1],\cdots,[\Z^K,\widetilde{\Z}^K]\right).\] 

Using the definition of the OSFF, we have
\begin{eqnarray*}
&&f\left([\Z^1,\widetilde{\Z}^1]_{\textnormal{Swap}(S_1)},[\Z^2,\widetilde{\Z}^2]_{\textnormal{Swap}(S_2)},\cdots,[\Z^K,\widetilde{\Z}^K]_{\textnormal{Swap}(S_K)}\right)\\
&=&f\left([\Z^1,\widetilde{\Z}^1],[\Z^2,\widetilde{\Z}^2]_{\textnormal{Swap}(S_2)},\cdots,[\Z^K,\widetilde{\Z}^K]_{\textnormal{Swap}(S_K)}\right)\odot \bmepsilon(S_1)\\
&=&\cdots\\
&=&f\left([\Z^1,\widetilde{\Z}^1],[\Z^2,\widetilde{\Z}^2],\cdots,[\Z^K,\widetilde{\Z}^K]\right)\odot_{k=1}^K \bmepsilon(S_k)\\
&=&f\left([\Z^1,\widetilde{\Z}^1],[\Z^2,\widetilde{\Z}^2],\cdots,[\Z^K,\widetilde{\Z}^K]\right)\odot\bmepsilon(S).
\end{eqnarray*}

So we obtain \[\W=f\left([\Z^1,\widetilde{\Z}^1],\cdots,[\Z^K,\widetilde{\Z}^K]\right)\eqd f\left([\Z^1,\widetilde{\Z}^1],\cdots,[\Z^K,\widetilde{\Z}^K]\right)\odot\bmepsilon(S)=\W\odot\bmepsilon(S).\] for any $S\subseteq \mathcal{H}$. Therefore, by choosing $S$ as the set $\{j:\epsilon_j=-1\}$, we have \[(W_1, \cdots, W_p) \eqd (W_1\cdot \epsilon_1, \cdots, W_p\cdot\epsilon_p).\] and thus we finish the proof of the lemma.
\end{proof}

\begin{lemma}\label{lem:mar}
For $k=m,m-1,\cdots,1,0$, put $V^{+}(k)=\#\{j:1\leq j\leq k, p_j\leq 1/2, j\in \mathcal{H}\}$ and $V^{-}(k)=\#\{ j:1\leq j\leq k, p_j> 1/2, j\in \mathcal{H}\}$ with the convention that $V^{\pm}(0)=0$. Let $\mathcal{F}_k$ be the filtration defined by knowing all the non-null $p$-values, as well as $V^{\pm}(k')$ for all $k'\geq k$. Then the process $M(k)=\frac{V^{+}(k)}{1+V^{-}(k)}$ is a super-martingale running backward in time with respect to $\mathcal{F}_k$. For any fixed $q$, $\widehat{k}=\widehat{k}_{+}$ or $\widehat{k}=\widehat{k}_{0}$ as defined in the proof of theorem 1 are stopping times, and as consequences
\begin{eqnarray*}
\EE{\frac{\#\{j\leq \widehat{k}:p_j\leq 1/2, j\in \mathcal{H}\}}{1+\#\{j\leq \widehat{k}:p_j>1/2, j\in \mathcal{H}\}}}\leq 1
\end{eqnarray*}
\end{lemma}

\begin{proof}
The filtration $\mathcal{F}_k$ contains the information of whether $k$ is null and non-null process is known exactly. If $k$ is non-null, then $M(k-1)=M(k)$ and if $k$ is null, we have
\begin{eqnarray*}
M(k-1)=\frac{V^{+}(k)-\mathbbm{1}_{p_k\leq 1/2}}{1+V^{-}(k)-(1-\mathbbm{1}_{p_k\leq 1/2})}=\frac{V^{+}(k)-\mathbbm{1}_{p_k\leq 1/2}}{\left(V^{-}(k)+\mathbbm{1}_{p_k\leq 1/2}\right)\vee 1}
\end{eqnarray*}
Given that nulls are \iid, we have
\begin{eqnarray*}
\PP{\mathbbm{1}_{p_k\leq 1/2}|\mathcal{F}_k}=\frac{V^{+}(k)}{\left(V^{+}(k)+V^{-}(k)\right)}.
\end{eqnarray*}
So when $k$ is null, we have
\begin{eqnarray*}
\EE{M(k-1)|\mathcal{F}_k}&=&\frac{1}{V^{+}(k)+V^{-}(k)}\left[V^{+}(k)\frac{V^{+}(k)-1}{V^{-}(k)+1}+V^{-}(k)\frac{V^{+}(k)}{V^{-}(k)\vee 1}\right]\\
&=&\frac{V^{+}(k)}{1+V^{-}(k)}\mathbbm{1}_{V^{-}(k)>0}+(V^{+}(k)-1)\mathbbm{1}_{V^{-}(k)=0}\\
&\leq &M(k)
\end{eqnarray*}
This finish the proof of super-martingale property. $\widehat{k}$ is a stopping time with respect to $\{\mathcal{F}_k\}$ since $\{\widehat{k}\geq k\}\in \mathcal{F}_k$. So we have $\EE{M(\widehat{k})}\leq \EE{M(m)}=\EE{\frac{\#\{j:p_j\leq 1/2, j\in \mathcal{H}\}}{1+\#\{j:p_j>1/2, j\in \mathcal{H}\}}}$.

Let $X=\#\{j:p_j\leq 1/2, j\in \mathcal{H}\}$, given that $p_j\geq Unif[0,1]$ independently for all nulls, we have $X\leq_d Binomial(N,1/2)$. Let $Y\sim Binomial(N,1/2)$ where $N$ is the total number of nulls. Given that $f(x)=\frac{x}{1+N-x}$ is non-decreasing, we have
\begin{eqnarray*}
\EE{\frac{X}{1+N-X}}&\leq &\EE{\frac{Y}{1+N-Y}}\\
&=&\sum_{i=1}^N (1/2)^{N}\frac{N!}{i!(N-i)!}\frac{i}{1+N-i}\\
&=&\sum_{i=1}^N \PP{Y=i-1}\\
&\leq &1.
\end{eqnarray*}
\end{proof}

\subsection*{B.3: Proof of Theorem 2}\label{pf:thm2}
The proof of Theorem 2 follows the proof idea in \cite{barber2020}. Define 
\begin{equation}
R_{\epsilon}(t)=\frac{\sum_{j\in \mathcal{H}} \mathbbm{1}\{W_j\geq t, \min_{k: j\in \mathcal{H}_{k}}\widehat{\text{KL}}_{kj}\leq \epsilon\}}{1+\sum_{j\in \mathcal{H}} \mathbbm{1}\{W_j\leq -t\}}.
\end{equation}
Then for the knockoff+ filter with threshold $\tau_{+}$, we have
\begin{eqnarray*}
&&\frac{|\{j:j\in \widehat{S}\cap \mathcal{H}~\textnormal{and}~\min_{k: j\in \mathcal{H}_{k}}\widehat{\textnormal{KL}}_{kj}\leq \epsilon\}|}{|\widehat{S}|\vee 1} = \frac{\sum_{j\in \mathcal{H}} I\{W_j\geq \tau_{+}, \min_{k: j\in \mathcal{H}_{k}}\widehat{\text{KL}}_{kj}\leq \epsilon\}}{\sum_{j} \mathbbm{1}\{W_j\geq \tau_{+}\}\vee 1}\\
&=&\frac{{1+\sum_{j} \mathbbm{1}\{W_j\leq -\tau_{+}\}}}{\sum_{j} \mathbbm{1}\{W_j\geq \tau_{+}\}\vee 1}\cdot \frac{\sum_{j\in \mathcal{H}} \mathbbm{1}\{W_j\geq \tau_{+}, \min_{k: j\in \mathcal{H}_{k}}\widehat{\text{KL}}_{kj}\leq \epsilon\}}{1+\sum_{j} \mathbbm{1}\{W_j\leq -\tau_{+}\}}\\
&\leq & q \cdot R_{\epsilon}(\tau_{+})
\end{eqnarray*}
The inequality holds by the construction of knockoff+ under the FDR level of $q$ and the fact $\sum_{j} \mathbbm{1}\{W_j\leq -\tau_{+}\}\geq \sum_{j\in \mathcal{H}} \mathbbm{1}\{W_j\leq -\tau_{+}\}$.

For the knockoff filter at threshold $\tau$, we have similar results for the modified FDR as below:
\begin{eqnarray*}
&&\frac{|\{j:j\in \widehat{S}\cap \mathcal{H}~\textnormal{and}~\min_{k: j\in H_{k}}\widehat{\textnormal{KL}}_{kj}\leq \epsilon\}|}{|\widehat{S}|+q^{-1}} = \frac{\sum_{j\in \mathcal{H}} \mathbbm{1}\{W_j\geq \tau, \min_{k: j\in \mathcal{H}_{k}}\widehat{\text{KL}}_{kj}\leq \epsilon\}}{\sum_{j} \mathbbm{1}\{W_j\geq \tau\}+q^{-1}}\\
&=&\frac{{1+\sum_{j} \mathbbm{1}\{W_j\leq -\tau\}}}{q^{-1}+\sum_{j} \mathbbm{1}\{W_j\geq \tau\}}\cdot \frac{\sum_{j\in \mathcal{H}_0} \mathbbm{1}\{W_j\geq \tau, \min_{k: j\in \mathcal{H}_{k}}\widehat{\text{KL}}_{kj}\leq \epsilon\}}{1+\sum_{j} \mathbbm{1}\{W_j\leq -\tau\}}\\
&\leq & q \cdot R_{\epsilon}(\tau)
\end{eqnarray*}

Next we will show that $\EE{R_{\epsilon}(T)}\leq e^{\epsilon}$ for both $T=\tau$ and $T=\tau_{+}$ to complete the proof.

For events $E_j=\min_{k: j\in \mathcal{H}_{k}}\widehat{\text{KL}}_{kj}$, we have by Lemma \ref{lem:KL},
\begin{equation}
P(W_j>0,E_j\leq \epsilon | |W_j|,W_{-j})\leq e^{\epsilon} P(W_j<0||W_j|,W_{-j}), \forall \epsilon\geq 0, j\in \mathcal{H}_0.
\end{equation}

So we have
\begin{eqnarray*}
\EE{R_{\epsilon}(T)}&=&\EE{\frac{\sum_{j\in \mathcal{H}} \mathbbm{1}\{W_j\geq T, \min_{k: j\in \mathcal{H}_{k}}\widehat{KL}_{kj}\leq \epsilon\}}{1+\sum_{j\in \mathcal{H}} \mathbbm{1}\{W_j\leq -T\}}}\\
&=&\sum_{j\in \mathcal{H}_0}\EE{\frac{\mathbbm{1}\{W_j\geq T_j, E_j\leq \epsilon\}}{1+\sum_{k\in\mathcal{H},k\neq j} \mathbbm{1}\{W_k\leq -T_j\}}}\\
&=&\sum_{j\in \mathcal{H}}\EE{\frac{\mathbbm{1}\{W_j>0,E_j\leq \epsilon\}\mathbbm{1}\{|W_j|\geq T_j\}}{1+\sum_{k\in\mathcal{H},k\neq j} \mathbbm{1}\{W_k\leq -T_j\}}}\\
&=&\sum_{j\in \mathcal{H}}\EE{\frac{P(W_j>0,E_j\leq \epsilon||W_j|,W_{-j})\mathbbm{1}\{|W_j|\geq T_j\}}{1+\sum_{k\in\mathcal{H},k\neq j} \mathbbm{1}\{W_k\leq -T_j\}}}\\
&\leq &e^{\epsilon}\sum_{j\in \mathcal{H}}\EE{\frac{P(W_j<0||W_j|,W_{-j})\mathbbm{1}\{|W_j|\geq T_j\}}{1+\sum_{k\in\mathcal{H},k\neq j} \mathbbm{1}\{W_k\leq -T_j\}}}\\
&=&e^{\epsilon}\sum_{j\in \mathcal{H}}\EE{\frac{\mathbbm{1}\{W_j<0\}\mathbbm{1}\{|W_j|\geq T_j\}}{1+\sum_{k\in\mathcal{H},k\neq j} \mathbbm{1}\{W_k\leq -T_j\}}}\\
&=&e^{\epsilon}\EE{\sum_{j\in \mathcal{H}}\frac{\mathbbm{1}\{W_j\leq -T_j\}}{1+\sum_{k\in\mathcal{H},k\neq j} \mathbbm{1}\{W_k\leq -T_j\}}}
\end{eqnarray*}
The last step is because if $W_j\geq -T_j$ for all $j$, then we have the summation within the expectation is 0, otherwise, the summation is 1 using Lemma 6 of \citet{barber2020}.
\begin{multline}
\sum_{j\in \mathcal{H}}\frac{\mathbbm{1}\{W_j\leq -T_j\}}{1+\sum_{k\in\mathcal{H},k\neq j} \mathbbm{1}\{W_k\leq -T_j\}}=\sum_{j\in \mathcal{H}}\frac{\mathbbm{1}\{W_j\leq -T_j\}}{1+\sum_{k\in\mathcal{H},k\neq j} \mathbbm{1}\{W_k\leq -T_k\}} \\ =\sum_{j\in \mathcal{H}}\frac{\mathbbm{1}\{W_j\leq -T_j\}}{\sum_{k\in\mathcal{H}} \mathbbm{1}\{W_k\leq -T_k\}}=1.
\end{multline}
This completes the proof of Theorem 2.

\subsection*{B.4: Additional proofs for lemmas related to Theorem 2}
\begin{lemma}\label{lem:KL}
For $E_j:=\min_{k:j\in \mathcal{H}_{k}}\widehat{\text{KL}}_{kj}$, we have 
\begin{equation}
P(W_j>0,E_j\leq \epsilon | |W_j|,W_{-j})\leq e^{\epsilon} P(W_j<0||W_j|,W_{-j}), \forall \epsilon\geq 0, j\in \mathcal{H}_0.
\end{equation}
\end{lemma}

\begin{proof}
For any $j\in\mathcal{H}_0$, find $k\in H_{0j}$ such that $\widehat{\text{KL}}_{kj}=\min_{k':j\in \mathcal{H}_{k'}}\widehat{\text{KL}}_{k'j}$. Define \[\X_{-(k,j)} = \{X^1_{1},\cdots,X^1_{p},\cdots,X^K_{1},\cdots,X^K_{p}\}\setminus \{X^k_{j}\}.\] We will be conditioning on $\X_{-(k,j)}$, $\widetilde{\X}_{-(k,j)}$, $Y$ and observing the unordered pair $\{X^k_{j},\widetilde{X}^k_{j}\} =\{ X_{j}^{k(0)},X_{j}^{k(1)}\}$ (we do not know which is which). Without loss of generality, let $W_j\geq 0$ when $X^k_{j}=X_{j}^{k(0)}$ and $\widetilde{X}^k_{j}=X_{j}^{k(1)}$. Then we have
\begin{eqnarray*}
&&\frac{P(W_j> 0|X_{j}^{k(0)},X_{j}^{k(1)},X_{-(k,j)},\widetilde{X}_{-(k,j)},\Y)}{P(W_j< 0|X_{j}^{k(0)},X_{j}^{k(1)},X_{-(k,j)},\widetilde{X}_{-(k,j)},\Y)}\\
&=&\frac{P((X^k_{j},\widetilde{X}^k_{j})=(X_{j}^{k(0)},X_{j}^{k(1)})|X_{j}^{k(0)},X_{j}^{k(1)},X_{-(k,j)},\widetilde{X}_{-(k,j)},\Y)}{P((X_{j}^k,\widetilde{X}_{j}^k)=(X_{j}^{k(1)},X_{j}^{k(0)})|X_{j}^{k(0)},X_{j}^{k(1)},X_{-(k,j)},\widetilde{X}_{-(k,j)},\Y)}\\
&=&\frac{P((X_{j}^k,\widetilde{X}_{j}^k)=(X_{j}^{k(0)},X_{j}^{k(1)})|X_{j}^{k(0)},X_{j}^{k(1)},X_{k,-j)},\widetilde{X}_{k,-j)},Y^k)}{P((X^k_{j},\widetilde{X}^k_{j})=(X_{j}^{k(1)},X_{j}^{k(0)})|X_{j}^{k(0)},X_{j}^{k(1)},X^k_{-j},\widetilde{X}^k_{-j)},Y_k)}\\ && \text{\quad because the $K$ groups are independent}\\
&=&\prod_{i=1}^{n_k}\frac{P((X_{ij}^k,\widetilde{X}_{ij}^k)=(X_{ij}^{k(0)},X_{ij}^{k(1)})|X_{ij}^{k(0)},X_{ij}^{k(1)},X^k_{i,-j},\widetilde{X}^k_{i,-j},Y_{i}^k)}{P((X_{ij}^k,\widetilde{X}_{ij}^k)=(X_{ij}^{k(0)},X_{ij}^{k(1)})|X_{ij}^{k(1)},X_{ij}^{k(0)},X_{i,-j}^k,\widetilde{X}^k_{i,-j},Y_{i}^k)}\\
&=&\prod_{i=1}^{n_k}\frac{P((X_{ij}^k,\widetilde{X}_{ij}^k)=(X_{ij}^{k(0)},X_{ij}^{k(1)})|X_{ij}^{k(0)},X_{ij}^{k(1)},X^k_{i,-j},\widetilde{X}^k_{i,-j})}{P((X_{ij}^k,\widetilde{X}_{ij}^k)=(X_{ij}^{k(0)},X_{ij}^{k(1)})|X_{ij}^{k(1)},X_{ij}^{k(0)},X^k_{i,-j},\widetilde{X}^k_{i,-j})} \text{\quad because $j \in \mathcal{H}_k$}\\
&=&\prod_{i=1}^{n_k}\frac{P_{kj}(X_{ij}^{k(0)}|X^k_{i,-j})Q_{kj}(X_{ij}^{k(1)}|X^k_{i,-j})}{Q_{kj}(X_{ij}^{k(0)}|X^k_{i,-j})P_{kj}(X_{ij}^{k(1)}|X^k_{i,-j})}\\
&=&e^{\mathbbm{1}(W_j>0)\widehat{\text{KL}}_{kj}}.
\end{eqnarray*}
So we have
\begin{eqnarray*}
&&P(W_j>0,E_j\leq \epsilon | |W_j|,W_{-j})\\
&=&\EE{P(W_j>0,E_j\leq \epsilon  |X_{j}^{k(0)},X_{j}^{k(1)},X_{-(k,j)},\widetilde{X}_{-(k,j)},\Y)}\\
&=&\EE{P(W_j>0,\text{sgn}(W_j)\widehat{\text{KL}}_{kj}\leq \epsilon  |X_{j}^{k(0)},X_{j}^{k(1)},X_{-(k,j)},\widetilde{X}_{-(k,j)},\Y)}\\
&=&\EE{\mathbbm{1}\{\text{sgn}(W_j)\widehat{\text{KL}}_{kj}\leq \epsilon\}P(W_j>0  |X_{j}^{k(0)},X_{j}^{k(1)},X_{-(k,j)},\widetilde{X}_{-(k,j)},\Y)}\\
&=&\EE{\mathbbm{1}\{\text{sgn}(W_j)\widehat{\text{KL}}_{kj}\leq \epsilon\}e^{\text{sgn}(W_j)\widehat{\text{KL}}_{kj}}P(W_j<0  |X_{j}^{k(0)},X_{j}^{k(1)},X_{-(k,j)},\widetilde{X}_{-(k,j)},\Y)}\\
&\leq &e^{\epsilon} \EE{P(W_j<0|X_{j}^{k(0)},X_{j}^{k(1)},X_{-(k,j)},\widetilde{X}_{-(k,j)},\Y)}\\
&= &e^{\epsilon} P(W_j<0||W_j|,W_{-j}).
\end{eqnarray*}
This finishes the proof of the lemma.
\end{proof}

\section*{Web Appendix C: Additional simulation details}
\subsection*{C.1: Simulation settings} \label{A:sim_setting}

\subsection*{C.1.1: Simulation settings for $K=2$} 
\paragraph{Data generation}
We first describe the data generation procedure to obtain data for a single experiment $(Y^k, \X^k)$. We first consider Setting 1 (continuous). We sample outcomes from the following linear model:
\begin{eqnarray*}
Y^k_{i}&=&\bmbeta^{k\top}\X^k_{i}+\varepsilon^k_{i},~\text{where $\varepsilon^k_{i}\sim \mathcal{N}(0,\sigma_k^2)$ for $k=1,2$ and $i=1,\cdots, n_k$.}
\end{eqnarray*}
 Here $\sigma_k$ control the signal noise ratio within each sample.

We simulate data with $K=2$ independent experiments with sample sizes $n_1 = n_2 = 1000$ and the number of covariates $p=200$. For each setting, we run $m=1000$ simulations to obtain the power (the expected proportion of true signals being selected) and the FDR. We first generate covariates $\X_i^1\sim \mathcal{N}(\mathbf{0},\bmSigma(\rho_1)),\cdots, \X_i^K\sim \mathcal{N}(\mathbf{0},\bmSigma(\rho_K))$ independently, where $\bmSigma(\rho)$ is an auto-correlation matrix with its $(i,j)$-th element equals to $\rho^{|i-j|}$. Next we generate coefficients $\bmbeta^1, \cdots,\bmbeta^K$ for the $K$ experiments. We use $s_0$ to denote the number of mutual signals among the $K$ experiments, and $s_k$ to denote the signals only present in the $k$-th experiment. We consider two scenarios for the strengths of mutual signals: 1. the directions and strengths of the mutual signals are identical among the $K$ experiments; 2. only the directions of the mutual signals are the same but the signal strengths are independent among the $K$ experiments.   

 For Scenario 1, we sample $\bmeta_{j} \in \mathbb{R}^{s_j}, j \in [K]$, where $\eta_{ji}\sim \textnormal{Uniform}[0,A]$ are independent for $j=[K]$ and $i=1,\cdots, s_j$. Then we sample $\bmepsilon \in \{-1,1\}^p$ where $\epsilon_l$ are independently sampled from rademacher distribution for $l=1,\cdots,p$. With $K=2$, the coefficients $\bmbeta^1, \bmbeta^2$ are determined by:
\begin{eqnarray*}
\bmbeta^1&=&(\bmeta_0^\top, \bmeta_1^\top, \mathbf{0}_{p-s_0-s_1}^\top)^\top\odot \bmepsilon,\\
\bmbeta^2&=&(\bmeta_0^\top, \mathbf{0}_{s_1}^\top,\bmeta_2^\top, \mathbf{0}_{p-s_0-s_1-s_1}^\top)^\top\odot \bmepsilon, ~\text{where $\odot$ is the Hadamard product}.
\end{eqnarray*}
For Scenario 2, we generate $\bmeta_{0k} \in \mathbb{R}^{s_0}$ for $k=1,\cdots,K$ independently; i.e., we sample $\eta_{0ki} \sim \textnormal{Uniform}[0,A]$ independently for $k=1,\cdots,K$ and $i=1,\cdots, s_0$. We generate $\bmeta_{1}$, $\bmeta_2$, and $\bmepsilon$ the same way as described in Scenario 1. The coefficients $\bmbeta^1, \bmbeta^2$ are determined by:
\begin{eqnarray*}
\bmbeta^1&=&(\bmeta_{01}^\top, \bmeta_1^\top, \mathbf{0}_{p-s_0-s_1}^\top)^\top\odot \bmepsilon\\
\bmbeta^2&=&(\bmeta_{02}^\top, \mathbf{0}_{s_1}^\top,\bmeta_2^\top, \mathbf{0}_{p-s_0-s_1-s_1}^\top)^\top\odot \bmepsilon.
\end{eqnarray*}
For the continuous setting, we choose the amplitude of signals as $A=1.2$.

Next, for Setting 2 (binary), we consider the following logistic model:
\begin{eqnarray*}
Y^k_{i}\sim \textnormal{Bernoulli} \left(\frac{\exp(\alpha_k+\bmbeta^{k\top}X^k_{i})}{1+\exp(\alpha_k+\bmbeta^{k\top}X^k_{i})}\right),~\text{for}~k=1,2~\text{and} i=1,\cdots, n_k,
\end{eqnarray*}
where $\X^1_{i}, \X^2_{i}$ and $\bmbeta^1,\bmbeta^2$ are generated the same way as the continuous setting but with the amplitude of signals as $A=2$.

Next, for Setting 3 (mixed), we first generate $(\overline{Y}^k, \X^k)$ following the same procedure as Setting 1. Then we let $Y^1=\overline{Y}^1$ and dichotomize $\overline{Y}^2$ to construct $Y^2$: 
\[Y^2_{i} = \mathbbm{1}\{\overline{Y}^2_{i} \geq 0\}, ~\text{where}~i=1,\cdots,n_2.\]

\paragraph{Simulation settings} We vary the parameters $s_0, s_1, s_2, \rho_1, \rho_2, \sigma_1, \sigma_2, \alpha_1, \alpha_2$ in our simulation studies.
For Setting 1 and Setting 3, we change the parameters as follows: 
\begin{itemize}
    \item Varying $s_0$, $s_1$, $s_2$. Fixing $\rho_1 =\rho_2 = 0.5$ and $\sigma_1^2 = 1, \sigma_2^2 = 4$. 
    \begin{align*}
        &\text{1. Fixing $s_1 = s_2 =0$, we vary $s_0 = 10,20,30,40,50,60$;}\\
        &\text{2. Fixing $s_0 =40$, we vary $s_1 = s_2= 10,20,30,40,50,60$;}\\
        &\text{3. Fixing $s_0 =40$ and $s_1=0$, we vary $s_2= 10,20,30,40,50,60$.}\\
    \end{align*}
    \item Varying $\rho_1$, $\rho_2$. Fixing $\sigma_1^2 = 1, \sigma_2^2 =4$, Let $\rho_1 = 0.25, 0.3, 0.35, 0.4, 0.45, 0.5$ and $\rho_2 = 1-\rho_1$ for the following three choices of $s_0$, $s_1$, $s_2$:
    \begin{align*}
        & s_0 =40,s_1=0,s_2=0;\\
        & s_0 =40,s_1=40,s_2=40;\\
        & s_0=40,s_1=0,s_2=40.
    \end{align*}
    \item Varying $\sigma_1$, $\sigma_2$. Fixing $\rho_1=\rho_2 =0.5$, Let $\sigma_1^2 +\sigma_2^2 =5, \sigma_1^2 = 0.5, 1, 1.5., 2, 2.5$, for the following three choices of $s_0$, $s_1$, $s_2$:
    \begin{align*}
        & s_0 =40,s_1=0,s_2=0;\\
        & s_0 =40,s_1=40,s_2=40;\\
        & s_0=40,s_1=0,s_2=40.
    \end{align*}
\end{itemize}

For Setting 2 (binary), we vary the following parameters. 
\begin{itemize}
    \item Varying $s_0$, $s_1$, $s_2$. Fixing $\rho_1 =\rho_2 = 0.5$ and $\alpha_1 = 1, \alpha_2 = -1$. 
    \begin{align*}
        &\text{1. Fixing $s_1 = s_2 =0$, we vary $s_0 = 10,20,30,40,50,60$;}\\
        &\text{2. Fixing $s_0 =40$, we vary $s_1 = s_2= 10,20,30,40,50,60$;}\\
        &\text{3. Fixing $s_0 =40$ and $s_1=0$, we vary $s_2= 10,20,30,40,50,60$.}\\
    \end{align*}
    \item Varying $\rho_1$, $\rho_2$. Fixing $\alpha_1 = 1, \alpha_2 =-1$, Let $\rho_1 = 0.25, 0.3, 0.35, 0.4, 0.45, 0.5$ and $\rho_2 = 1-\rho_1$ for the following three choices of $s_0$, $s_1$, $s_2$:
    \begin{align*}
        & s_0 =40,s_1=0,s_2=0;\\
        & s_0 =40,s_1=40,s_2=40;\\
        & s_0=40,s_1=0,s_2=40.
    \end{align*}
    \item Varying $\alpha_1$, $\alpha_2$. Fixing $\rho_1=\rho_2 =0.5$, Let $\alpha_1 +\alpha_2 =0, \alpha_1 = 0, 0.5, 1, 1.5, 2, 2.5$, for the following three choices of $s_0$, $s_1$, $s_2$:
    \begin{align*}
        & s_0 =40,s_1=0,s_2=0;\\
        & s_0 =40,s_1=40,s_2=40;\\
        & s_0=40,s_1=0,s_2=40.
    \end{align*}
\end{itemize}

\subsection*{C.1.2: Simulation settings for $K=3$} 
\paragraph{Data generation}
We first describe the data generation procedure to obtain data for a single experiment $(Y^k, \X^k)$. We sample outcomes from the following linear model:
\begin{eqnarray*}
Y^k_{i}&=&\bmbeta^{k\top}\X^k_{i}+\varepsilon^k_{i},~\text{where $\varepsilon^k_{i}\sim \mathcal{N}(0,\sigma_k^2)$ for $k=1,2,3$ and $i=1,\cdots, n_k$.}
\end{eqnarray*}
 Here $\sigma_k$ control the signal noise ratio within each sample.

We simulate data with $K=3$ independent experiments with sample sizes $n_1 = n_2 = n_3 = 1000$ and the number of covariates $p=200$. For each setting, we run $m=1000$ simulations to obtain the power (the expected proportion of true signals being selected) and the FDR. We first generate covariates $\X_i^1\sim \mathcal{N}(\mathbf{0},\bmSigma(\rho_1)),\cdots, \X_i^K\sim \mathcal{N}(\mathbf{0},\bmSigma(\rho_K))$ independently, where $\bmSigma(\rho)$ is an auto-correlation matrix with its $(i,j)$-th element equals to $\rho^{|i-j|}$. Next we generate coefficients $\bmbeta^1, \cdots,\bmbeta^K$ for the $K$ experiments. We use $s_0$ to denote the number of mutual signals among the $K$ experiments, $s_1$, $s_2$, $s_3$ to be the signals only present in experiment 1, 2, 3 respectively, and $s_{12}$, $s_{13}$, $s_{23}$ to be the signals only present in experiments 1 and 2, 1 and 3, and 2 and 3 respectively. We consider the scenarios that only the direction of the mutual signals are the same but the signal strengths are independent among the $K$ experiments.   

 We generate $\bmeta_{0k} \in \mathbb{R}^{s_0}$ for $k=[K]$ independently; i.e., we sample $\eta_{0ki} \sim \textnormal{Uniform}[0,1]$ independently for $k=[K]$ and $i=1,\cdots, s_0$. Then we sample $\bmeta_{j} \in \mathbb{R}^{s_j}, j \in \{1,2,3,12,13,23\}$, where $\eta_{ji}\sim \textnormal{Uniform}[0,A]$ are independent for $j=[K]$ and $i=1,\cdots, s_j$. Then we sample $\bmepsilon \in \{-1,1\}^p$ where $\epsilon_l$ are independently sampled from Rademacher distribution for $l=1,\cdots,p$. With $K=3$, the coefficients $\bmbeta^1, \bmbeta^2$, $\bmbeta^3$ are determined by:
\begin{eqnarray*}
\bmbeta^1&=&(\bmeta_{01}^\top, \bmeta_1^\top,\mathbf{0}_{s_2},\mathbf{0}_{s_3},\bmeta_{12},\bmeta_{13},\mathbf{0}_{s_{23}}, \mathbf{0}_{p_s}^\top)^\top\odot \epsilon,\\
\bmbeta^2&=&(\bmeta_{02}^\top, \mathbf{0}_{s_1},\bmeta_2^\top,\mathbf{0}_{s_3},\bmeta_{12},\mathbf{0}_{s_{13}},\bmeta_{23}, \mathbf{0}_{p_s}^\top)^\top\odot \epsilon,\\
\bmbeta^3&=&(\bmeta_{03}^\top, \mathbf{0}_{s_1},\mathbf{0}_{s_2},\eta_3^\top,\mathbf{0}_{s_{12}},\eta_{13},\bmeta_{23}, \mathbf{0}_{p_s}^\top)^\top\odot \epsilon,
~\text{where $\odot$ is the Hadamard product},
\end{eqnarray*}
and $p_s=p-s_0-s_1-s_2-s_3-s_{12}-s_{13}-s_{23}$. 

In the simulation, we fix $\rho_1=0.4$, $\rho_2=0.5$, $\rho_3=0.6$, $q=0.2$, $\sigma_1=1$, $\sigma_2=2$, $\sigma_3=1.5$ and $A=1.2$.

\subsection*{C.1.3: Simulation settings for power comparison} 
When using the difference function ($\W=\Z-\widetilde{\Z}$) as the flip-sign function, we have
$\W=\odot_{k=1}^K (\Z^k-\widetilde{\Z}^k)$ which is clear to have power because if feature $j$ has effects within each subsample (i.e., $Z_{kj}-\widetilde{Z}_{kj}>0$ with high probability), then we will expect $W_{j}>0$ with high probability. However, for other combination functions, this is not very straightforward. So we plot the distribution under various null and alternatives of 8 functions defined below to show they have power when $K=3$ and test statistics $Z_{jk}$s from the absolute value of the normal distribution. 
\begin{itemize}
\item Max: $\W=(\Z\vee \widetilde{\Z})\odot(-1)^{\mathbbm{1}(\Z<\widetilde{\Z})}$
\item Diff: $\W=\Z-\widetilde{\Z}=|\Z-\widetilde{\Z}|\odot(-1)^{\mathbbm{1}(\Z<\widetilde{\Z})}$
\item Direct Max Sum: $\W=\sum_{k=1}^K(\Z^k\vee \widetilde{\Z}^k)\odot(-1)^{\mathbbm{1}(\Z<\widetilde{\Z})}$
\item Direct Max Max: $\W=\max_{k=1}^K(\Z^k\vee \widetilde{\Z}^k)\odot(-1)^{\mathbbm{1}(\Z<\widetilde{\Z})}$
\item Direct Diff Sum: $\W=\sum_{k=1}^K|\Z^k-\widetilde{\Z}^k|\odot(-1)^{\mathbbm{1}(\Z<\widetilde{\Z})}$
\item Direct Diff Max: $\W=\max_{k=1}^K|\Z^k-\widetilde{\Z}^k|\odot(-1)^{\mathbbm{1}(\Z<\widetilde{\Z})}$
\item Direct Sum Sum: $\W=\sum_{k=1}^K(\Z^k+\widetilde{\Z}^k)\odot(-1)^{\mathbbm{1}(\Z<\widetilde{\Z})}$
\item Direct Sum Max: $\W=\max_{k=1}^K|\Z^k+\widetilde{\Z}^k|\odot(-1)^{\mathbbm{1}(\Z<\widetilde{\Z})}$,
\end{itemize}
where the max between vectors are taken elementwise.

Additional simulations were performed for $K=2$ under the settings $s_0=40$, $q=0.2$, $\rho_1=0.4$, $\rho_2=0.6$, $\sigma_1=1$,$\sigma_2=2$, $A=1.2$ and $s_1=s_2=0$ or $s_1=s_2=s_0$, for Scenarios 1 and 2 to compare the power of these functions and it turns out that the signed max and difference functions have the best performance among the functions we explored when there are sample-specific effects (i.e., $s_1,s_2\neq 0$).

To see the potential limitations of the proposed method, we consider the setting for $K=2$ under the settings $s_0=40$, $q=0.2$, $\rho_1=0.4$, $\rho_2=0.6$, $\sigma_1=1$,$\sigma_2=2$. We let the fake signals (those shown only in one of the two studies) have stronger effects than the true signals. We vary the ratio between the magnitudes of the fake signals and the true signals from 1.5 to 5. 

\subsection*{C.2 Algorithm implementing details for simulation}
\begin{itemize}
\item Knockoffs construction: For the \textit{simultaneous} and \textit{intersection} knockoffs, we construct the knockoffs for each experiment separately. We use the Fixed-X knockoff method for the continuous outcome setting of sections C.1.1 and C.1.2 and use the second-order Model-X knockoff construction method for the binary outcomes setting of section C.1.1. For the mixed outcome setting of section C.1.1, we use the Fixed-X knockoff method for the first experiment with the continuous outcome and use the second-order Model-X knockoff construction method for the second experiment with binary outcome. For \textit{pooling} knockoffs, we construct knockoffs for the pooled data. We use the Fixed-X knockoff method for the continuous outcome setting of sections C.1.1 and C.1.2 and use the second-order Model-X knockoff construction method for the binary outcomes setting of section C.1.1.
\item Statistics calculation: For each experiment (\textit{simultaneous} and \textit{intersection} knockoffs) or pooled data (\textit{pooling} knockoffs), We choose the absolute value of the coefficients from the $\ell_1$-penalized linear regression (if the outcome is continuous) of the outcome on both original features and knockoffs or $\ell_1$-penalized logistic regression (if the outcome is binary) of the outcome on both original features and knockoffs with the tuning parameters selected from 5-fold cross-validation as our statistics $[\Z^k,\widetilde{\Z}^k]$s.
\item Calculating the filter statistics $\W$: For the \textit{simultaneous} knockoffs, we use the OSSF function derived from the composition of the OSCF in equation (10) of the main paper and the flip sign function in equation (9) of the main paper. For \textit{intersection} and \textit{pooling} knockoffs, we use the antisymmetric function $\W=\Z-\widetilde{\Z}$.
\item Calculate the threshold and selection set: we control the FDR at a level of $q=0.2$ and use the knockoff{\color{black}+} filter as defined in equation (8) of the main paper. 
\end{itemize}

\subsection*{C.3 Additional simulation results}\label{app:sim_res}
\subsubsection*{C.3.1 Additional simulation results for continuous outcomes with $K=2$\\}
In Figures \ref{fig:cont_samesig=1} and \ref{fig:cont_samesig=0}, we show results for Setting 1 (continuous) with Scenario 1 (same signal strengths) and Scenario 2 (different signal strengths). 
\begin{figure}[!p]
    \centering
    \includegraphics[scale=0.3]{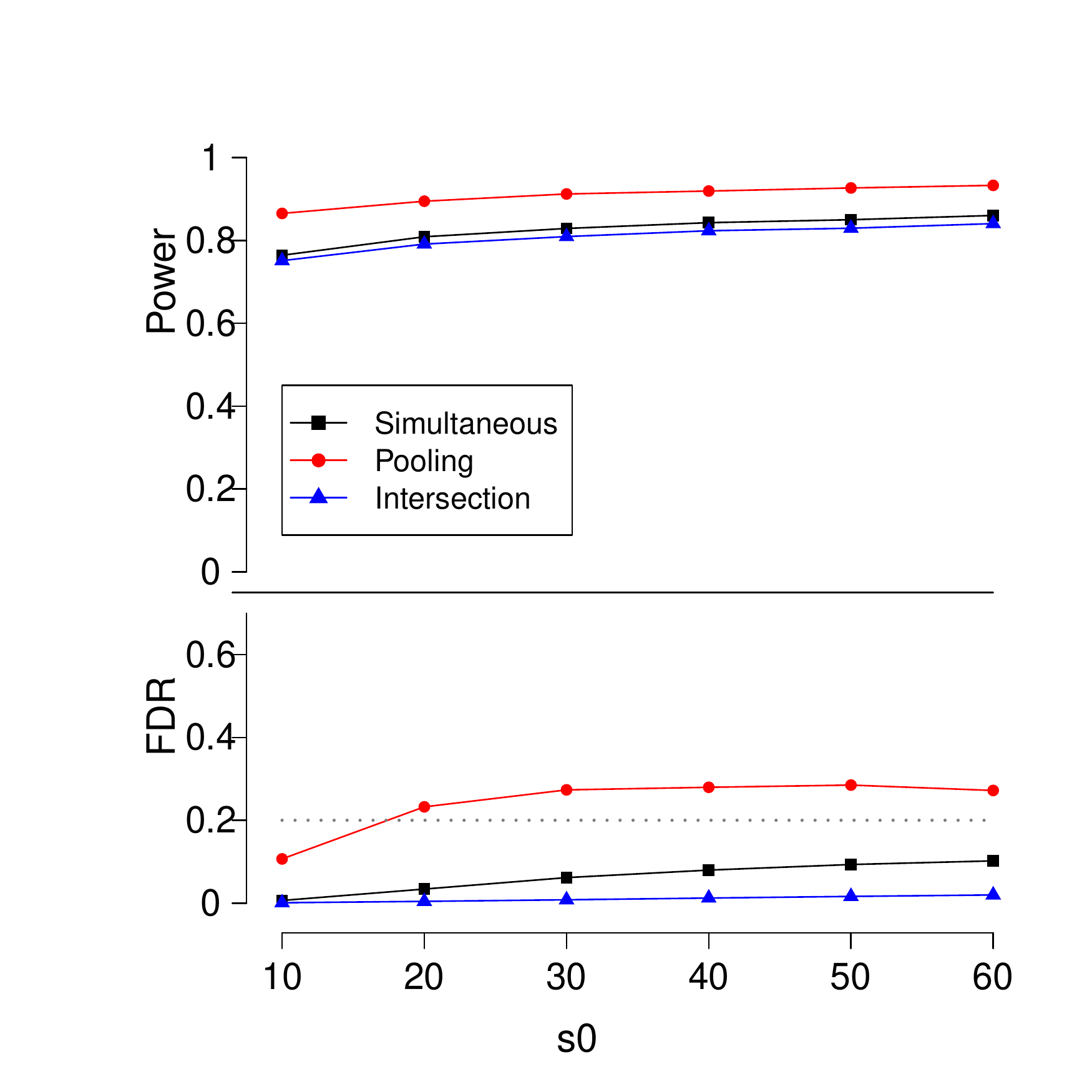}
    \includegraphics[scale=0.3]{Continuous_s1_s2.pdf}
    \includegraphics[scale=0.3]{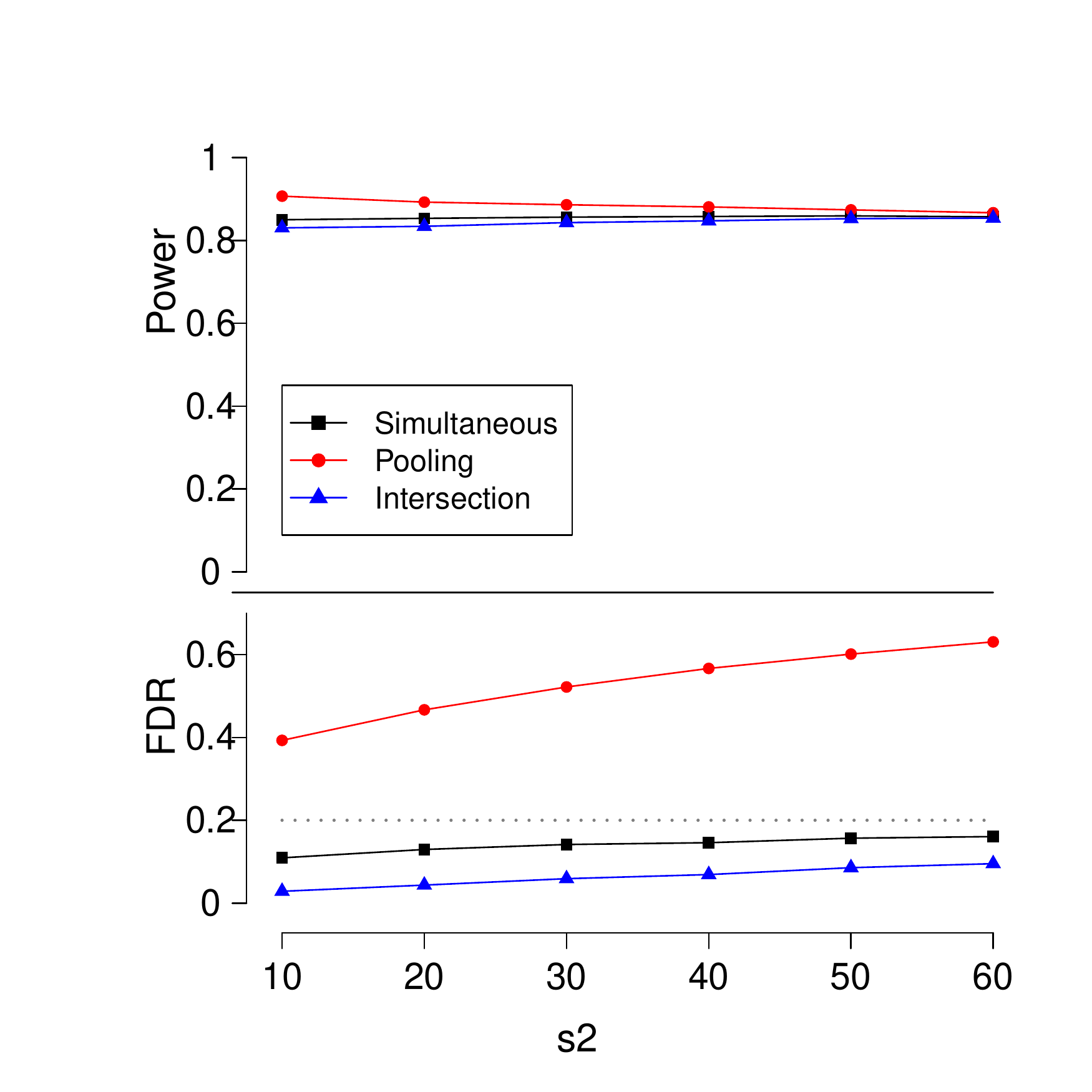}
     \includegraphics[scale=0.3]{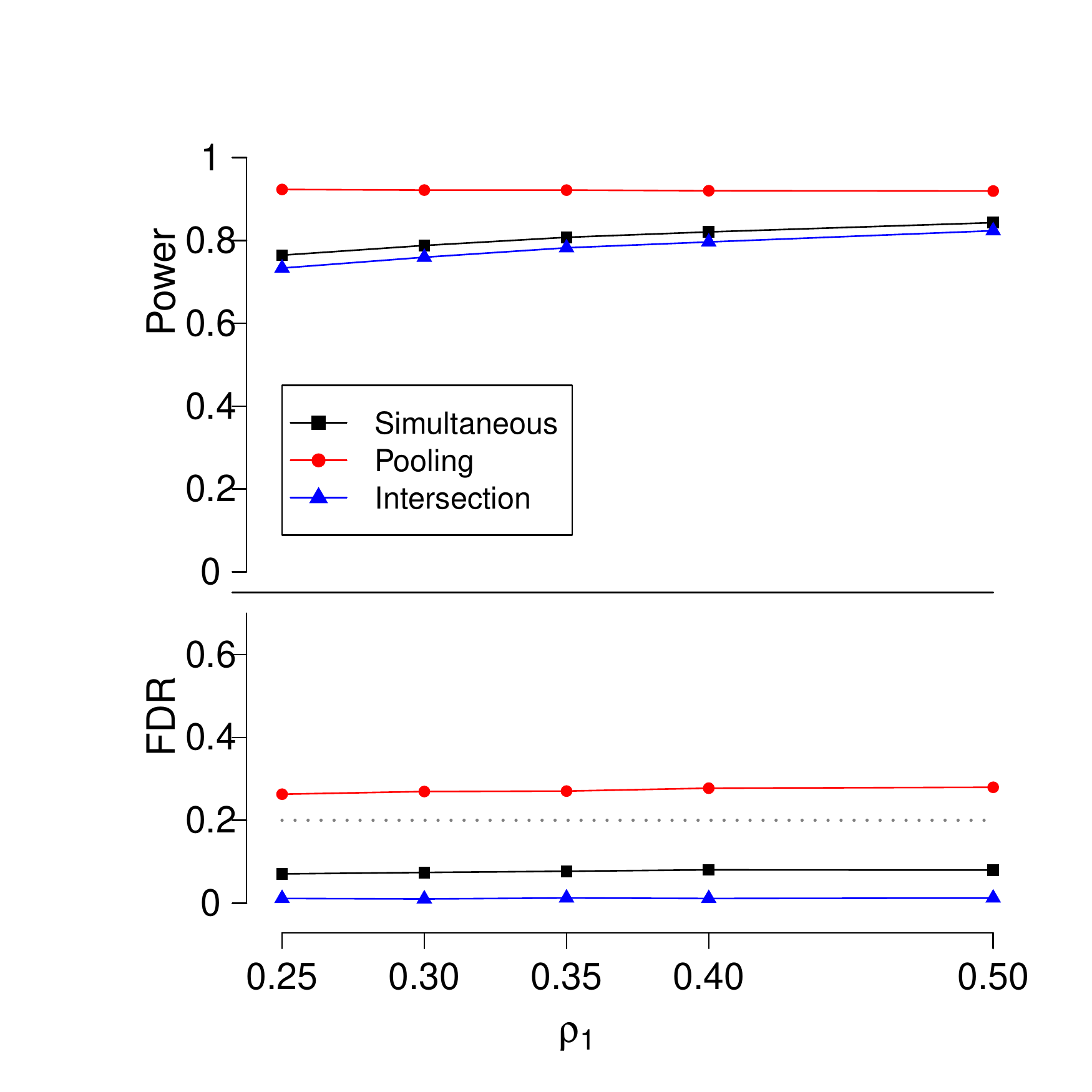}
    \includegraphics[scale=0.3]{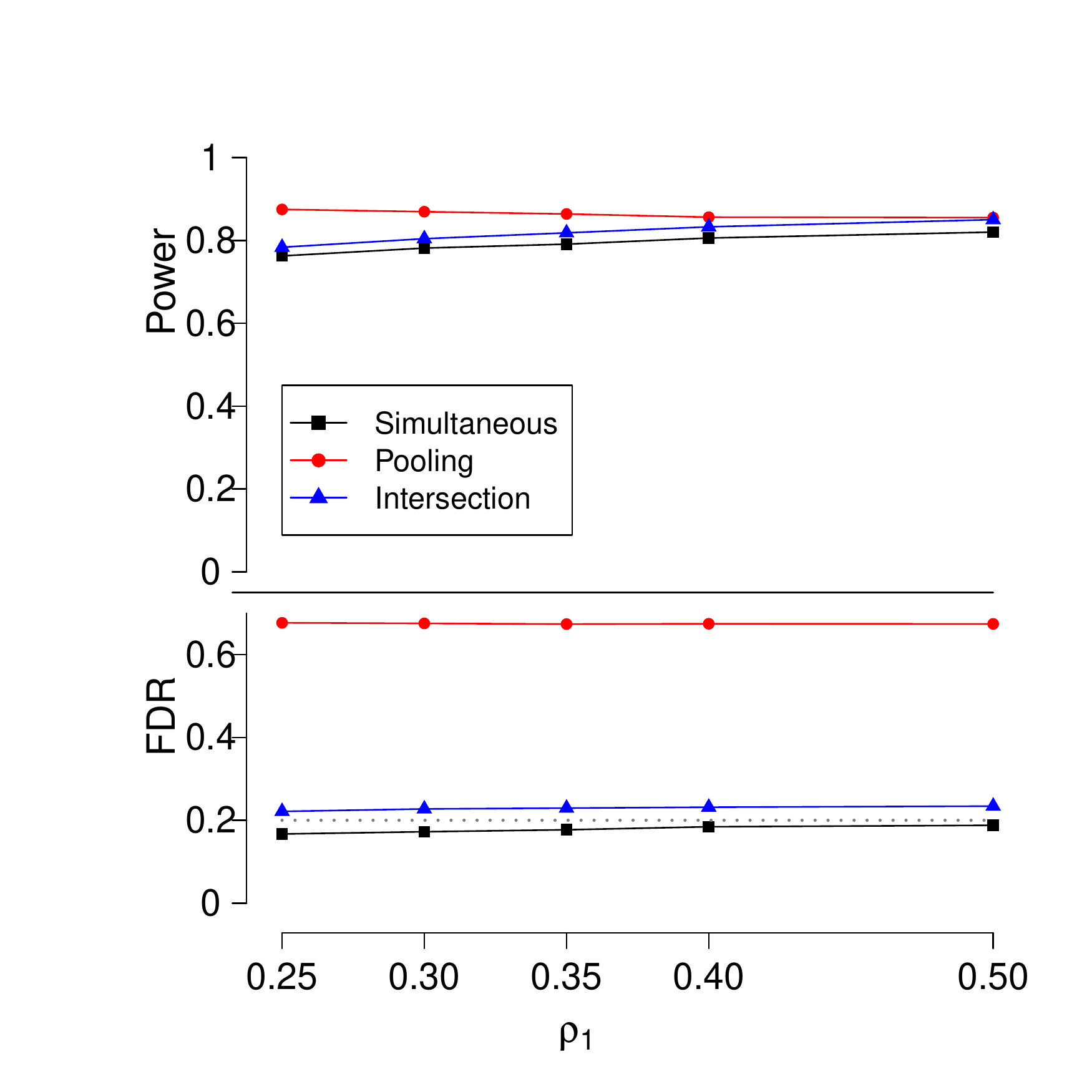}
    \includegraphics[scale=0.3]{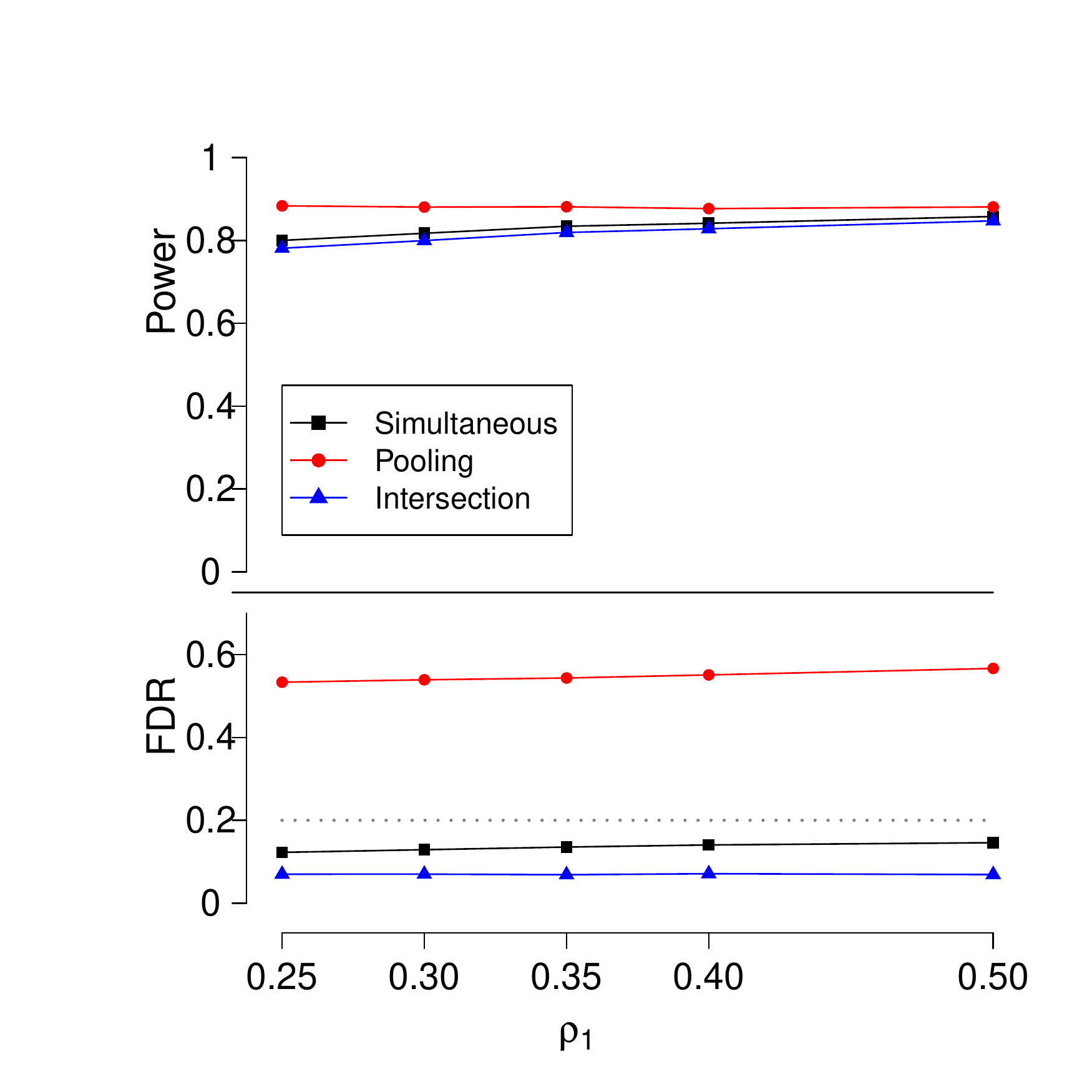}
    \includegraphics[scale=0.3]{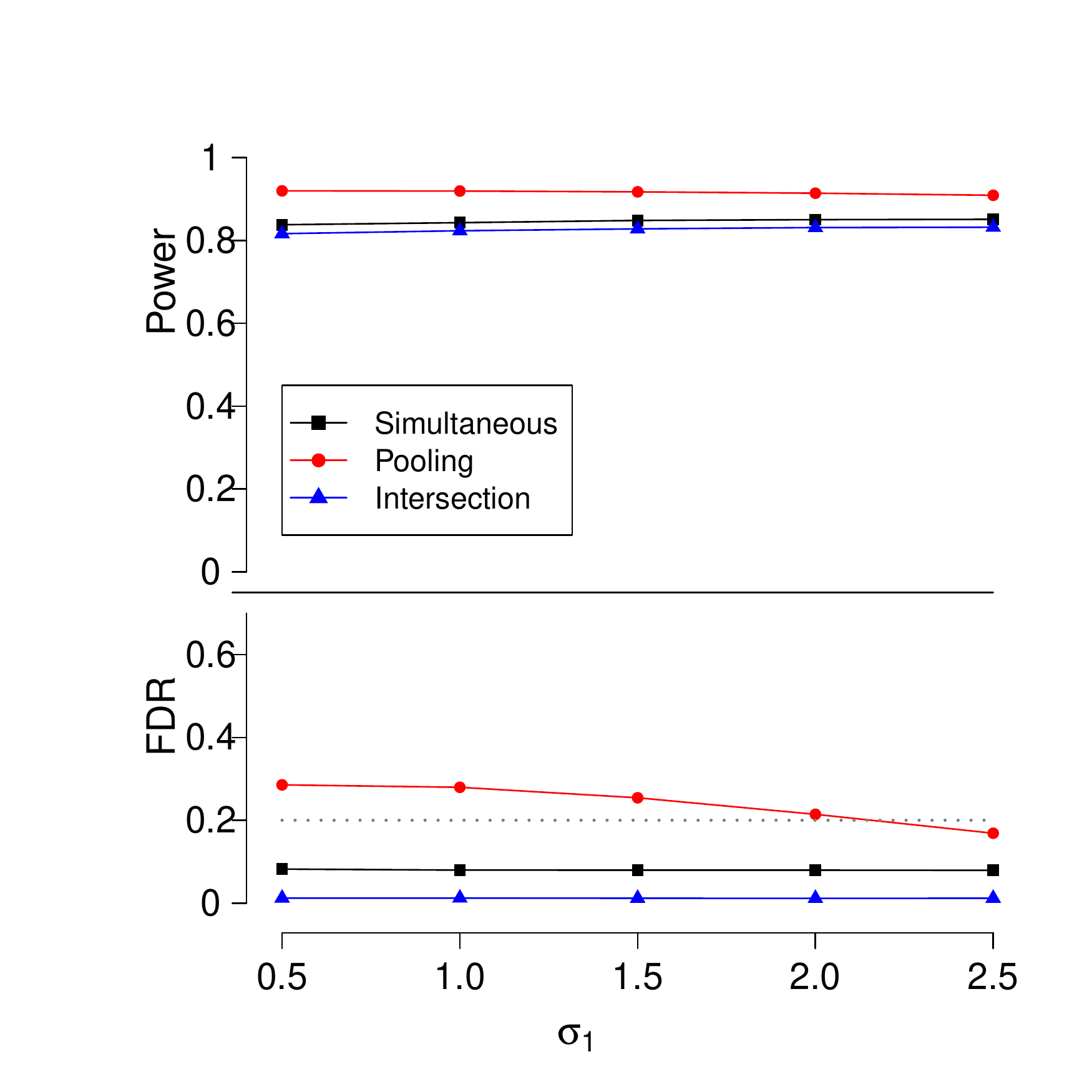}
    \includegraphics[scale=0.3]{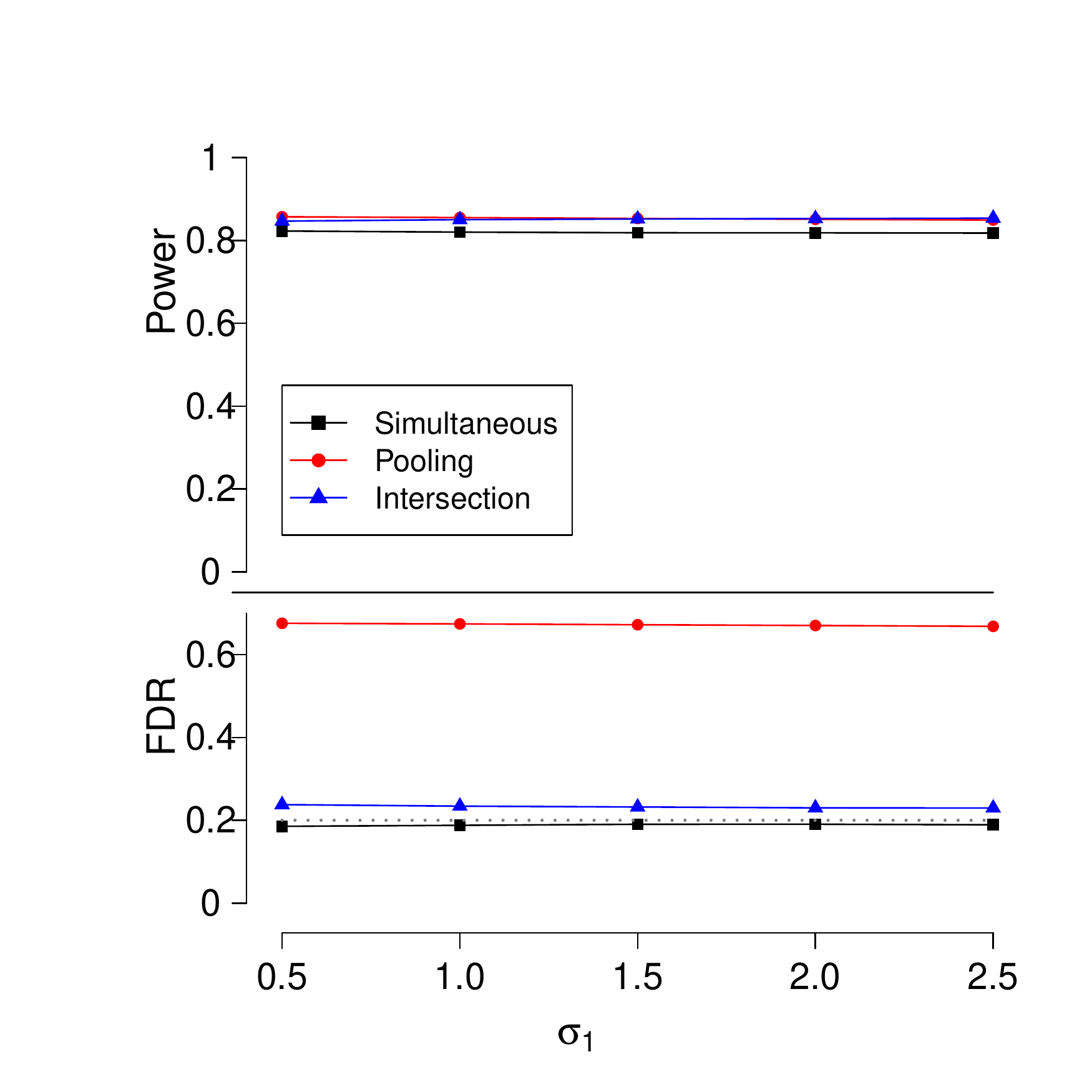}
    \includegraphics[scale=0.3]{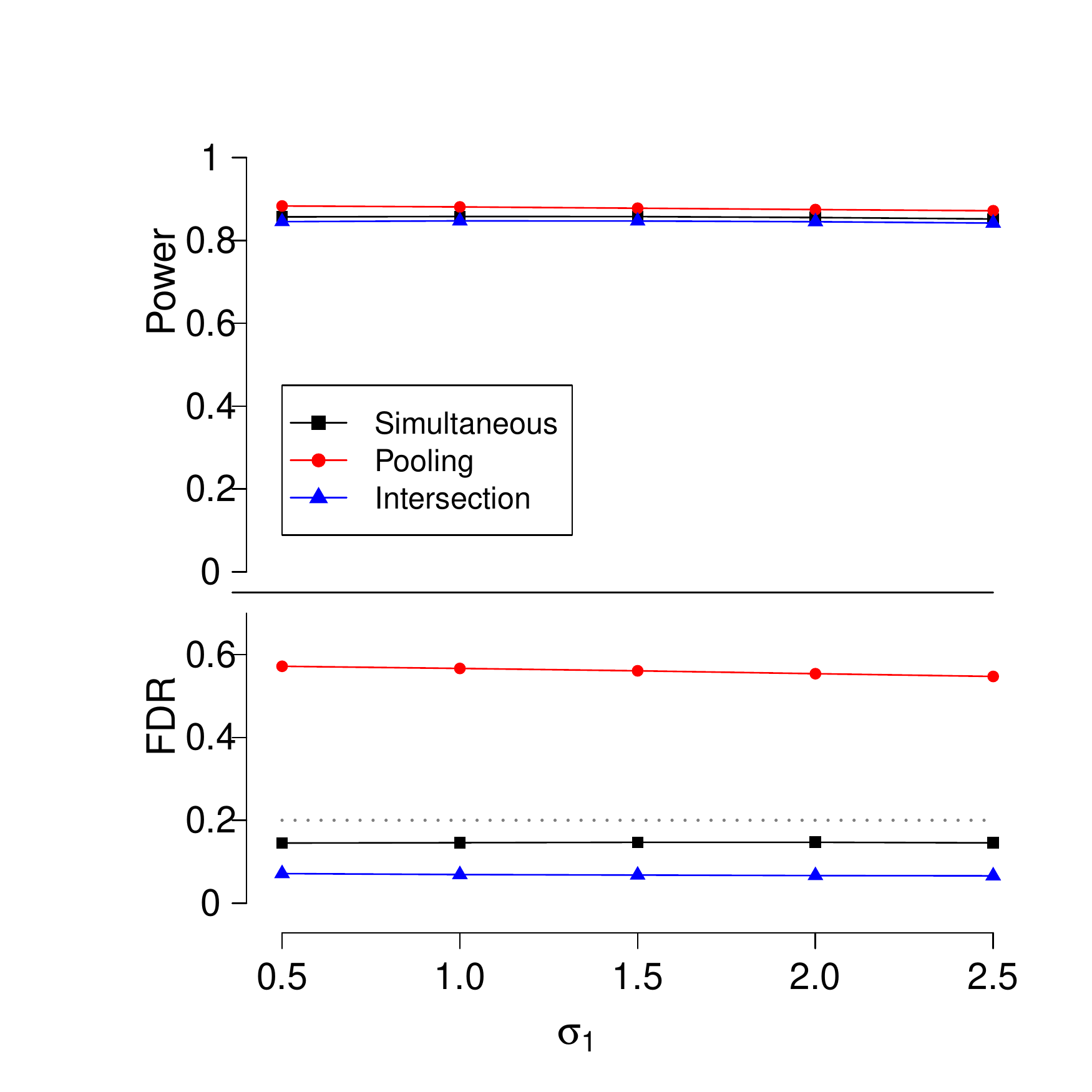}
    \caption{The power and the FDR for simulations with Setting 1 (continuous $K=2$) and Scenario 1 data. Column 1 includes the settings with $s_1=s_2=0$, column 2 includes the settings with $s_1=s_2\neq 0$, column 3 includes the settings with $s_1\neq 0, s_2=0$. Row 1 shows the experiments varying $s_0, s_1, s_2$, row 2 shows experiments varying $\rho_1, \rho_2$, row 3 shows experiments varying $\sigma_1, \sigma_2$.}
    \label{fig:cont_samesig=1}
\end{figure}

\begin{figure}[!p]
    \centering
    \includegraphics[scale=0.3]{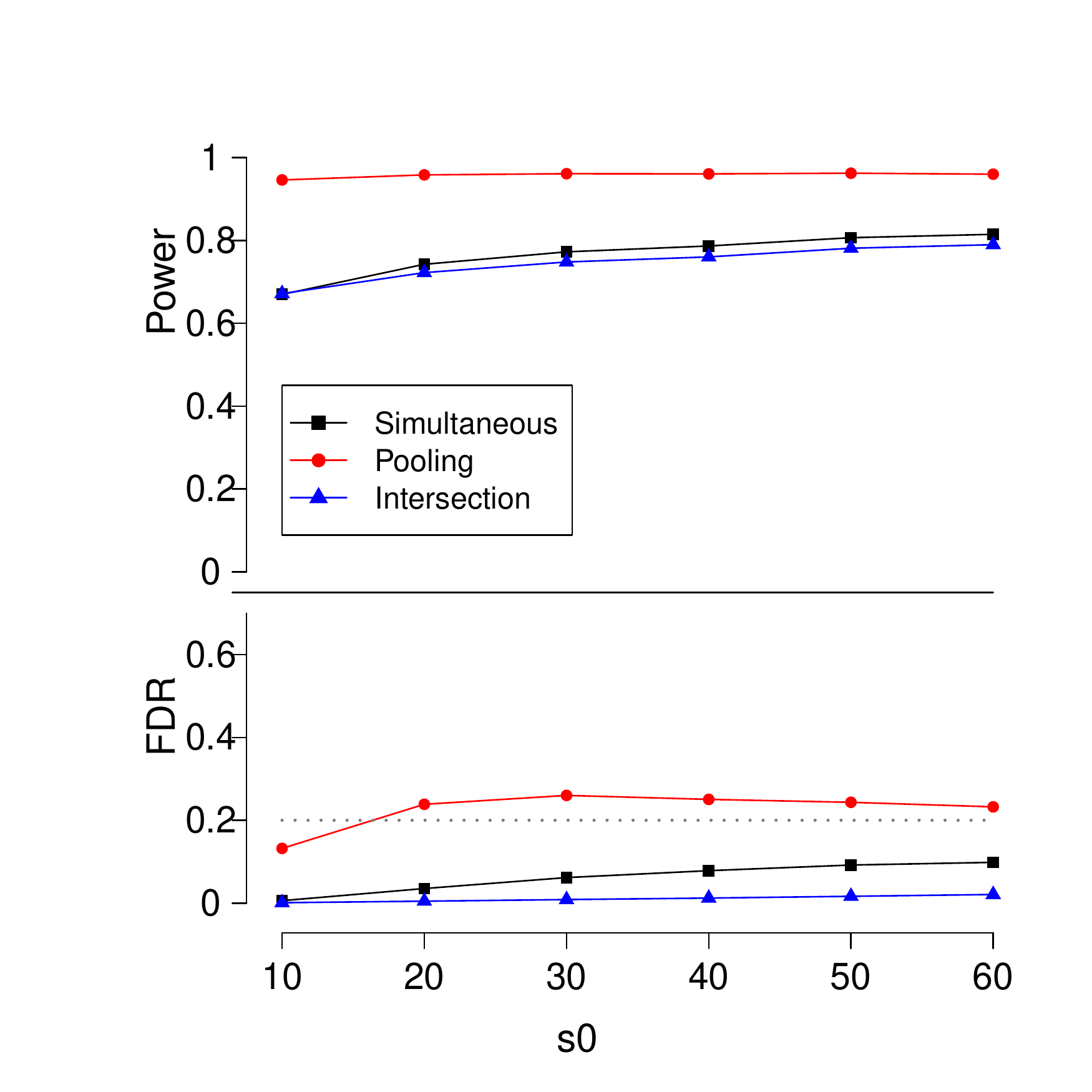}
    \includegraphics[scale=0.3]{Continuous_s1_s2samesig=0.pdf}
    \includegraphics[scale=0.3]{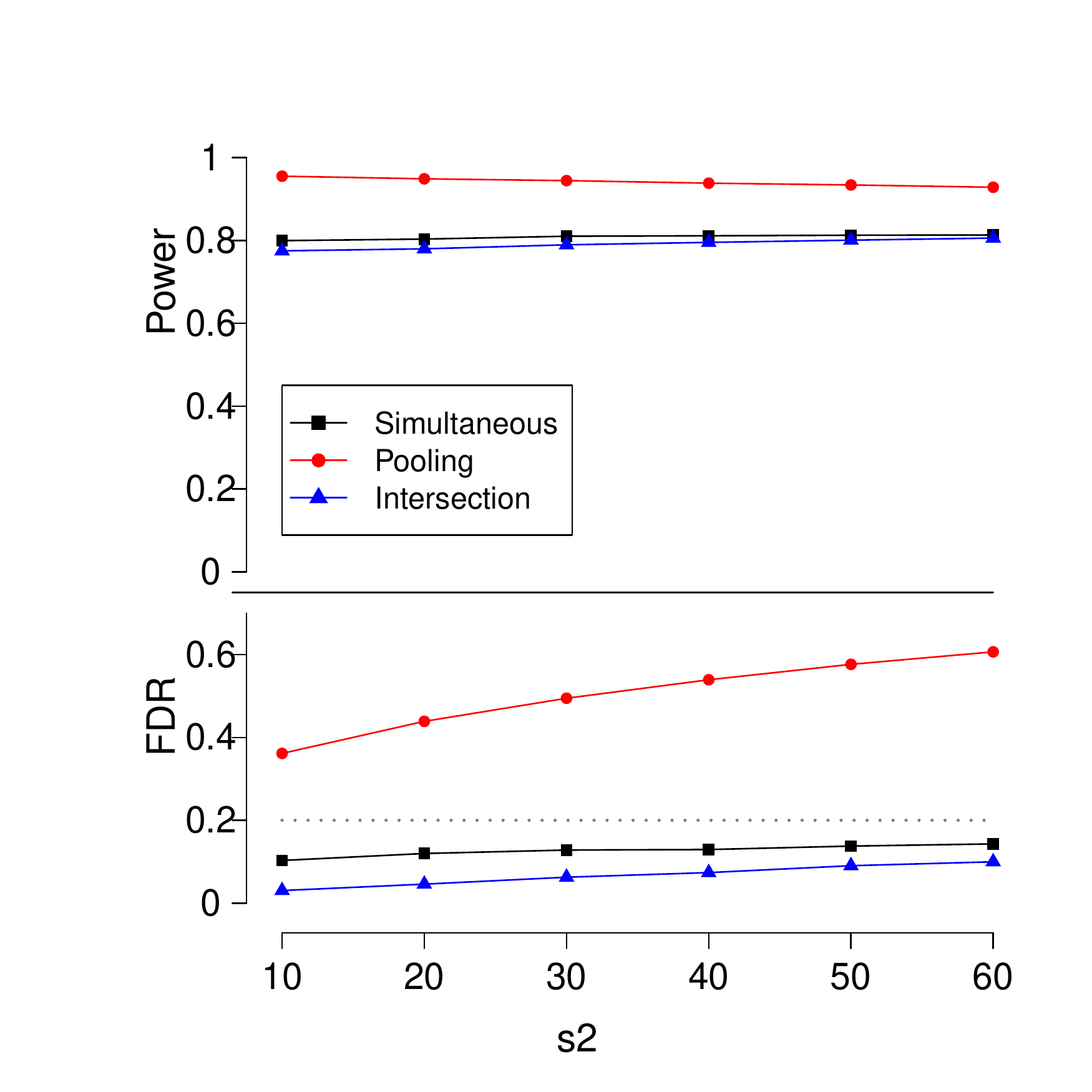}
     \includegraphics[scale=0.3]{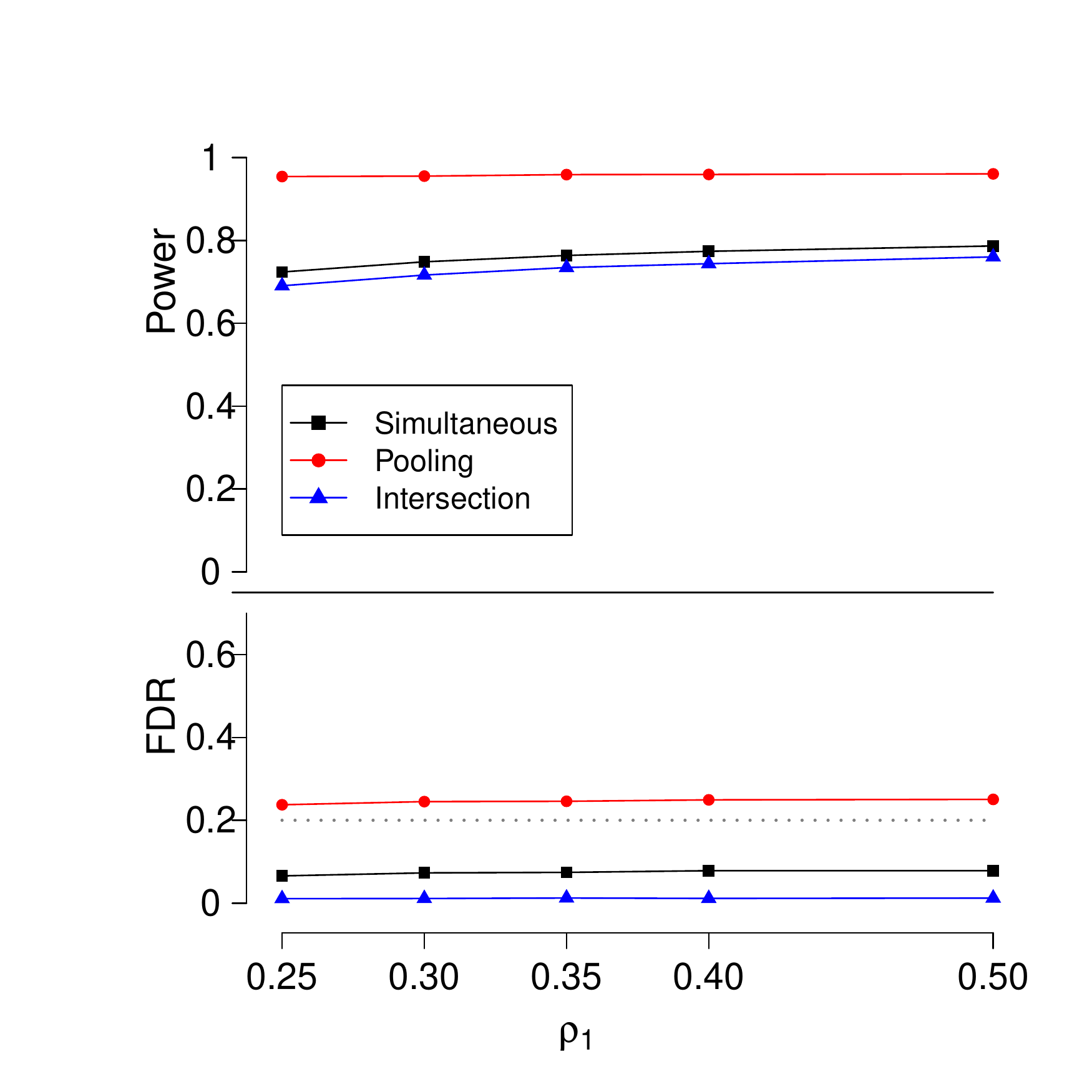}
    \includegraphics[scale=0.3]{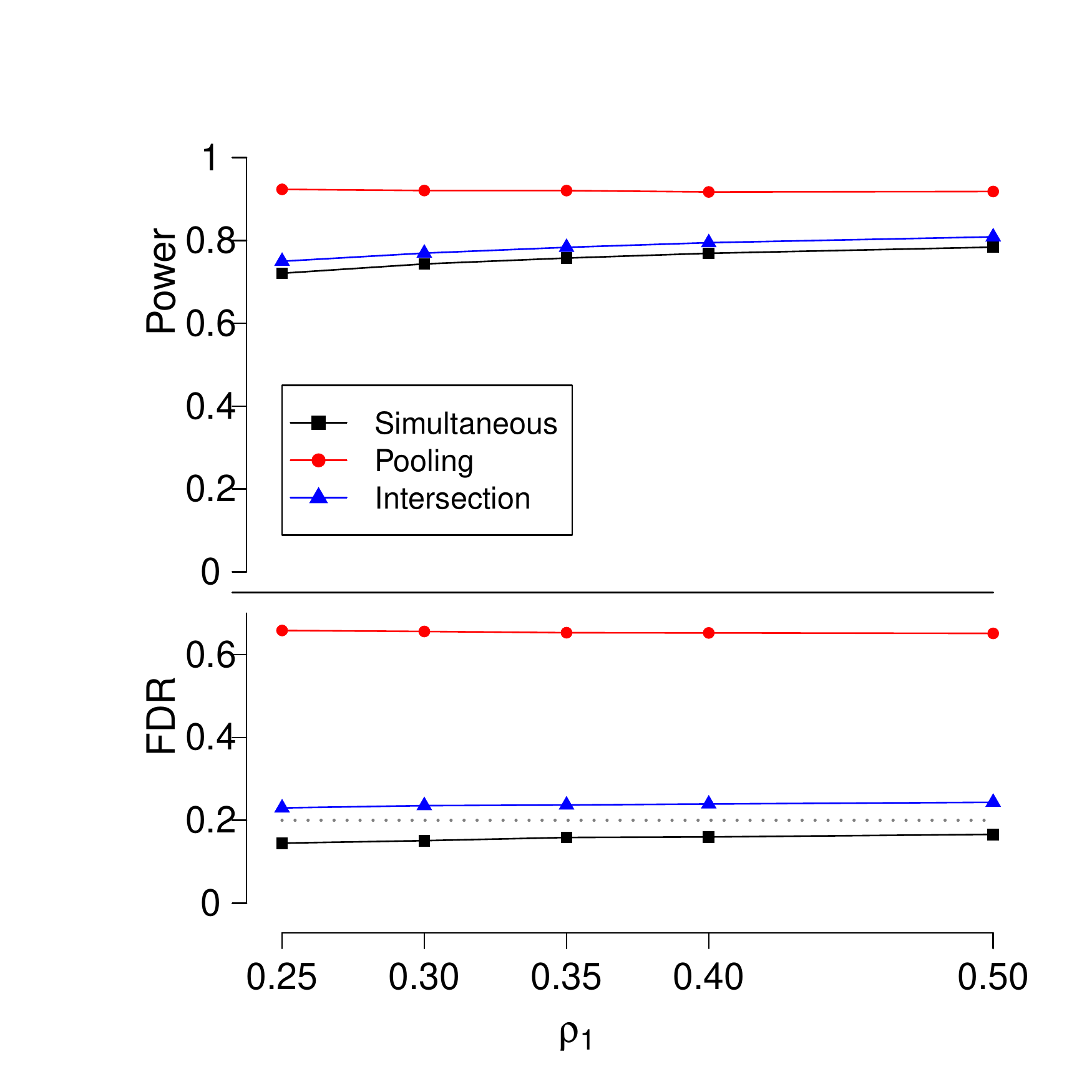}
    \includegraphics[scale=0.3]{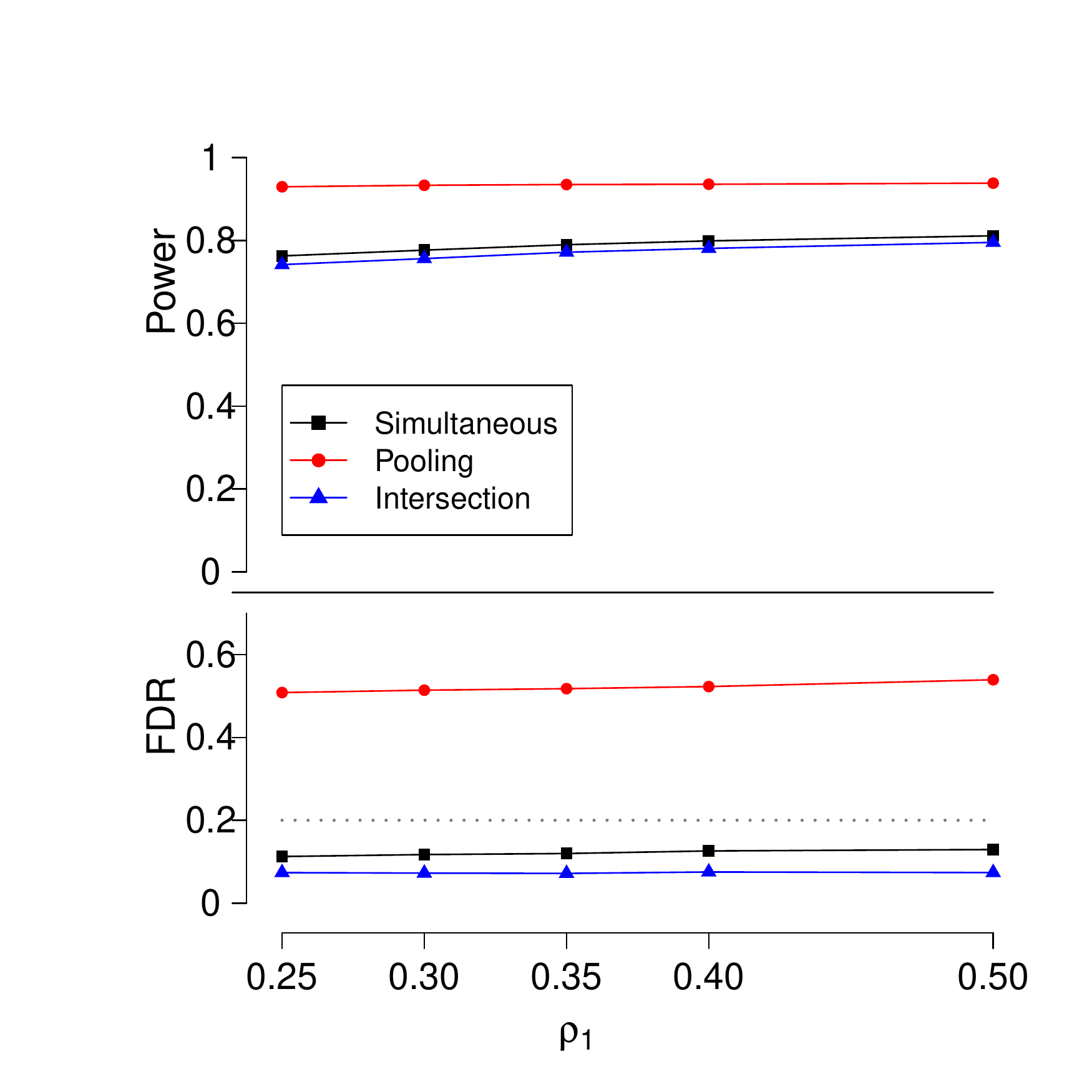}
    \includegraphics[scale=0.3]{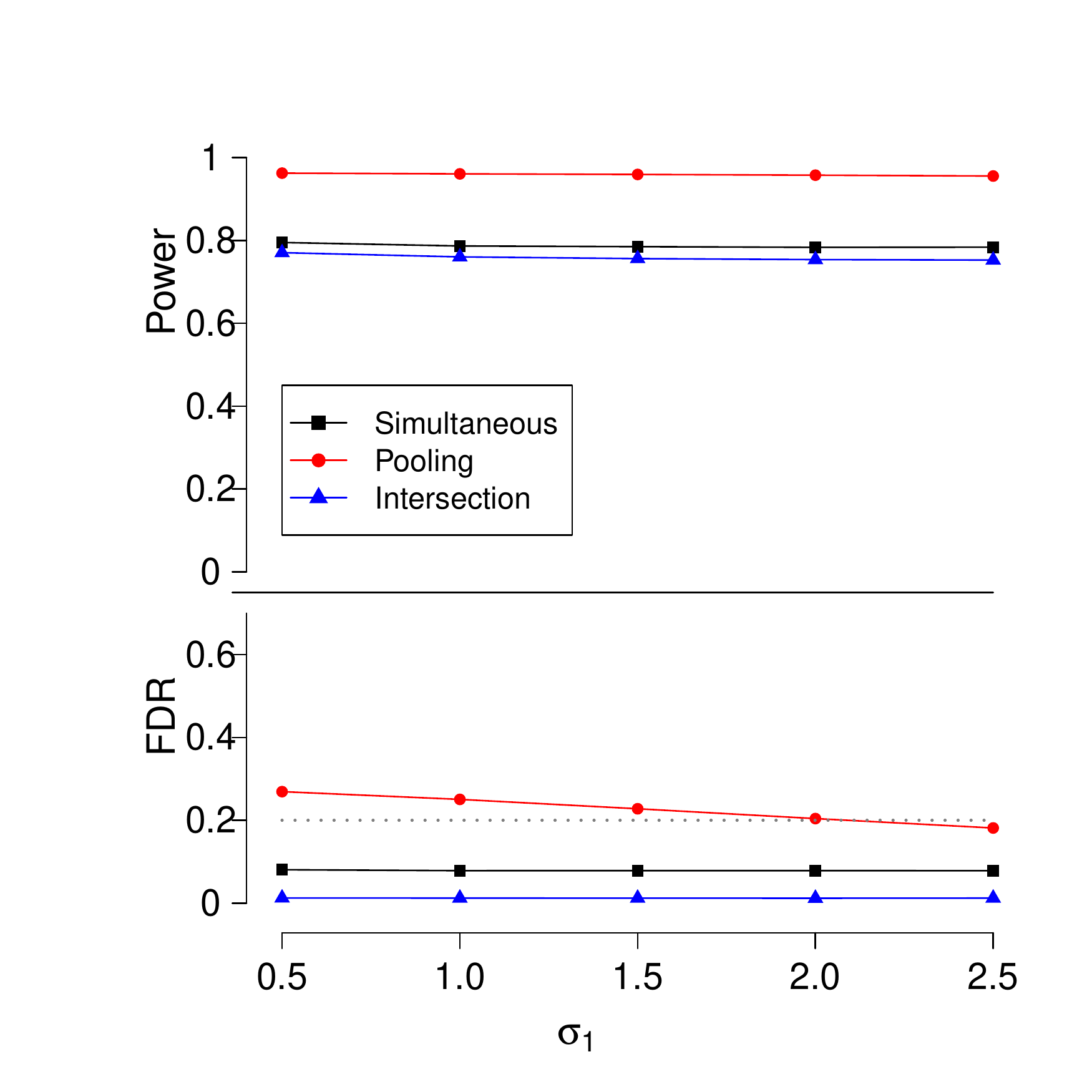}
    \includegraphics[scale=0.3]{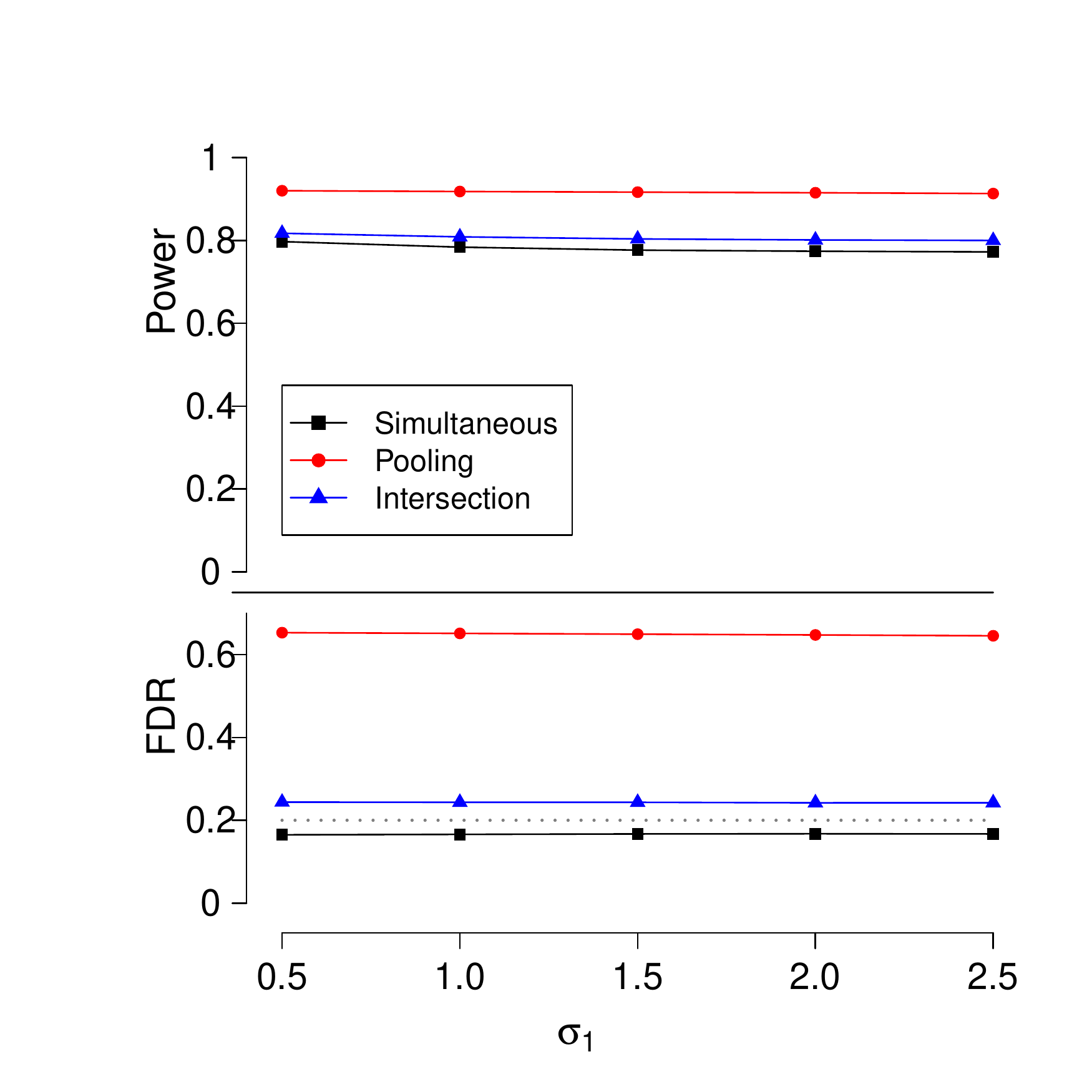}
    \includegraphics[scale=0.3]{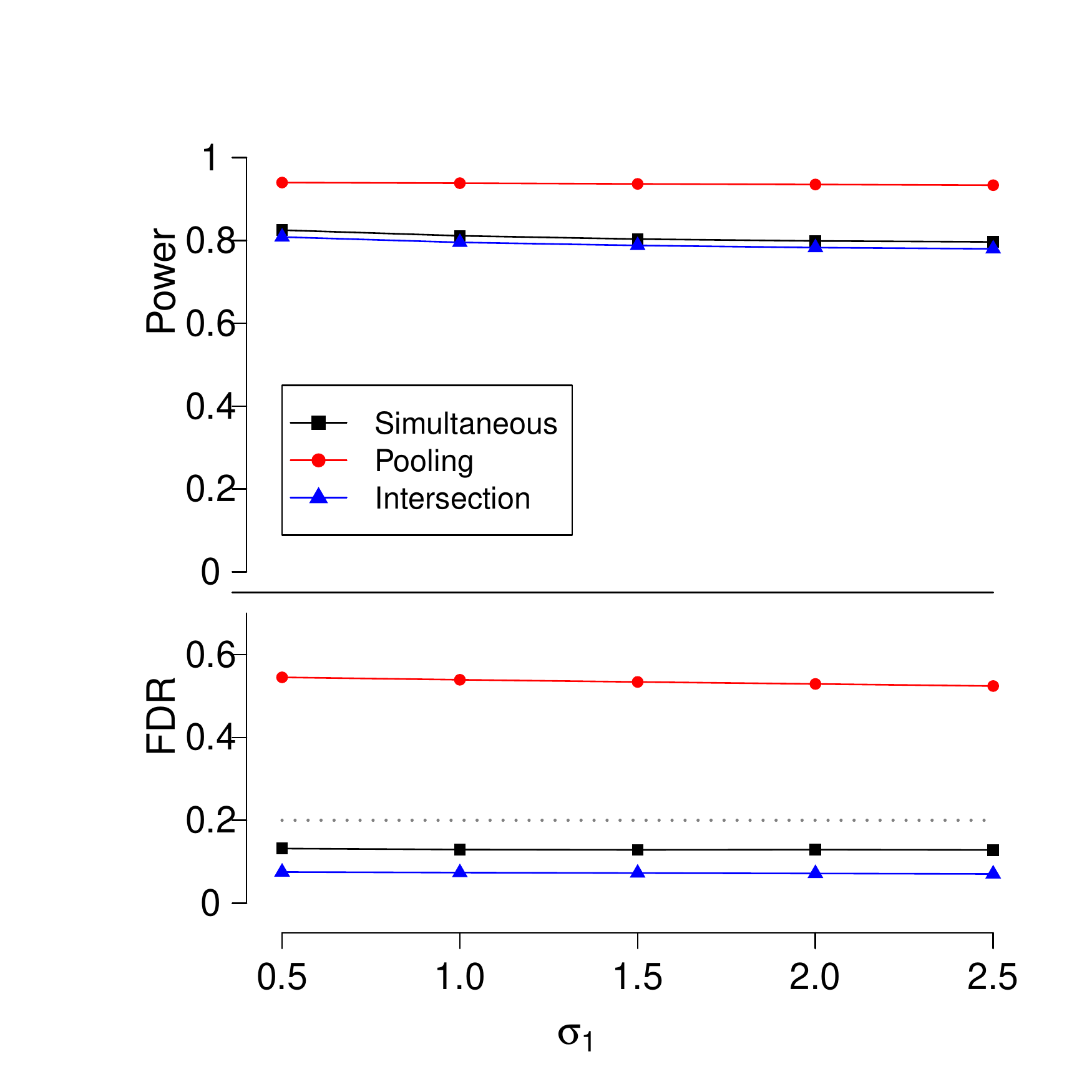}
    \caption{The power and the FDR for simulations with Setting 1 (continuous $K=2$) and Scenario 2 data. Column 1 includes the settings with $s_1=s_2=0$, column 2 includes the settings with $s_1=s_2\neq 0$, column 3 includes the settings with $s_1\neq 0, s_2=0$. Row 1 shows the experiments varying $s_0, s_1, s_2$, row 2 shows experiments varying $\rho_1, \rho_2$, row 3 shows experiments varying $\sigma_1, \sigma_2$.}
    \label{fig:cont_samesig=0}
\end{figure}

For the continuous settings, when there are only simultaneous signals present ($s_1=s_2=0$) and the signal strengths are equal ($\sigma_1=\sigma_2$), \textit{pooling} method also controls the FDR. When $s_1=s_2=0$, but signal strengths are different ($\sigma_1\neq \sigma_2$), \textit{pooling} method does not control the FDR. When there exist non-mutual signals in one group ($s_1=0, s_2\neq 0$), the \textit{pooling} method fails to control the FDR, and the FDR increases as $s_2$ increases. The \textit{intersection} method still controls the FDR and has slightly lower power than the proposed \textit{simultaneous} method. When non-mutual signals are present in both groups ($s_1=s_2\neq 0$), we observe that with $s_1=s_2$ increases, the \textit{intersection} method does not control the FDR. In this case, only the \textit{simultaneous} method successfully controls the FDR with very small power loss compared with the other two methods. The power does not change as the difference between the signal strengths from the two experiments increases. In general, for Scenario 1, where both the strengths and the directions (positive or negative) of mutual signals are the same, the \textit{simultaneous} and the \textit{intersection} methods have similar power, which is only slightly lower than the \textit{pooling} method. For Scenario 2, there is some power loss with our proposed \textit{simultaneous} method when compared with the \textit{pooling} method; the loss is small, suggesting that gaining communication efficient property does not cost us much.

\subsubsection*{C.3.2 Additional simulation results for binary outcomes with $K=2$\\}
In Figures \ref{fig:bin_samesig=1} and \ref{fig:bin_samesig=0} we show results for Setting 2 (binary) with Scenario 1 (same signal strengths) and Scenario 2 (different signal strengths).

\begin{figure}[!p]
    \centering
    \includegraphics[scale=0.3]{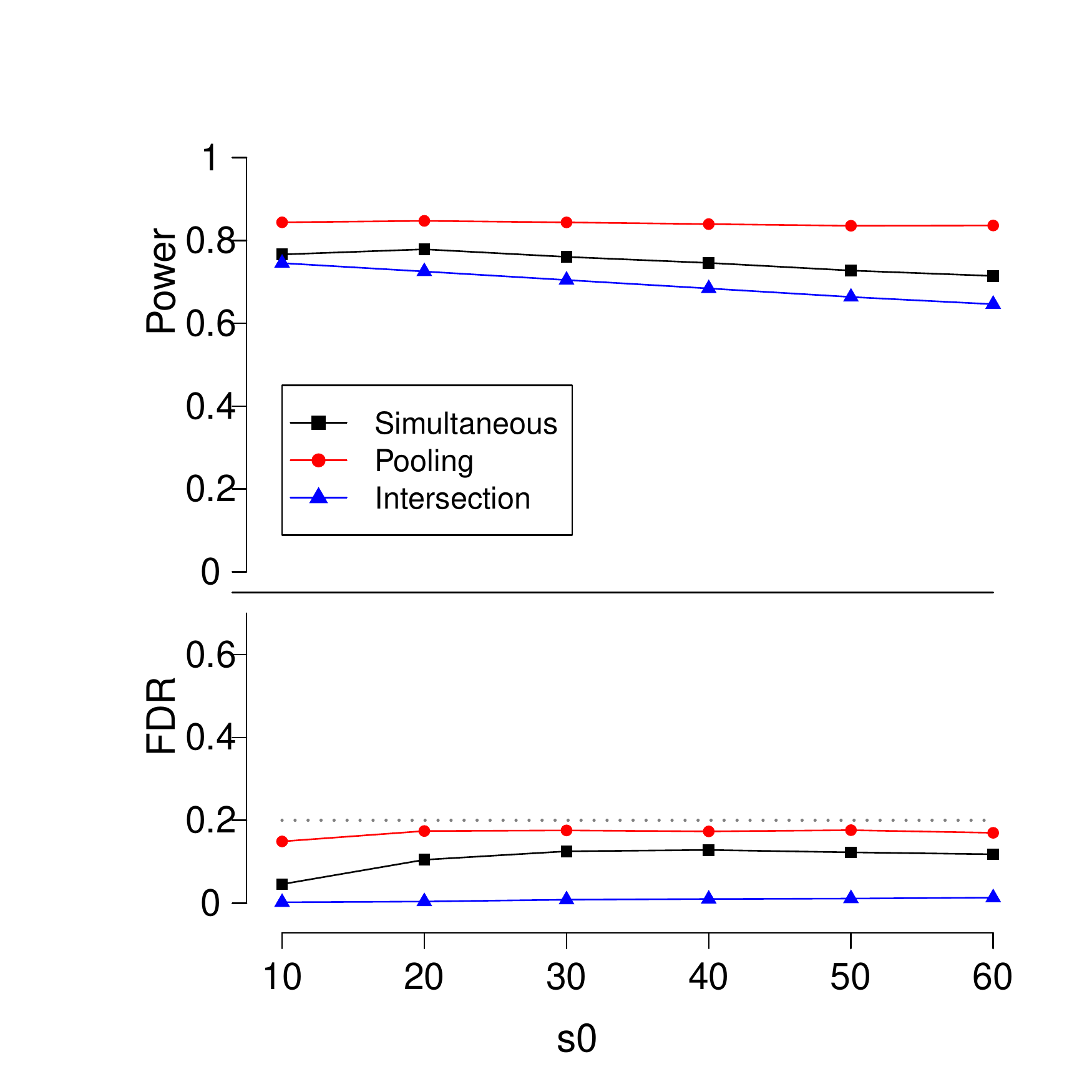}
    \includegraphics[scale=0.3]{Binary_s1_s2.pdf}
    \includegraphics[scale=0.3]{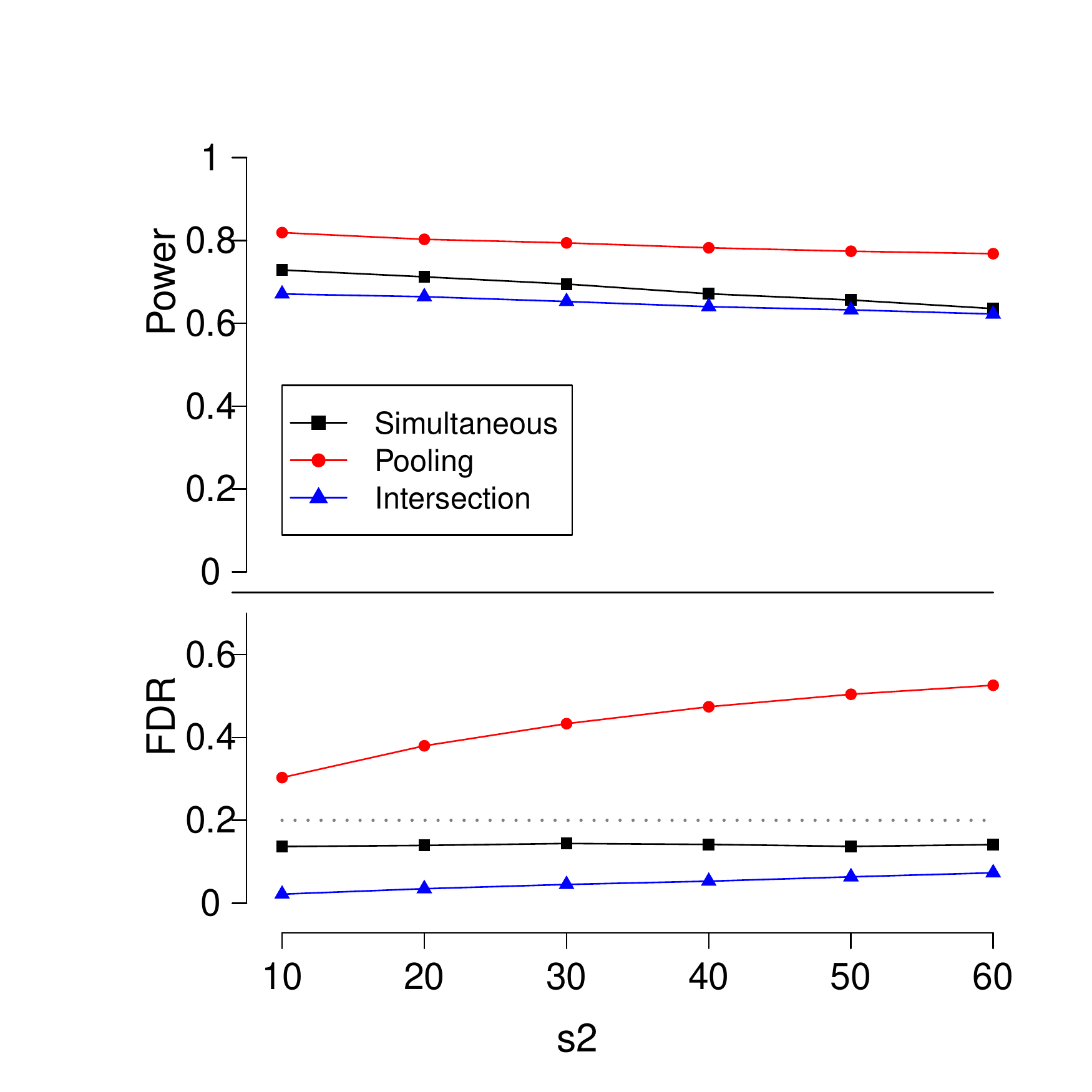}
     \includegraphics[scale=0.3]{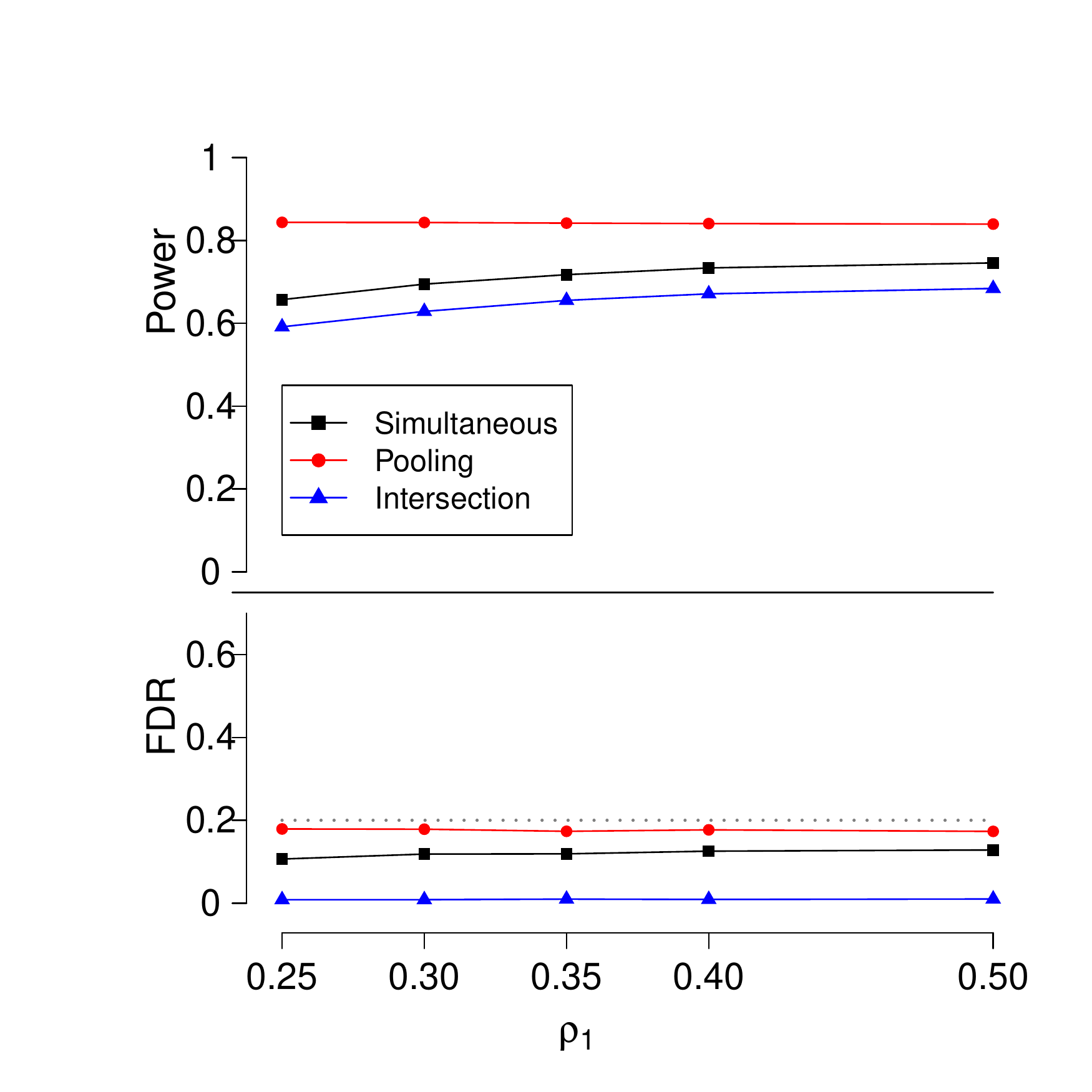}
    \includegraphics[scale=0.3]{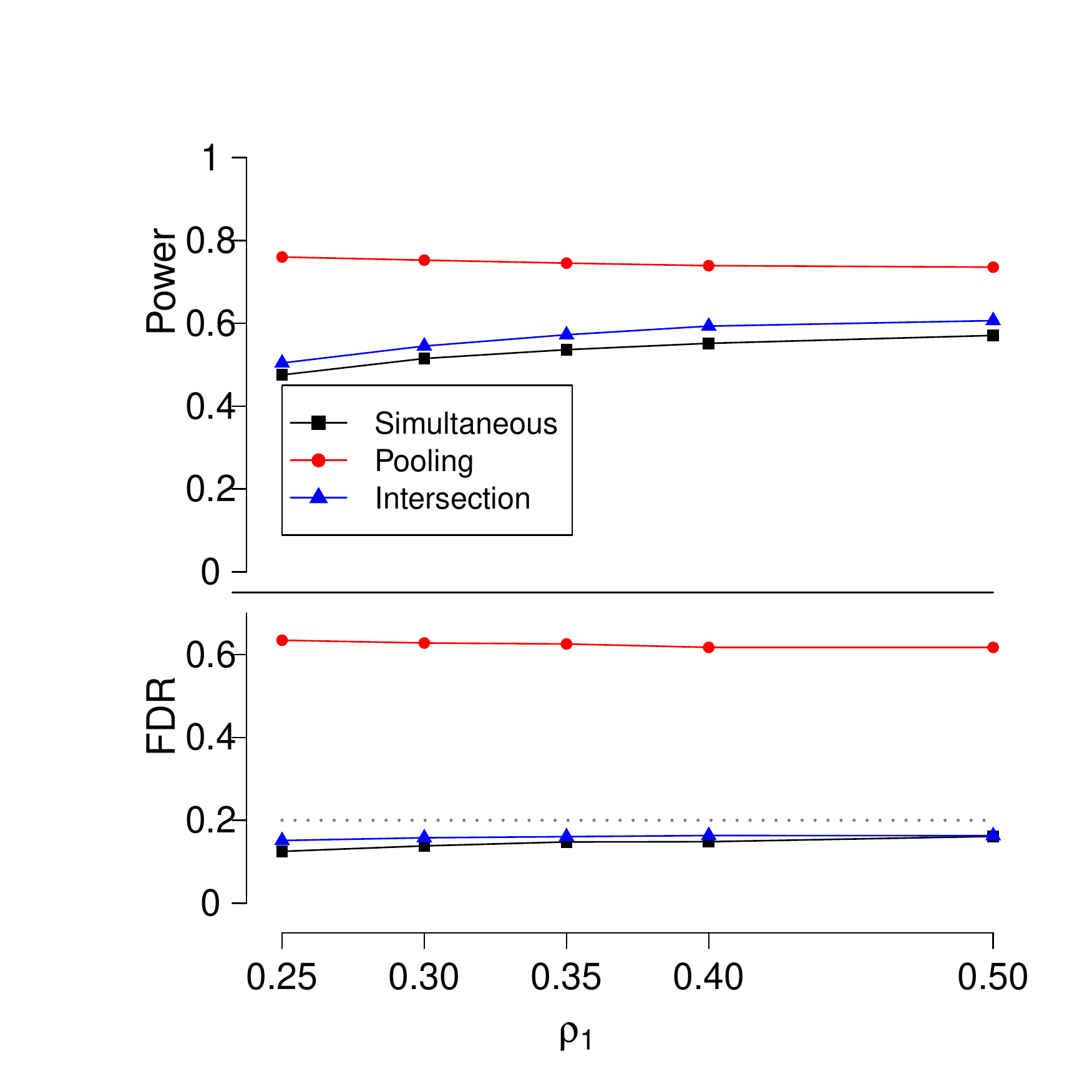}
    \includegraphics[scale=0.3]{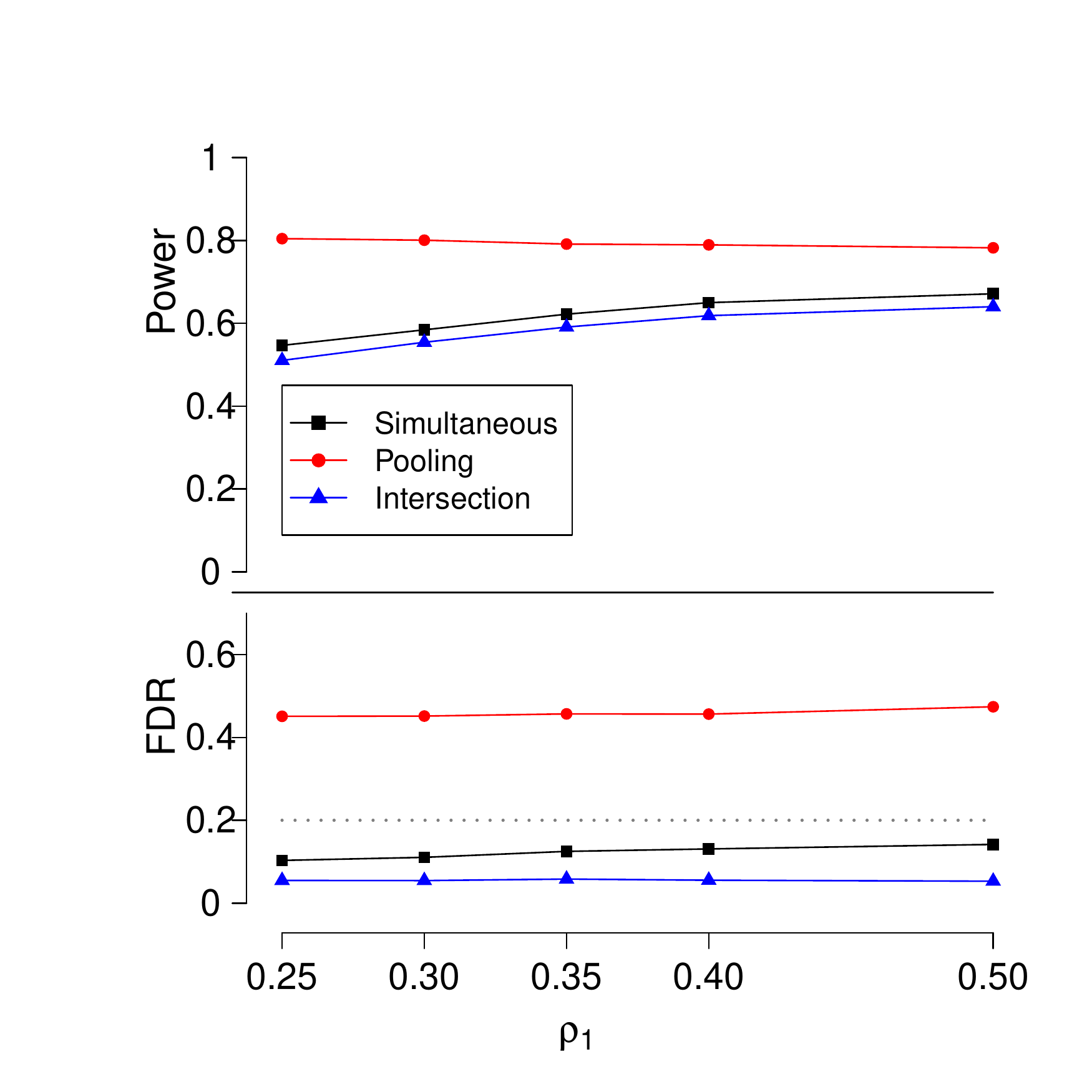}
    \includegraphics[scale=0.3]{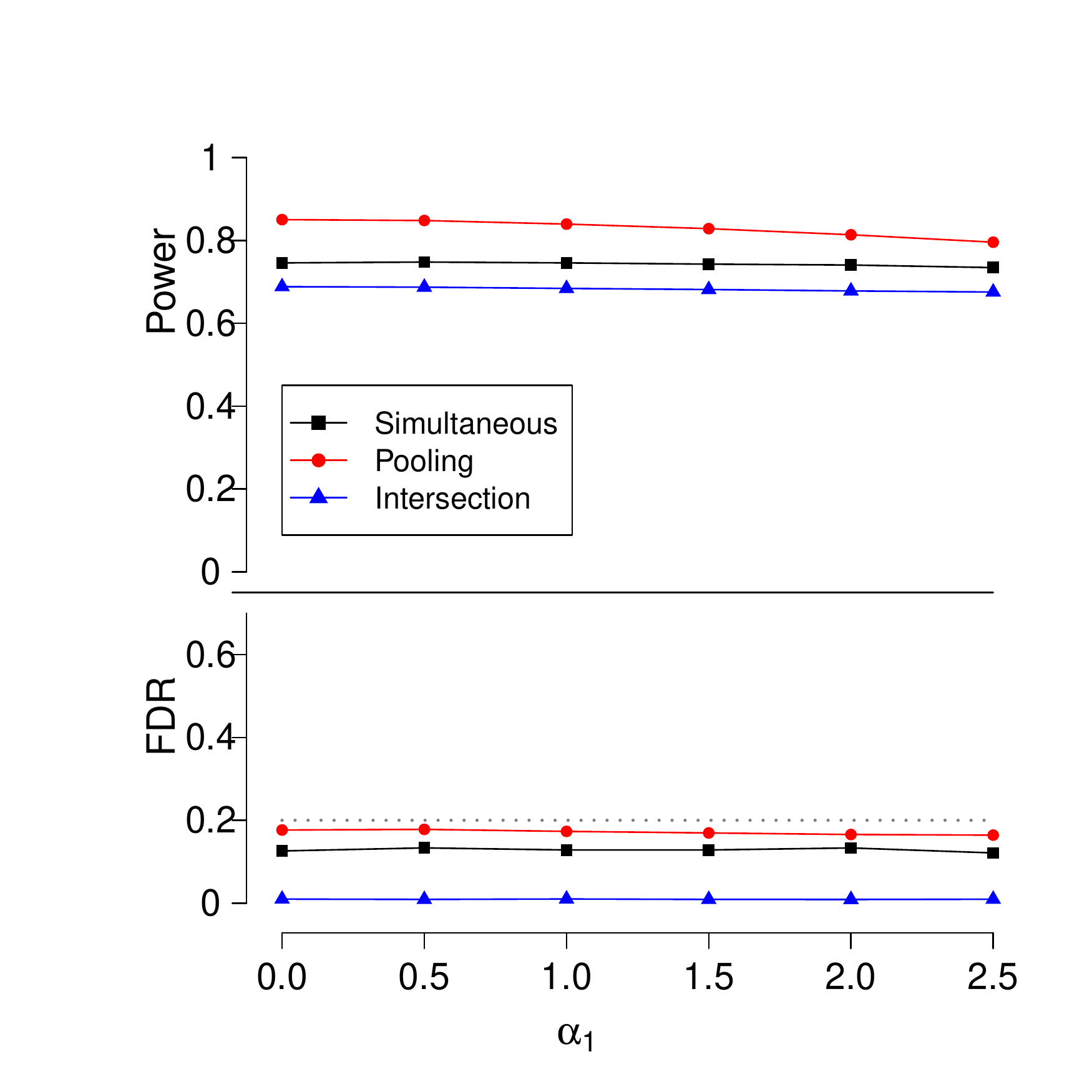}
    \includegraphics[scale=0.3]{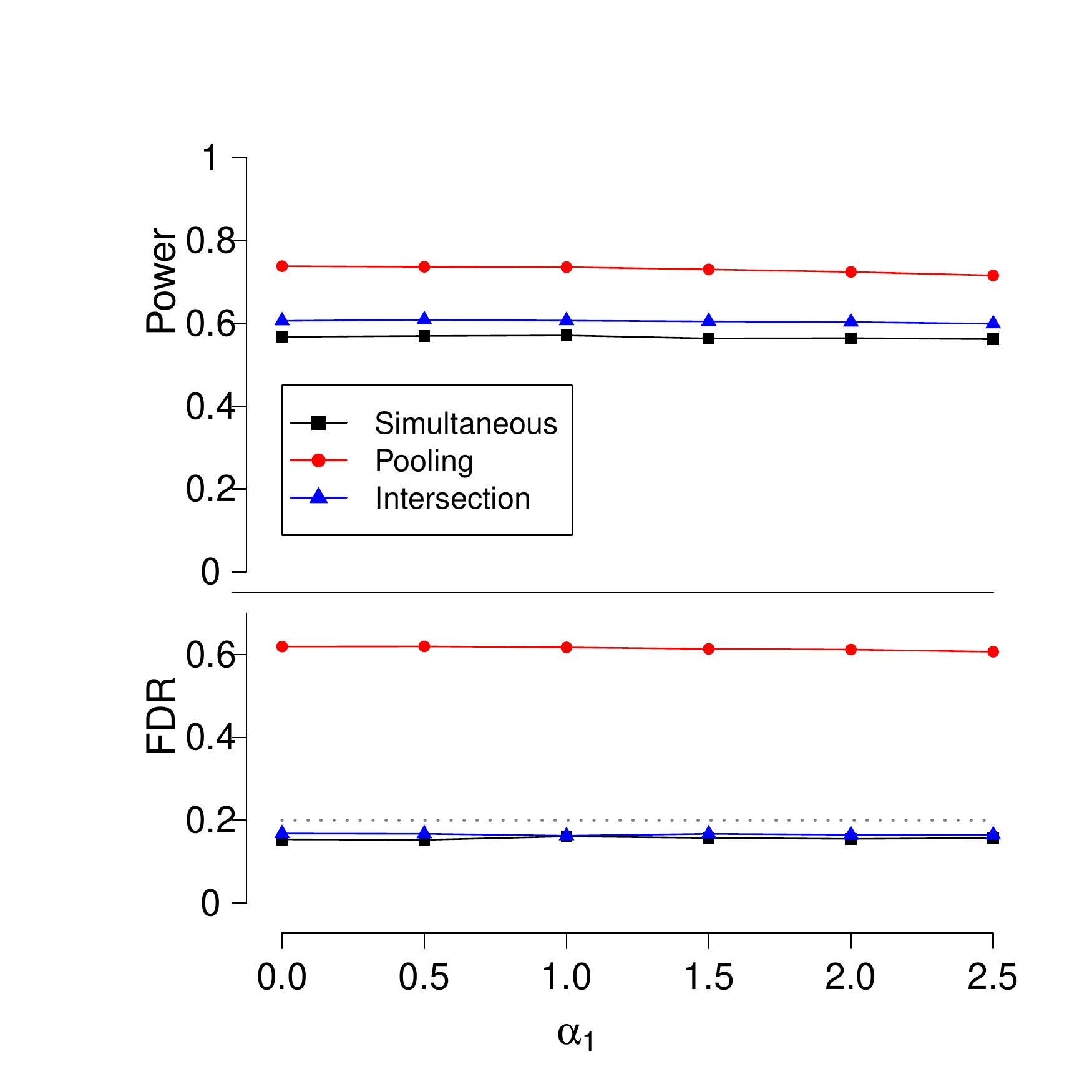}
    \includegraphics[scale=0.3]{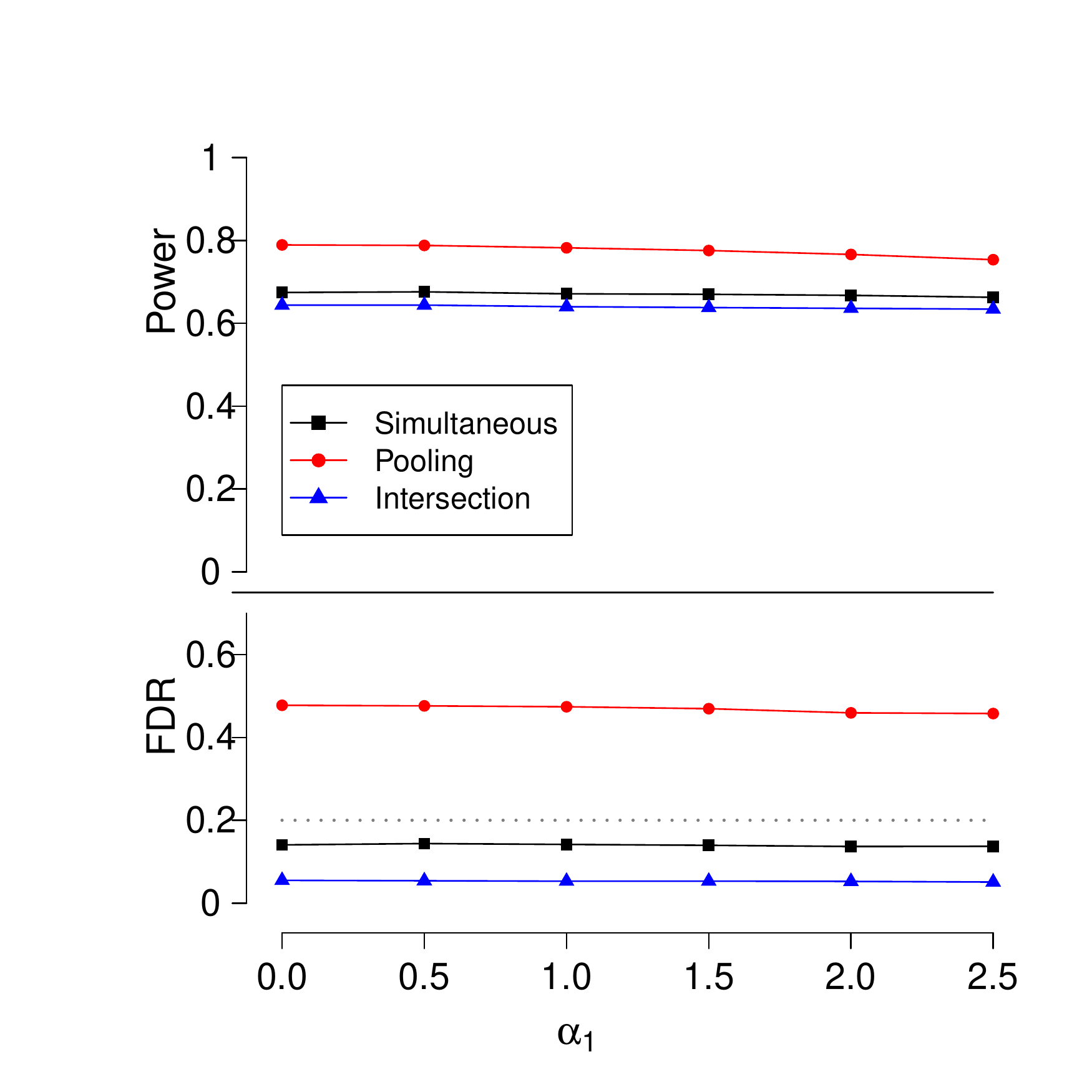}
    \caption{The power and the FDR for simulations with Setting 2 (binary  $K=2$) and Scenario 1 data. Column 1 includes the settings with $s_1=s_2=0$, column 2 includes the settings with $s_1=s_2\neq 0$, column 3 includes the settings with $s_1\neq 0, s_2=0$. Row 1 shows the experiments varying $s_0, s_1, s_2$, row 2 shows experiments varying $\rho_1, \rho_2$, row 3 shows experiments varying $\alpha_1, \alpha_2$.}
    \label{fig:bin_samesig=1}
\end{figure}

\begin{figure}[!p]
    \centering
    \includegraphics[scale=0.3]{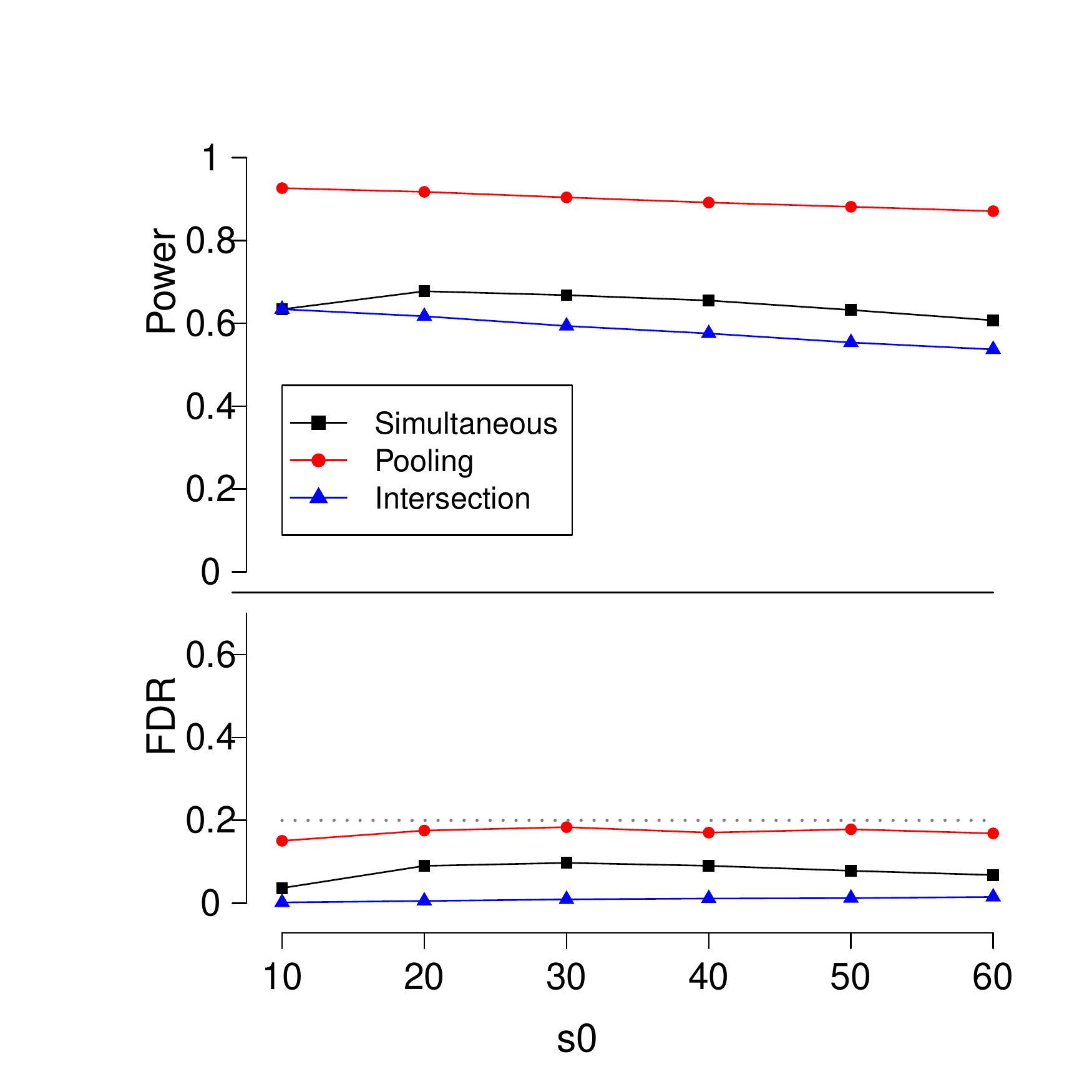}
    \includegraphics[scale=0.3]{Binary_s1_s2samesig=0.pdf}
    \includegraphics[scale=0.3]{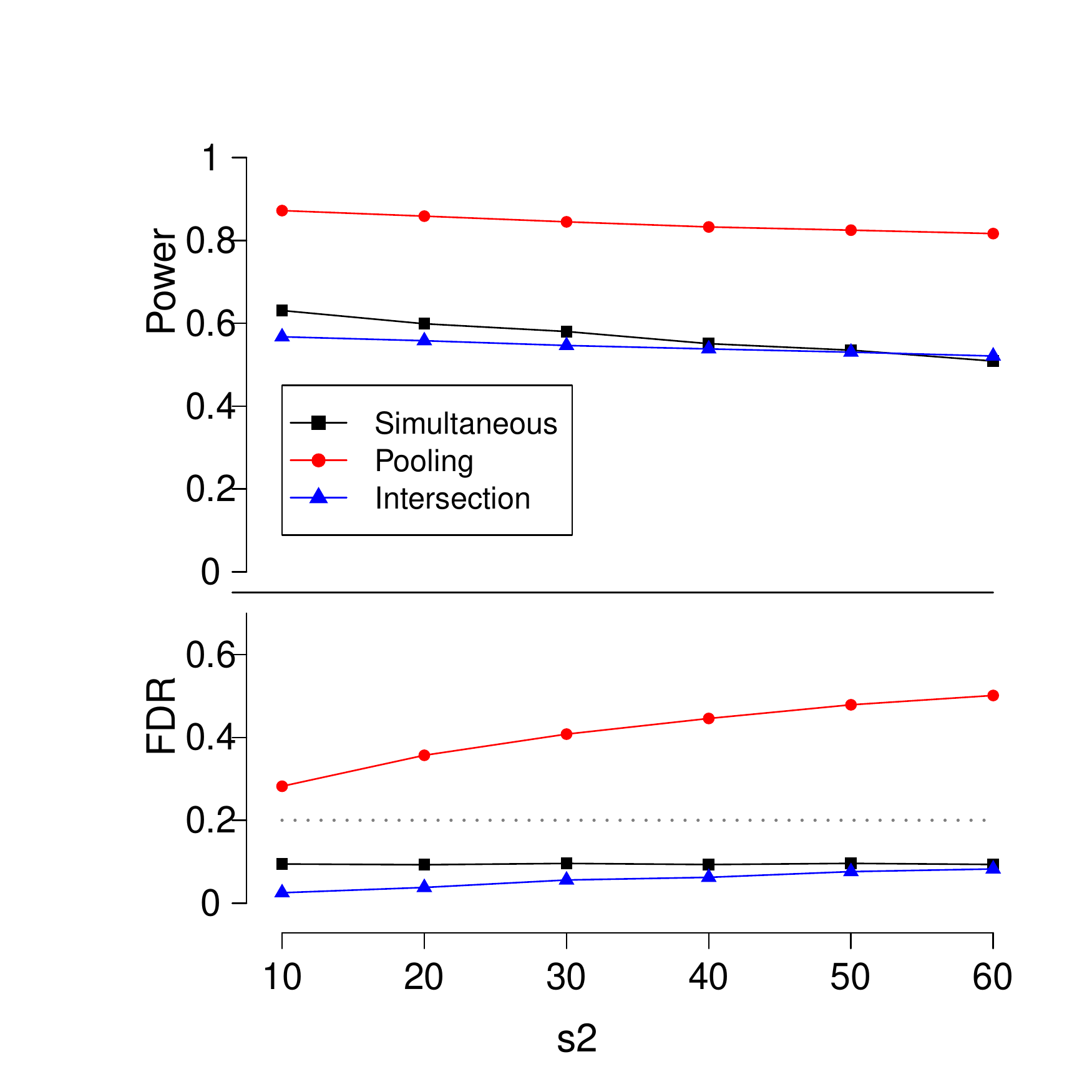}
     \includegraphics[scale=0.3]{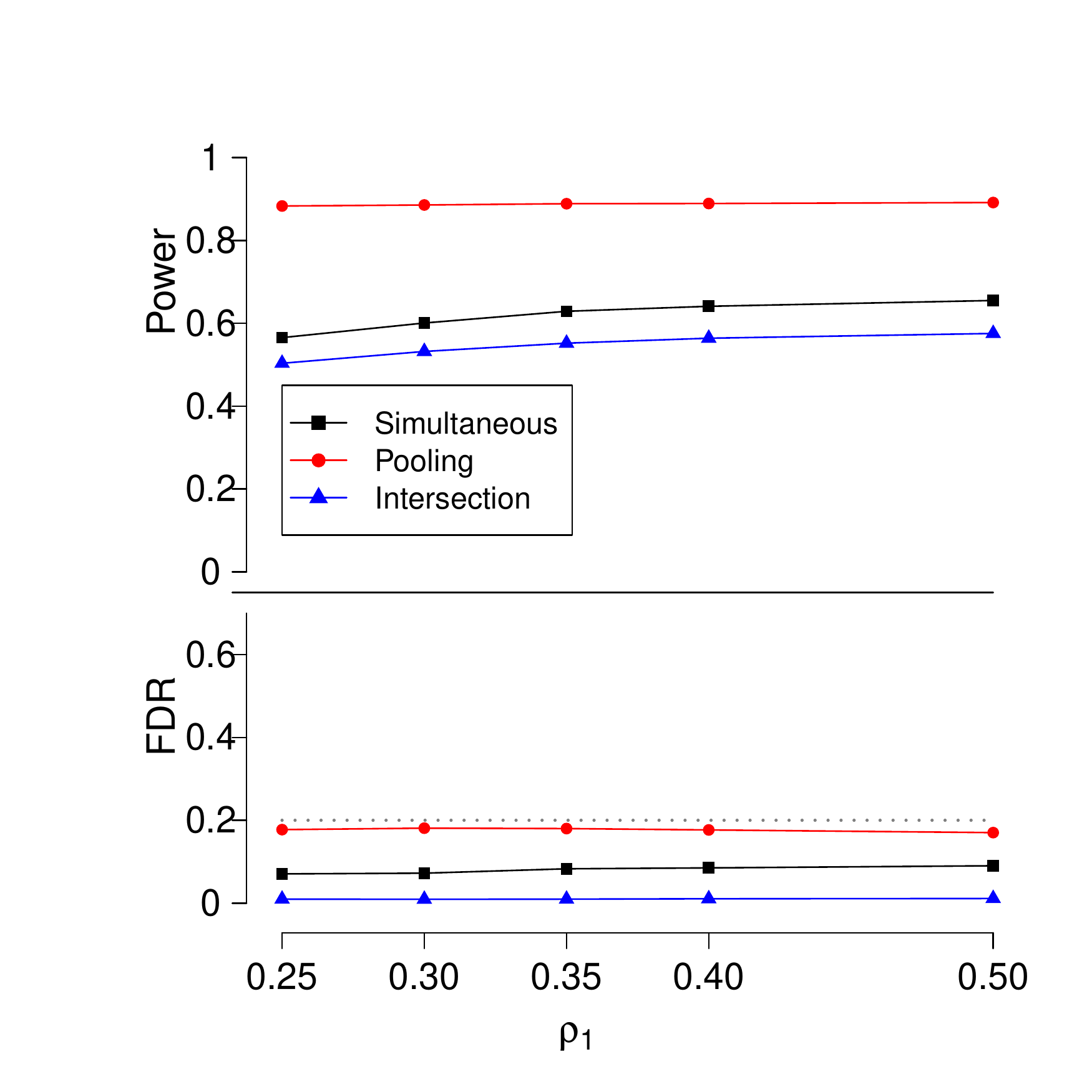}
    \includegraphics[scale=0.3]{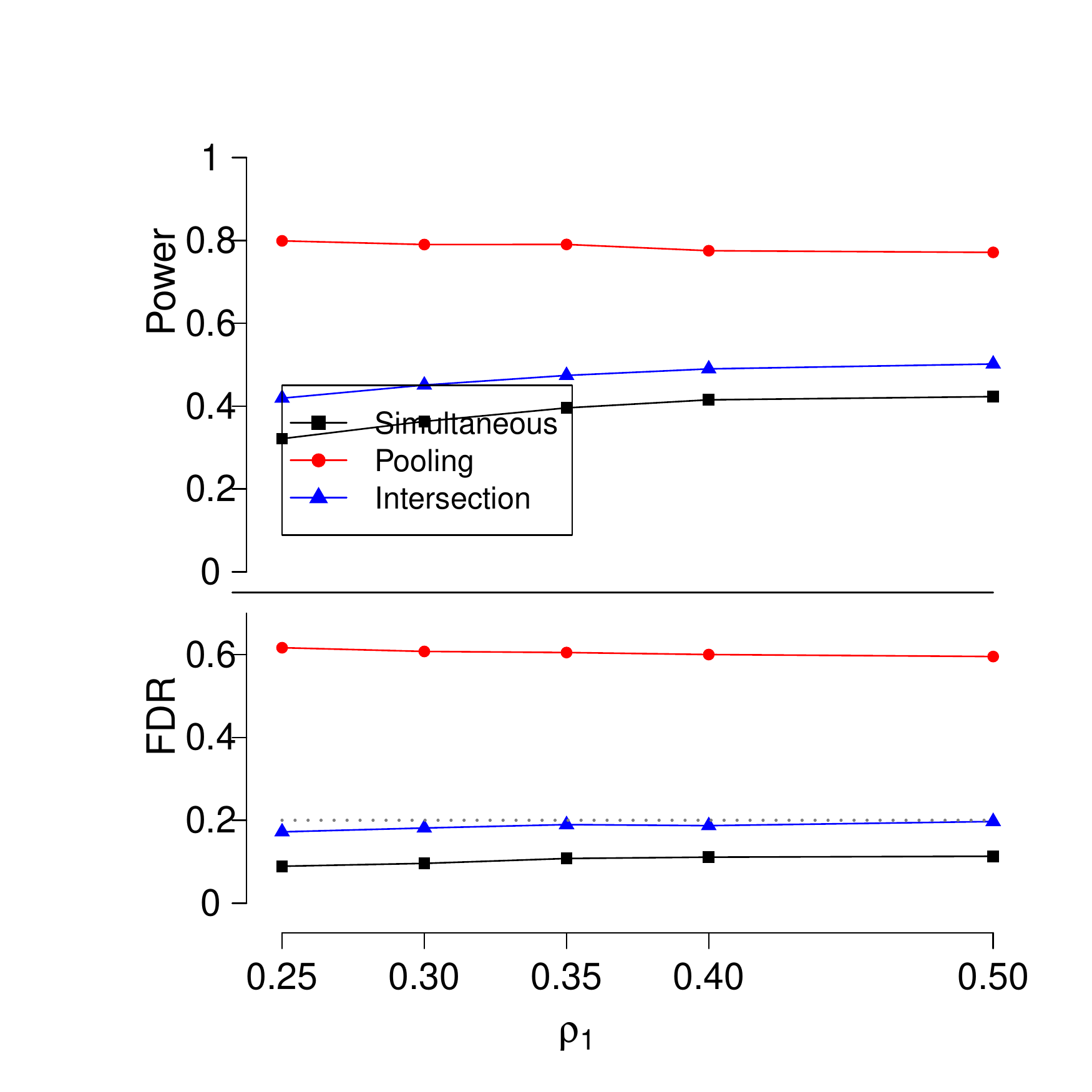}
    \includegraphics[scale=0.3]{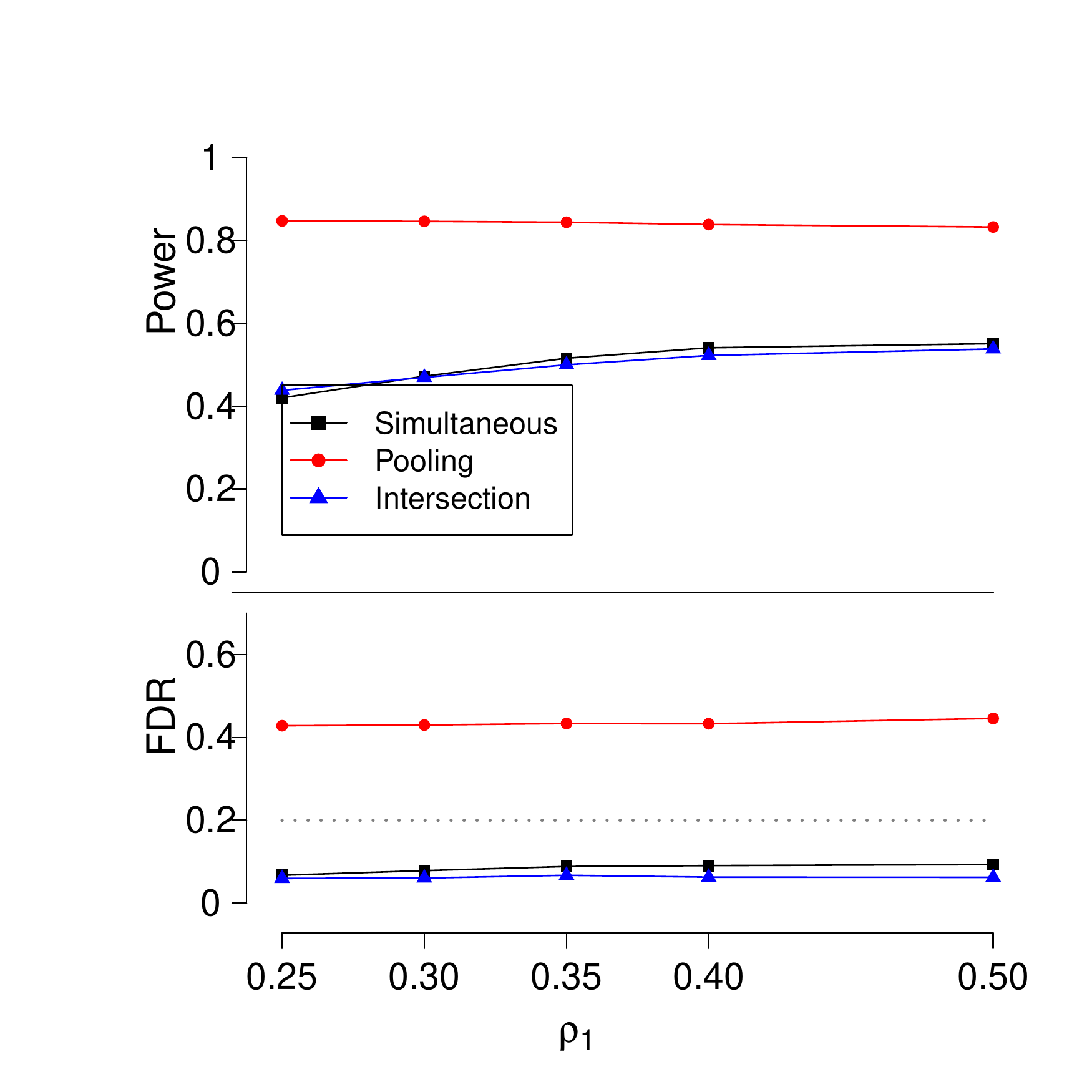}
    \includegraphics[scale=0.3]{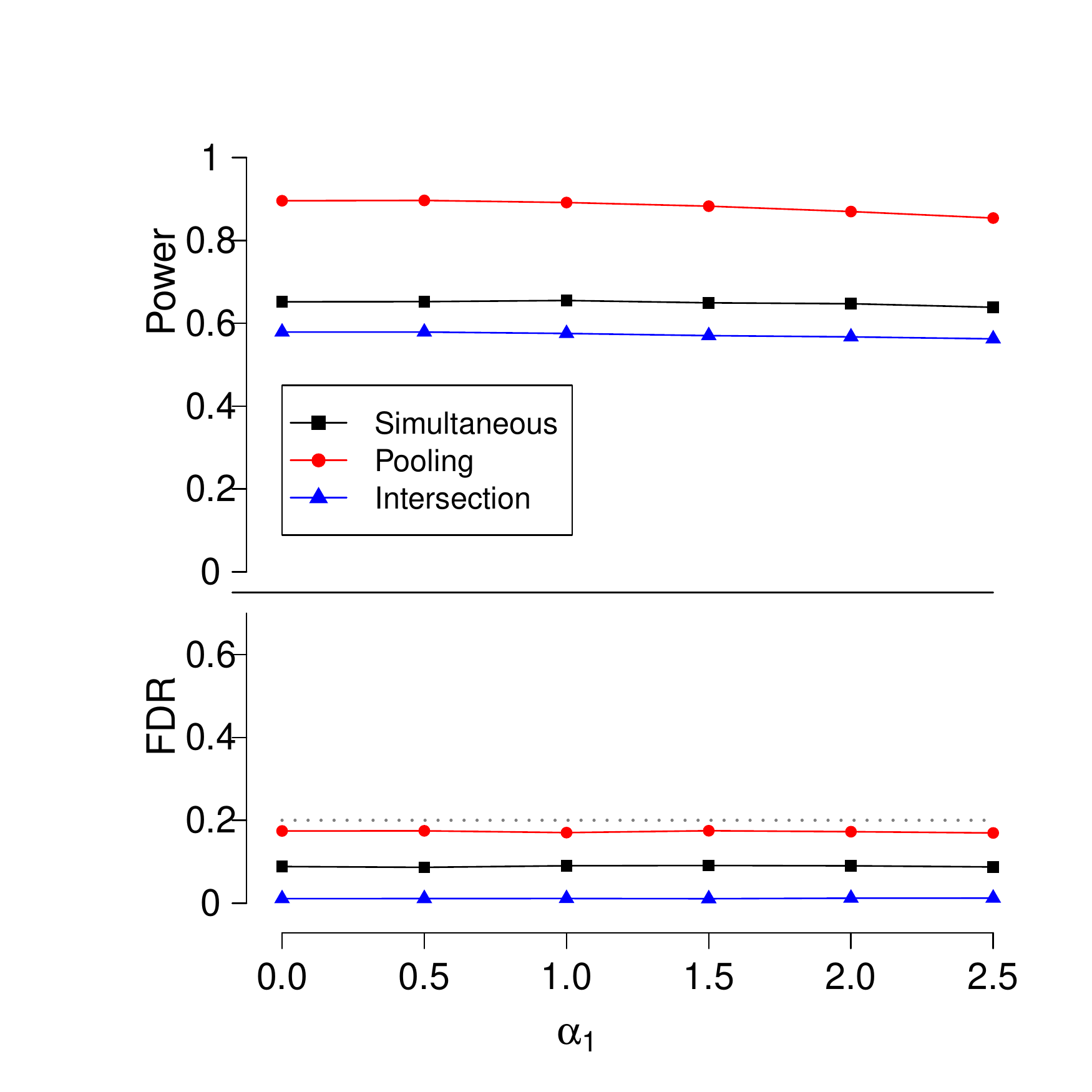}
    \includegraphics[scale=0.3]{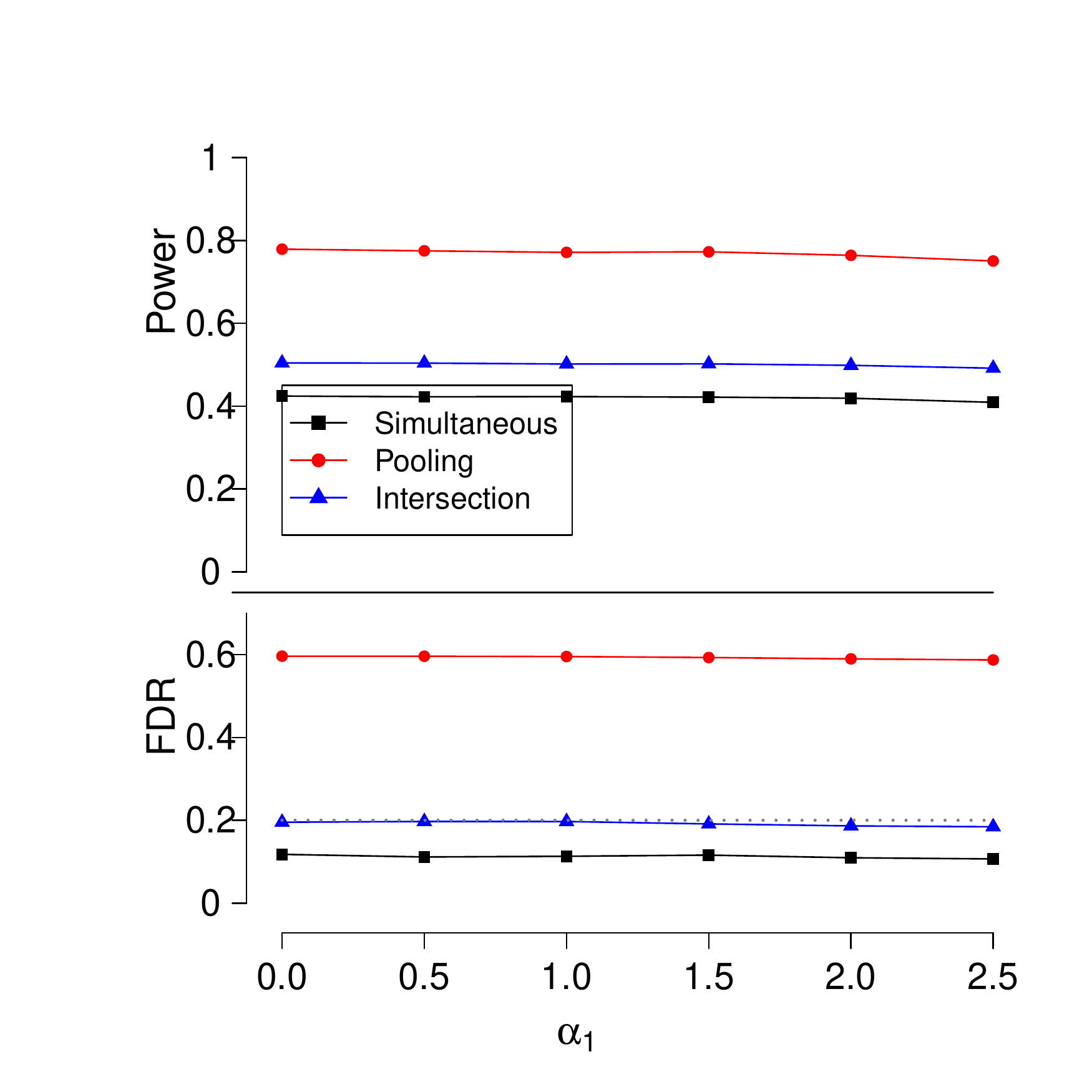}
    \includegraphics[scale=0.3]{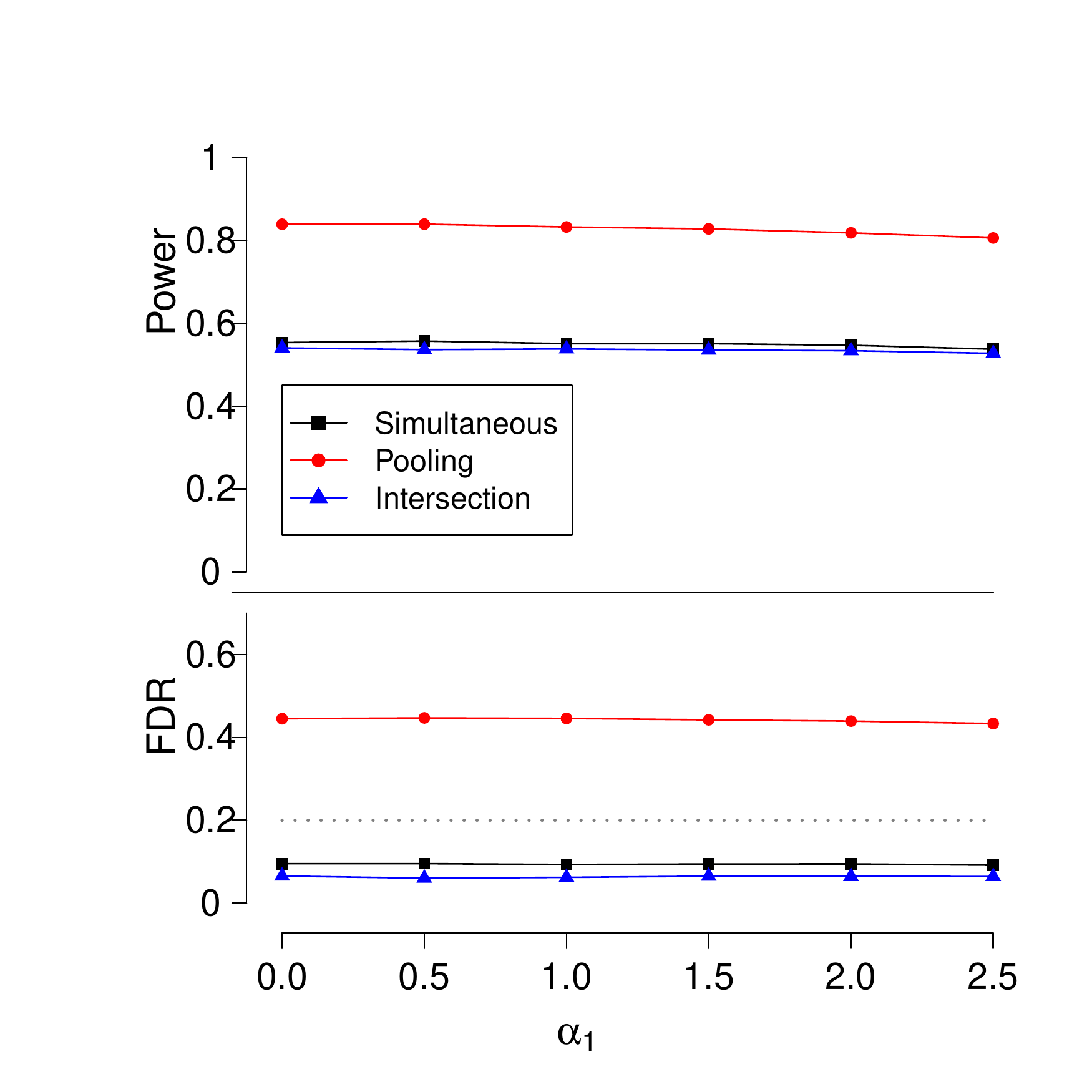}
    \caption{The power and the FDR for simulations with Setting 2 (binary  $K=2$) and Scenario 2 data. Column 1 includes the settings with $s_1=s_2=0$, column 2 includes the settings with $s_1=s_2\neq 0$, column 3 includes the settings with $s_1\neq 0, s_2=0$. Row 1 shows the experiments varying $s_0, s_1, s_2$, row 2 shows experiments varying $\rho_1, \rho_2$, row 3 shows experiments varying $\alpha_1, \alpha_2$.}
    \label{fig:bin_samesig=0}
\end{figure}

For the binary settings, we see similar results as the continuous setting. When there are only mutual signals present ($s_1=s_2=0$), all three methods control the FDR and the \textit{pooling} method has the best power, and our proposed \textit{simultaneous} method has better power than the \textit{intersection} method. When there exist non-mutual signals in one group ($s_1=0, s_2\neq 0$), the \textit{pooling} method fails to control the FDR, the \textit{intersection} method still controls the FDR, but it has lower power than the proposed \textit{simultaneous} method. When non-mutual signals are present in both groups ($s_1=s_2\neq 0$), we observe that with $s_1=s_2$ increases, the \textit{intersection} method does not control the FDR. In this case, only the \textit{simultaneous} method successfully controls the FDR. The power loss of the \textit{simultaneous} method compared with the \textit{pooling} method is larger than the continuous settings but is still acceptable. The performance does not change much with varying $\rho_1,\rho_2$ or $\alpha_1, \alpha_2$.

\subsubsection*{C.3.3 Additional simulation results for mixed outcomes with $K=2$\\}
In Figures \ref{fig:mix_samesig=1} and \ref{fig:mix_samesig=0} we show results for Setting 3 (mixed) with Scenario 1 (same signal strengths) and Scenario 2 (different signal strengths).

\begin{figure}[!p]
    \centering
    \includegraphics[scale=0.3]{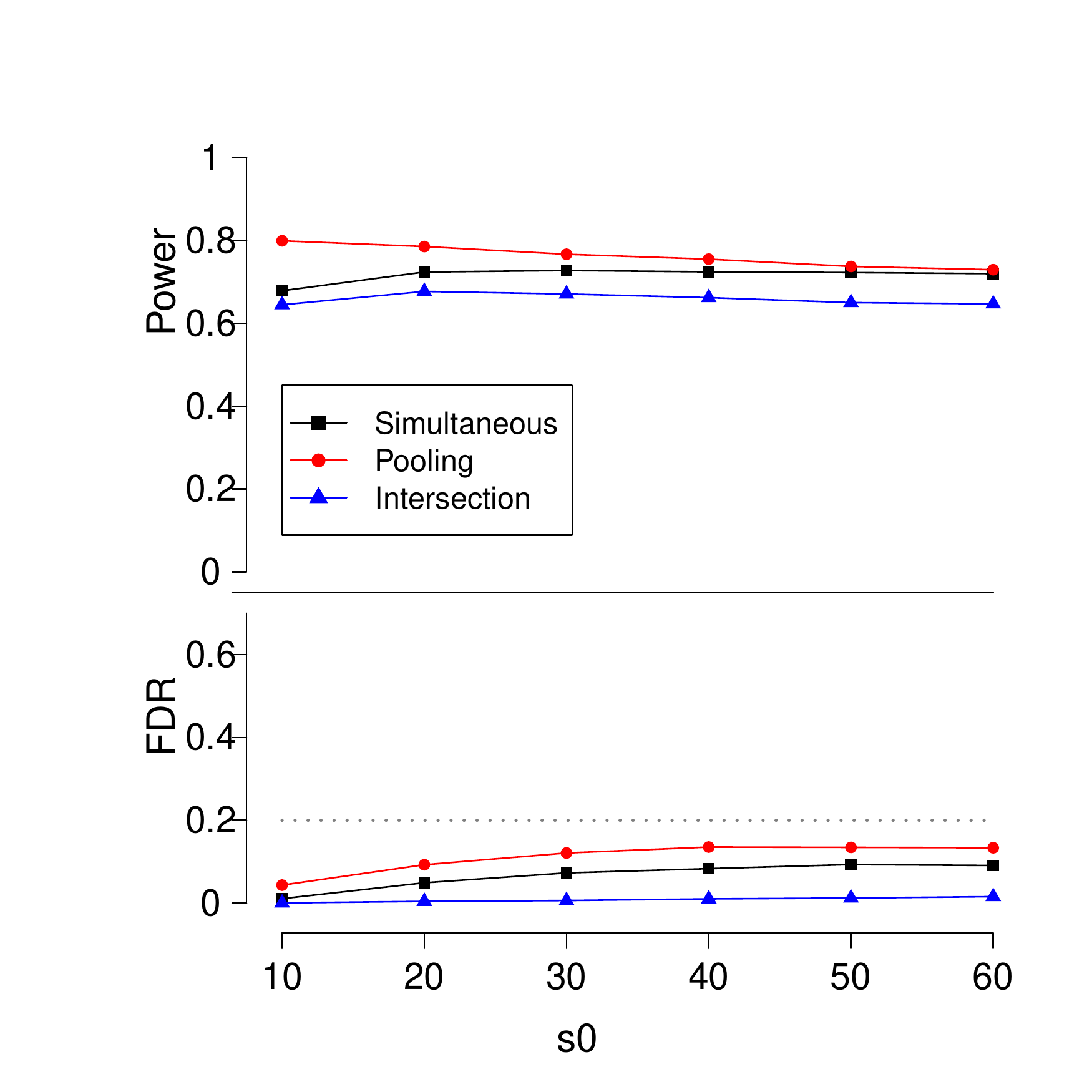}
    \includegraphics[scale=0.3]{Mixed_s1_s2.pdf}
    \includegraphics[scale=0.3]{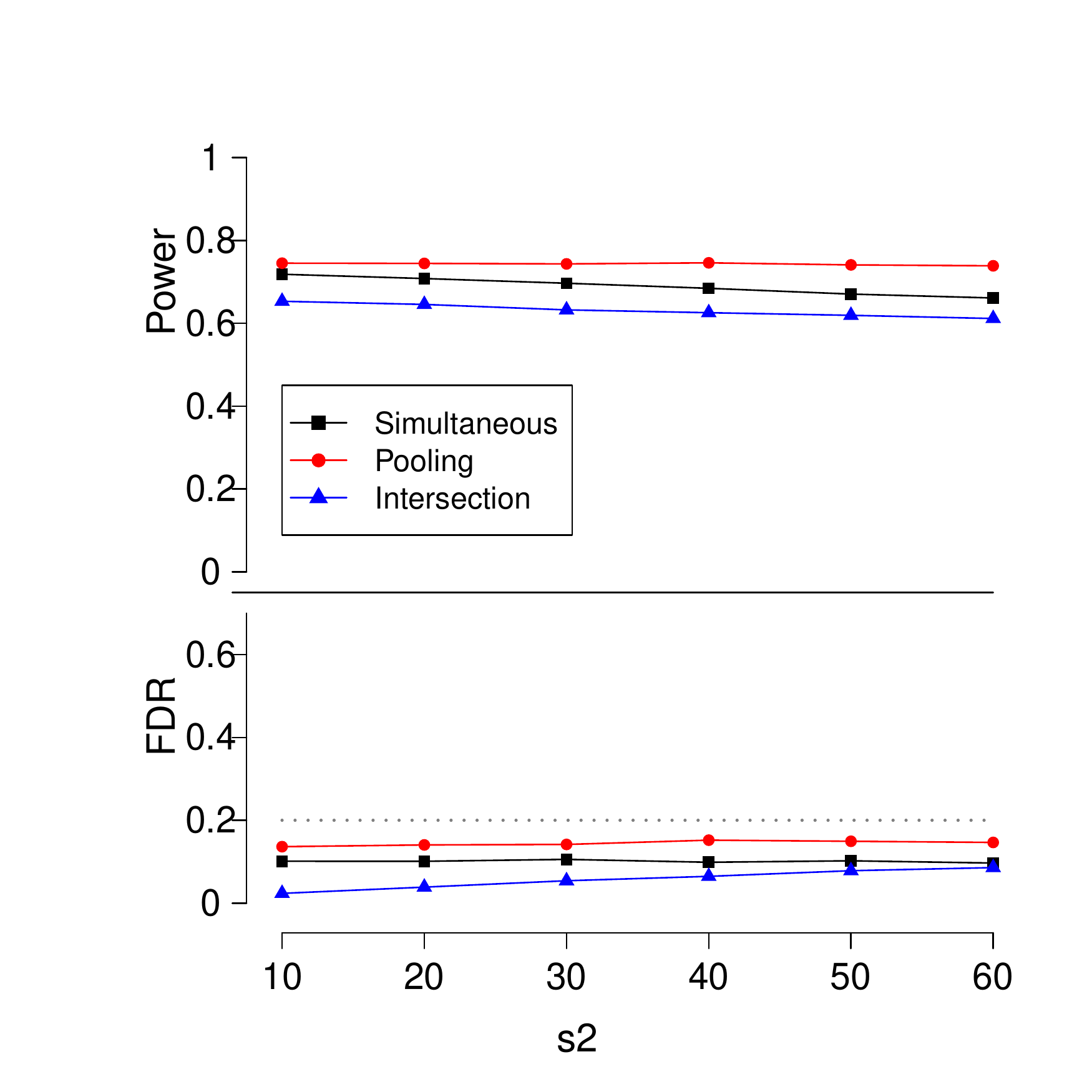}
     \includegraphics[scale=0.3]{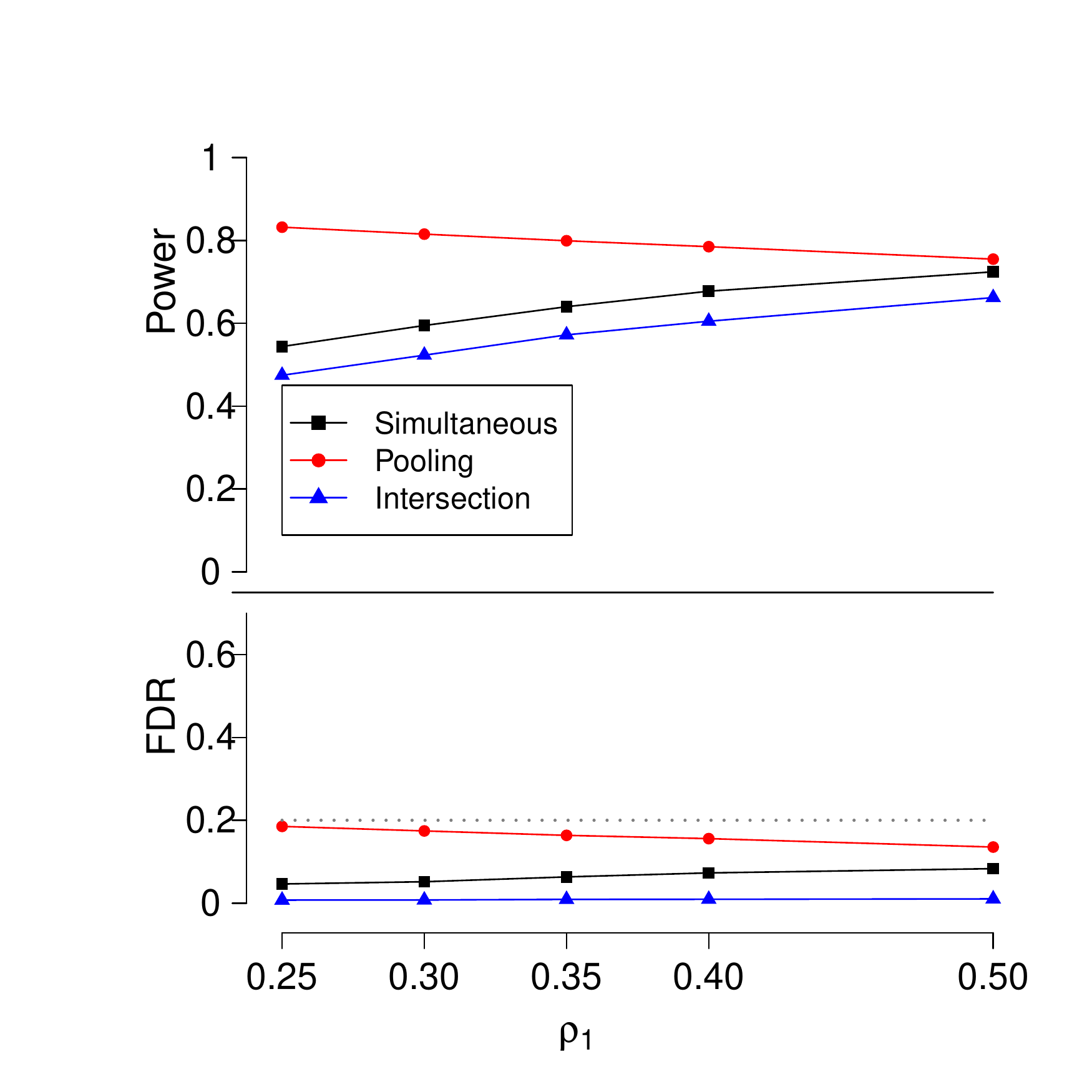}
    \includegraphics[scale=0.3]{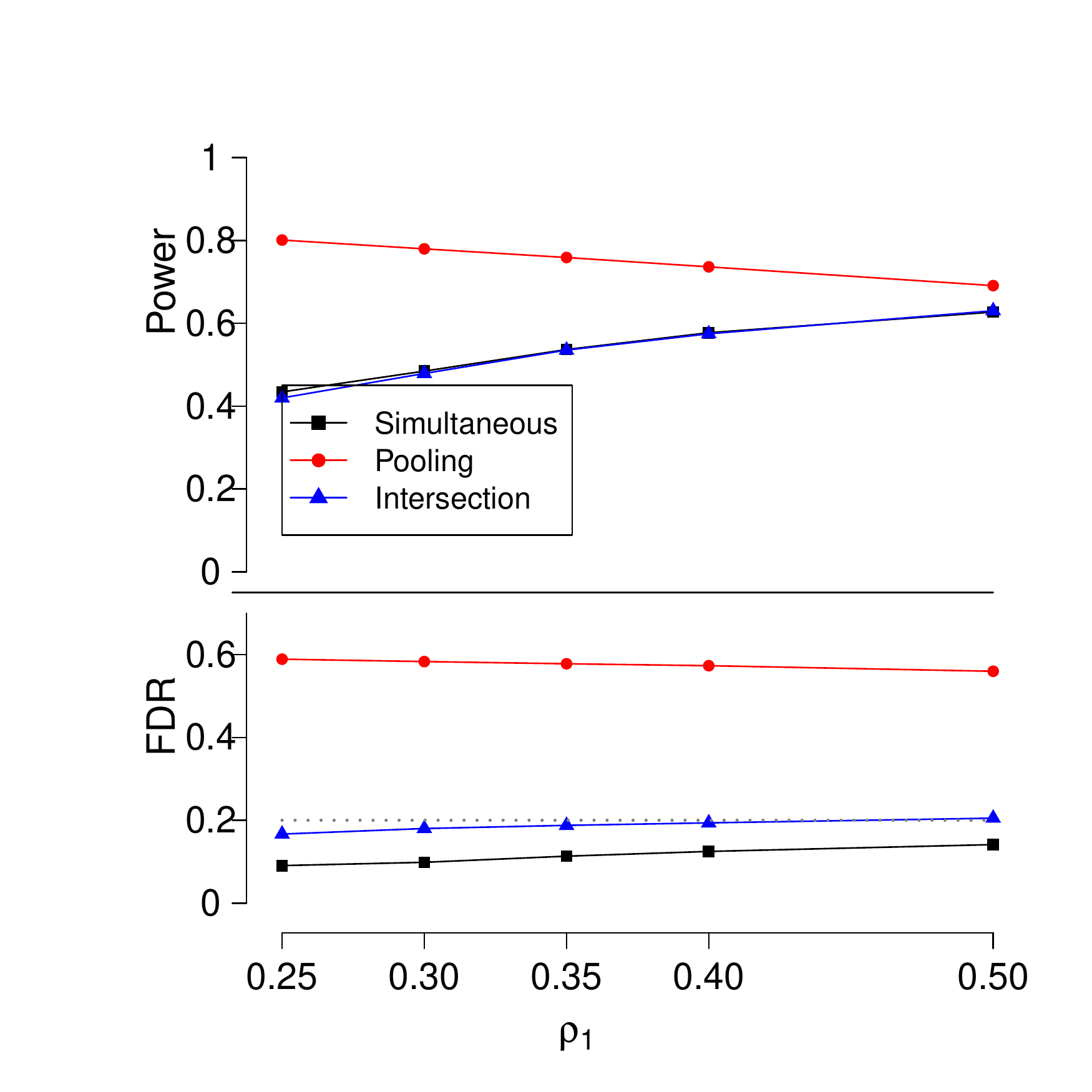}
    \includegraphics[scale=0.3]{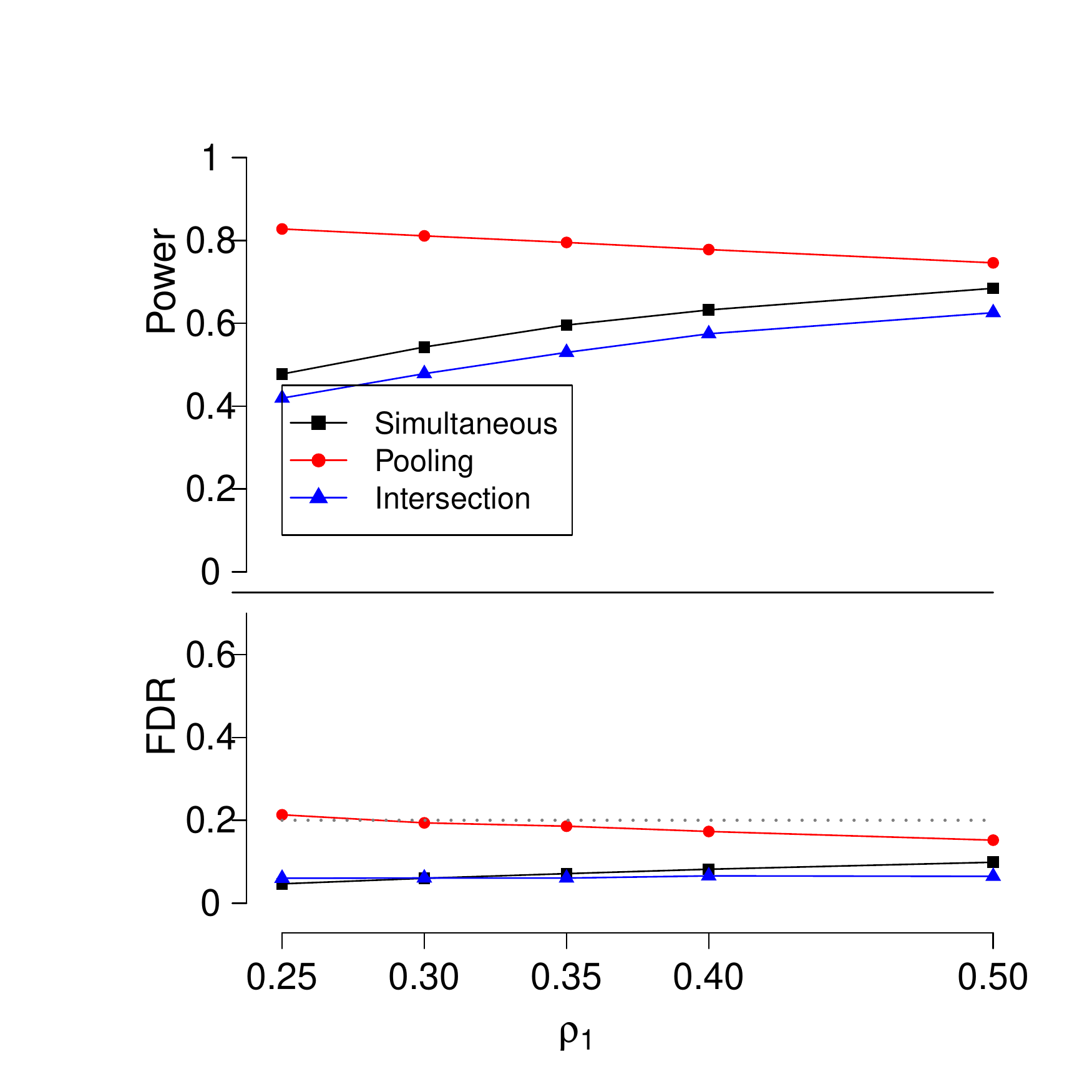}
    \includegraphics[scale=0.3]{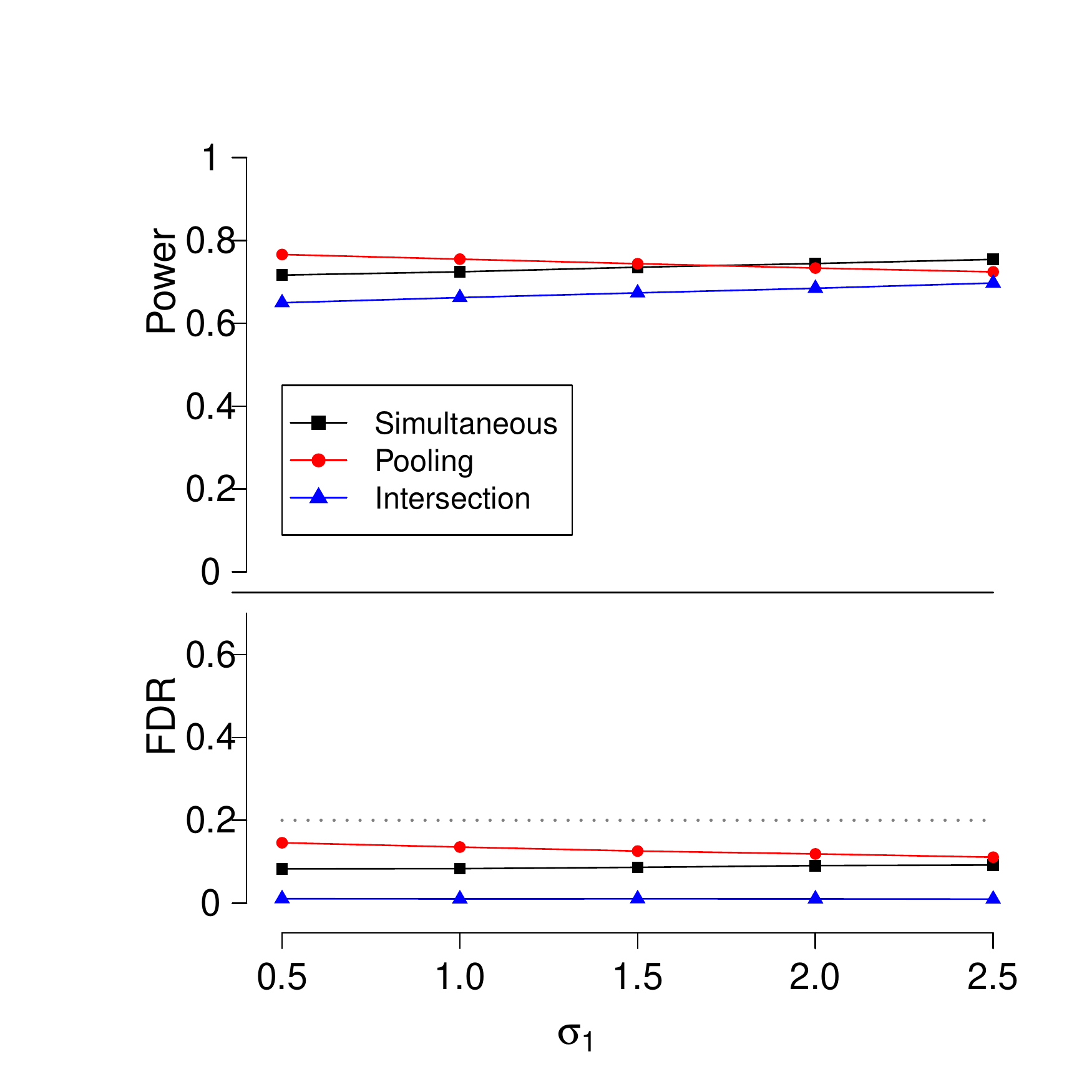}
    \includegraphics[scale=0.3]{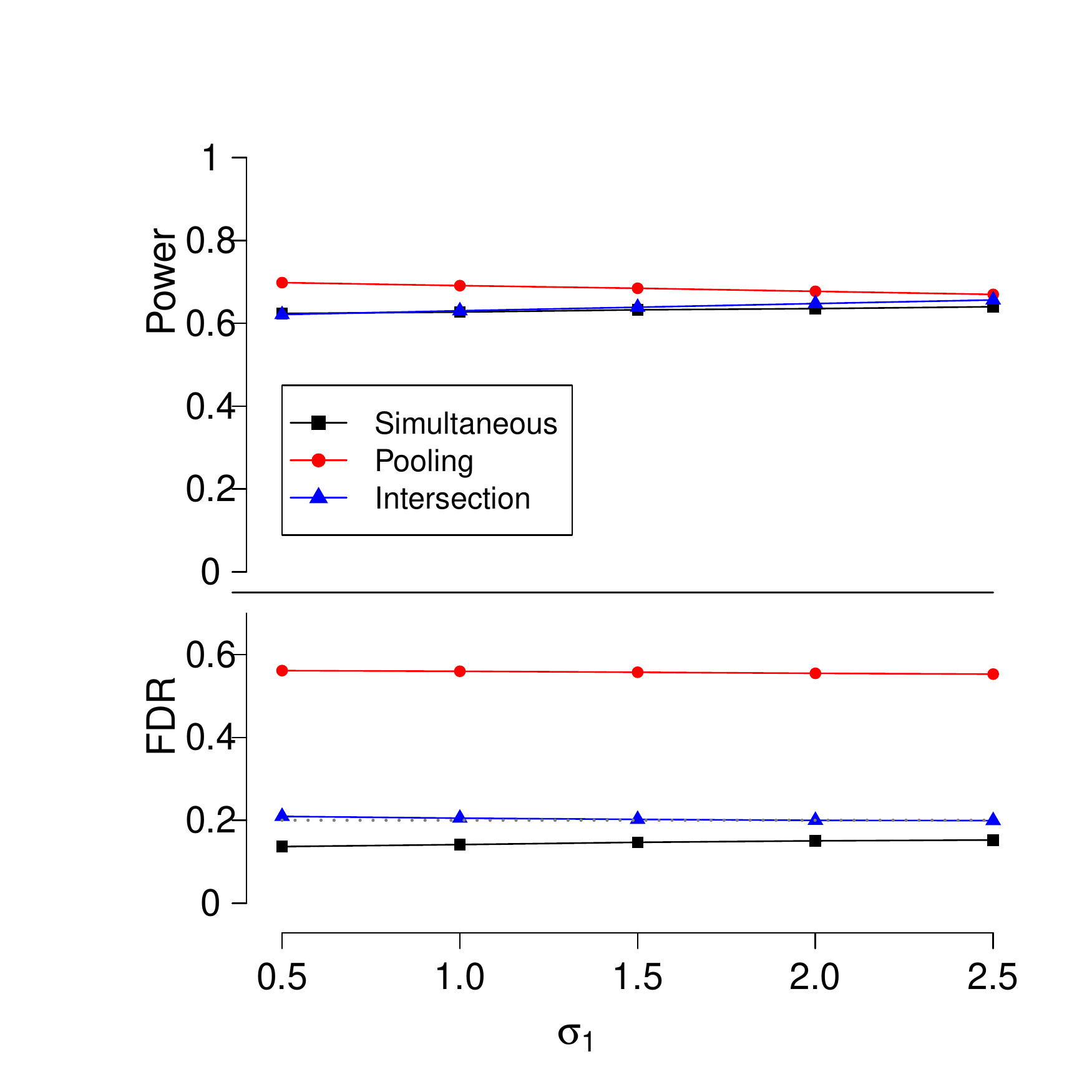}
    \includegraphics[scale=0.3]{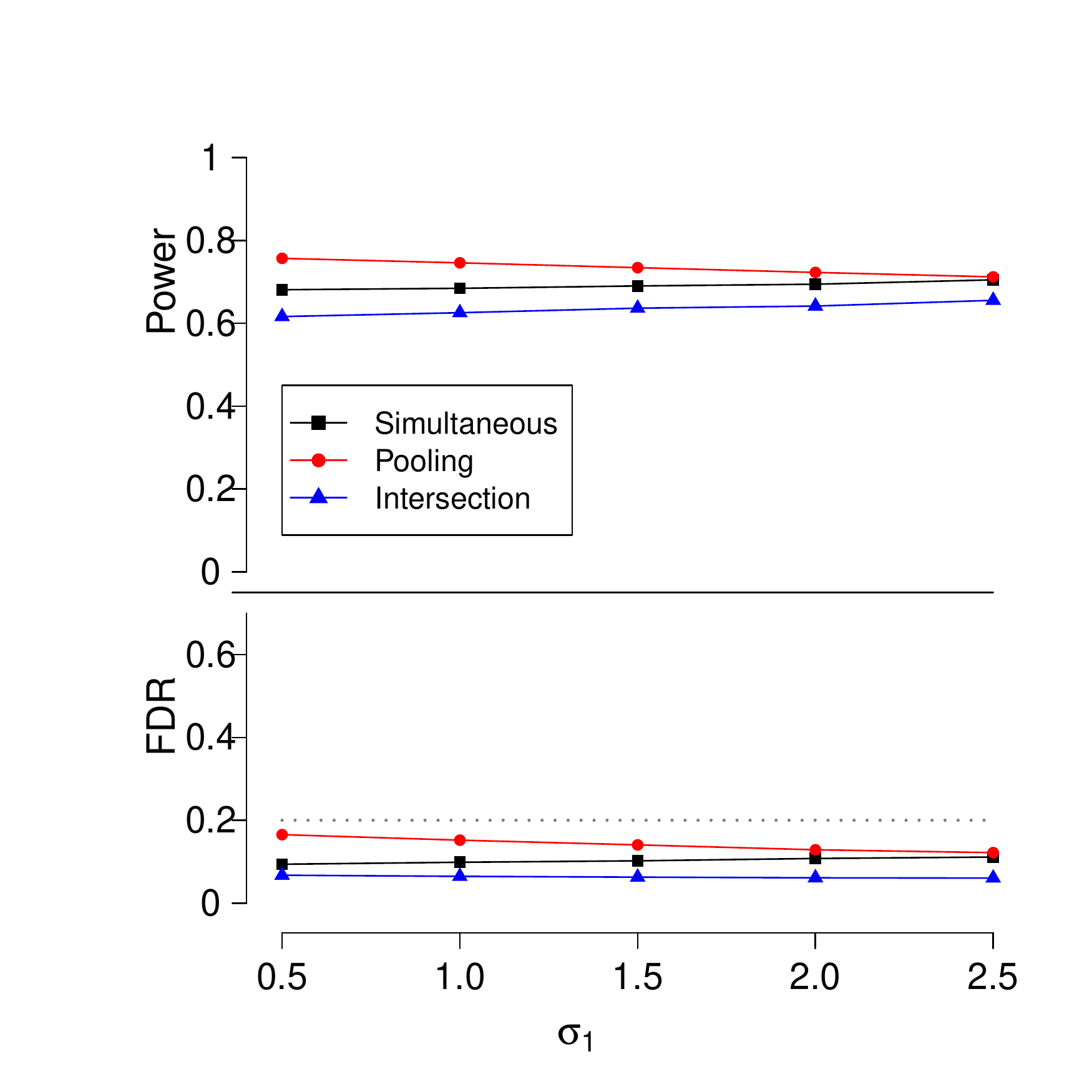}
    \caption{The power and the FDR for simulations with Setting 3 (mixed $K=2$) and Scenario 1 data. Column 1 includes the settings with $s_1=s_2=0$, column 2 includes the settings with $s_1=s_2\neq 0$, column 3 includes the settings with $s_1\neq 0, s_2=0$. Row 1 shows the experiments varying $s_0, s_1, s_2$, row 2 shows experiments varying $\rho_1, \rho_2$, row 3 shows experiments varying $\sigma_1, \sigma_2$.}
    \label{fig:mix_samesig=1}
\end{figure}

\begin{figure}[!p]
    \centering
    \includegraphics[scale=0.3]{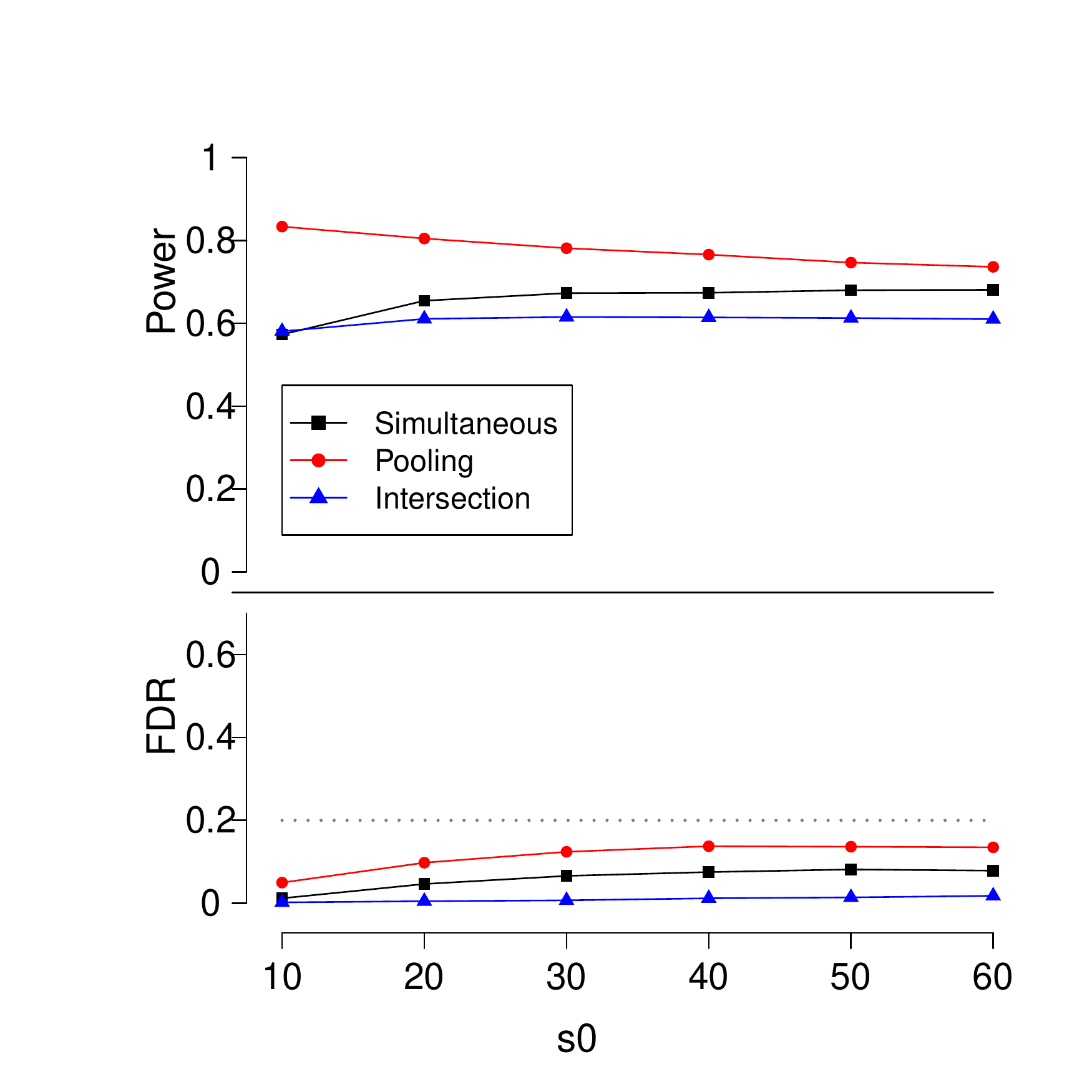}
    \includegraphics[scale=0.3]{Mixed_s1_s2samesig=0.pdf}
    \includegraphics[scale=0.3]{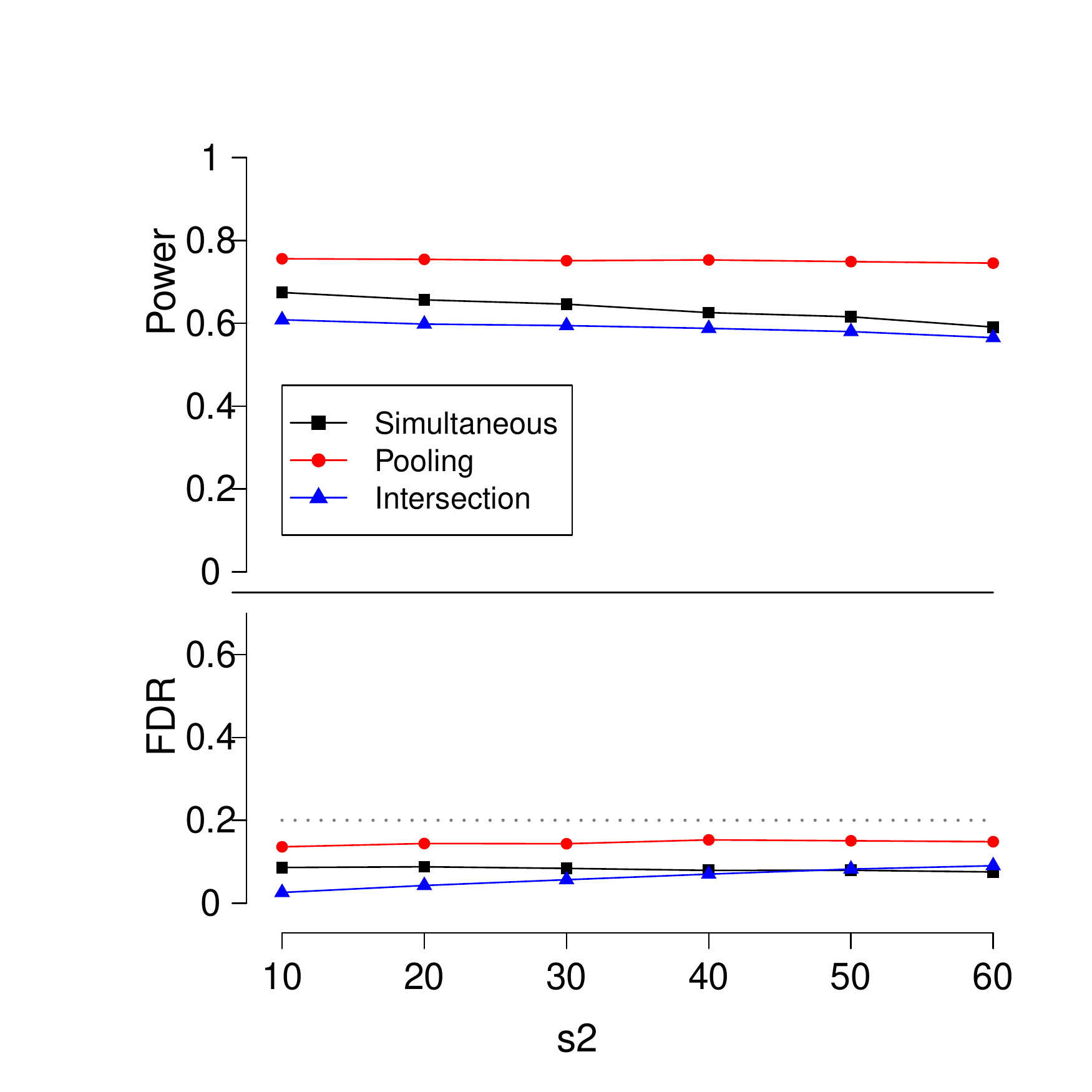}
     \includegraphics[scale=0.3]{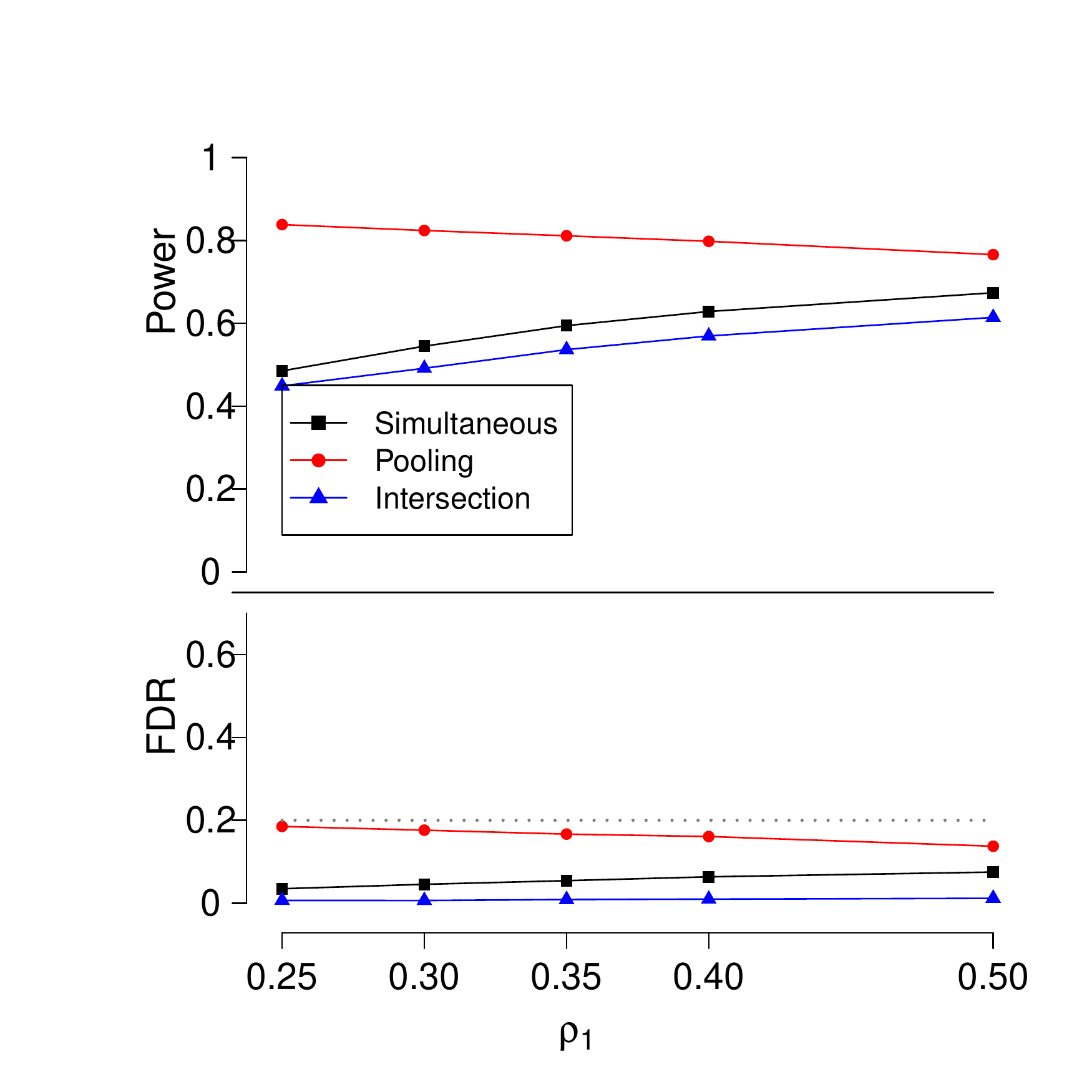}
    \includegraphics[scale=0.3]{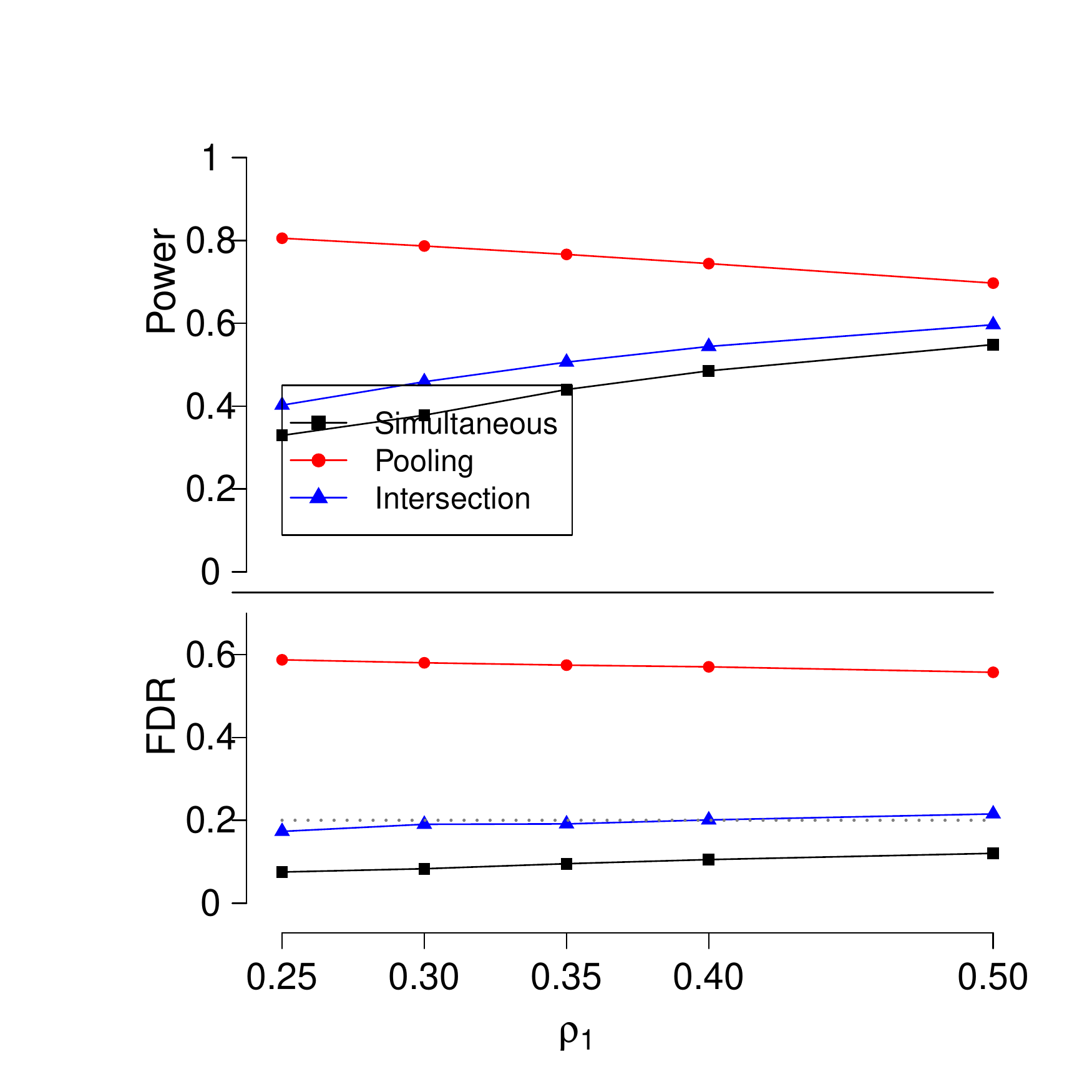}
    \includegraphics[scale=0.3]{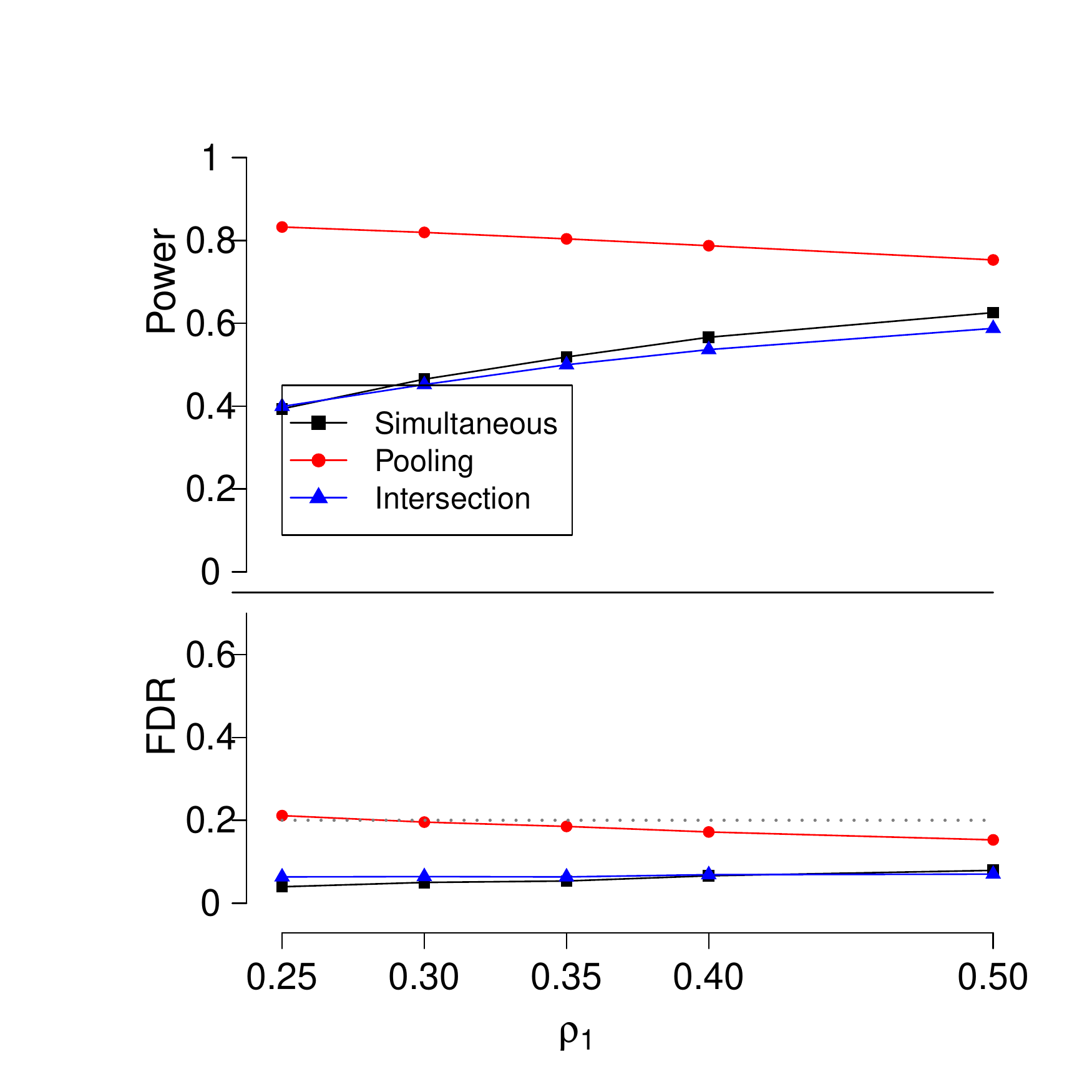}
    \includegraphics[scale=0.3]{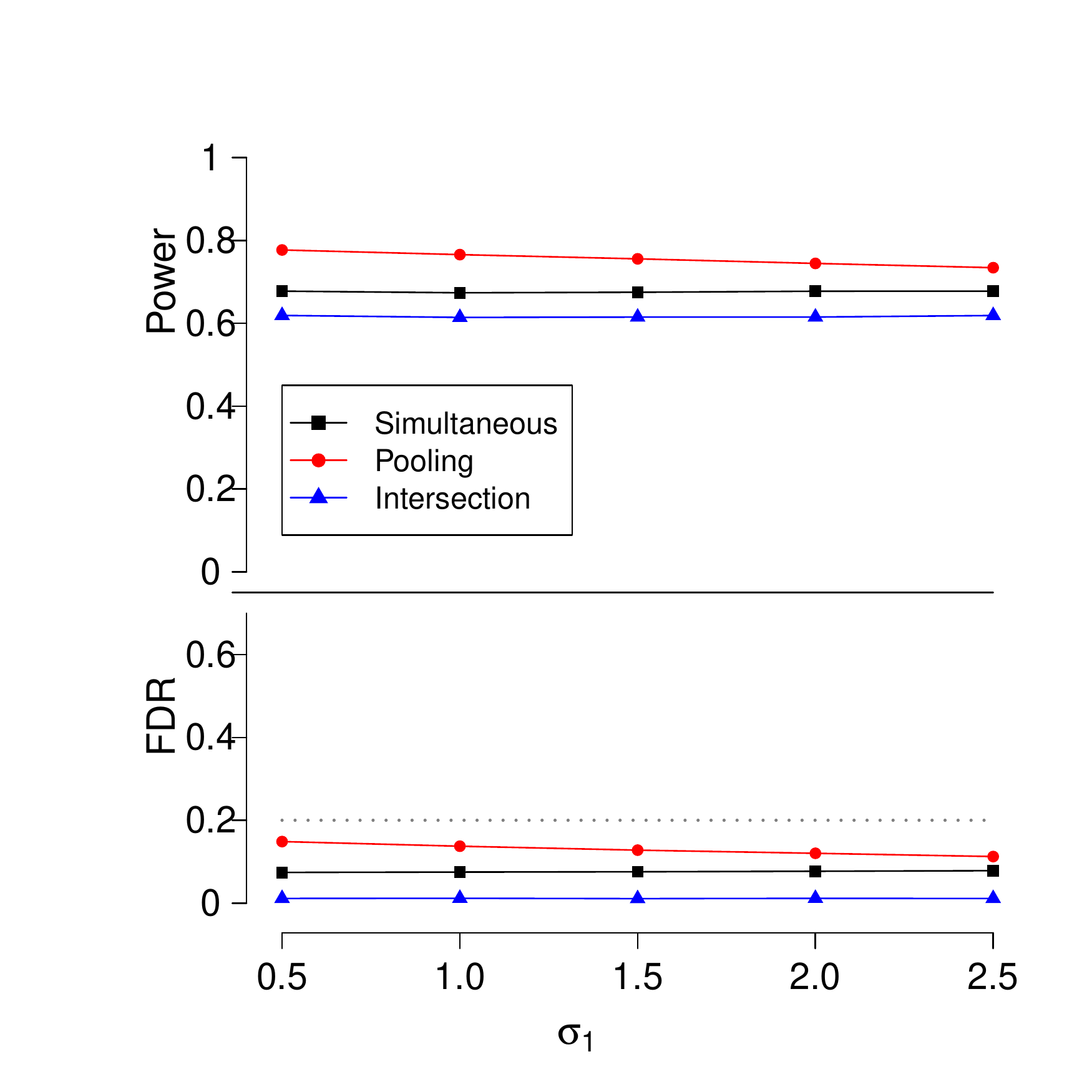}
    \includegraphics[scale=0.3]{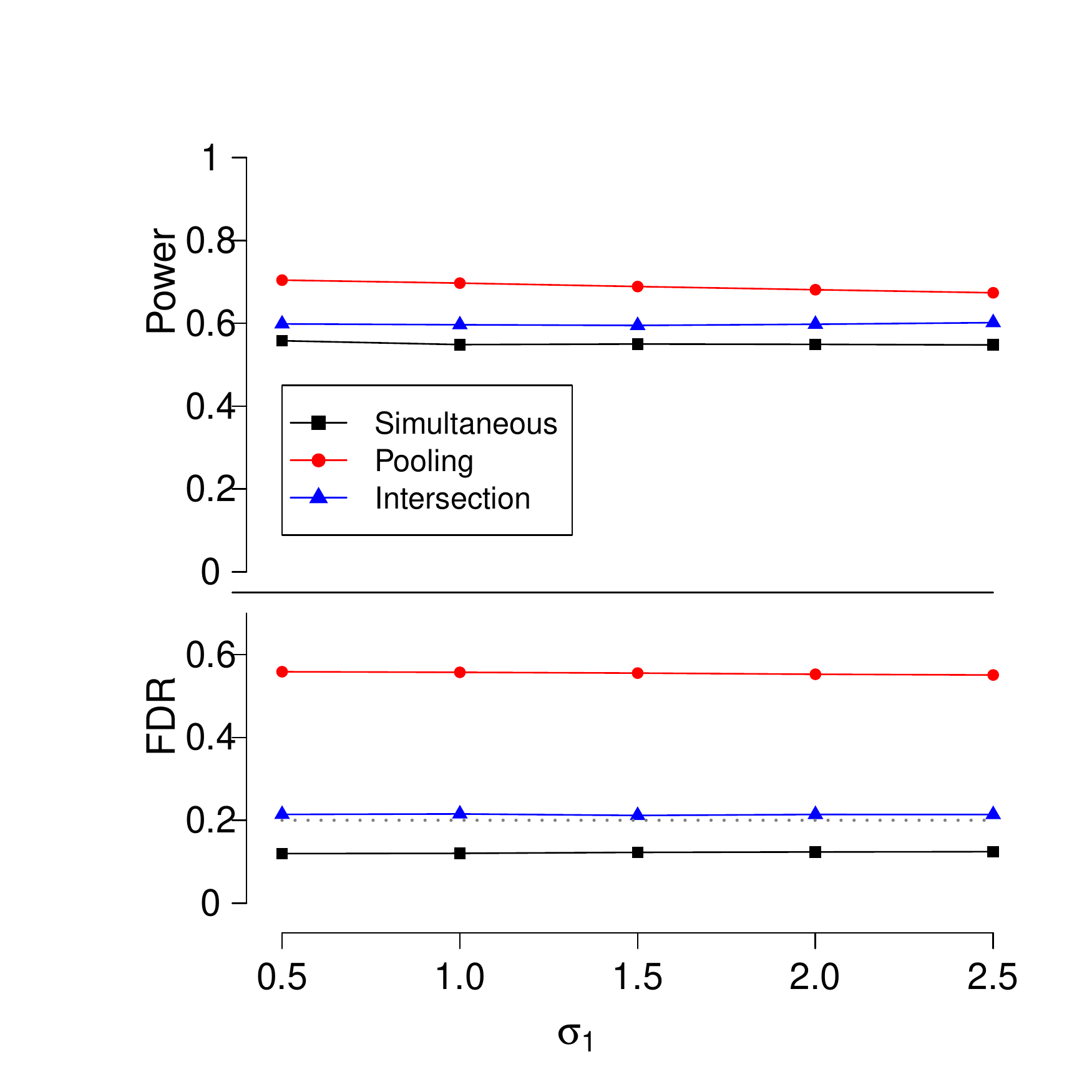}
    \includegraphics[scale=0.3]{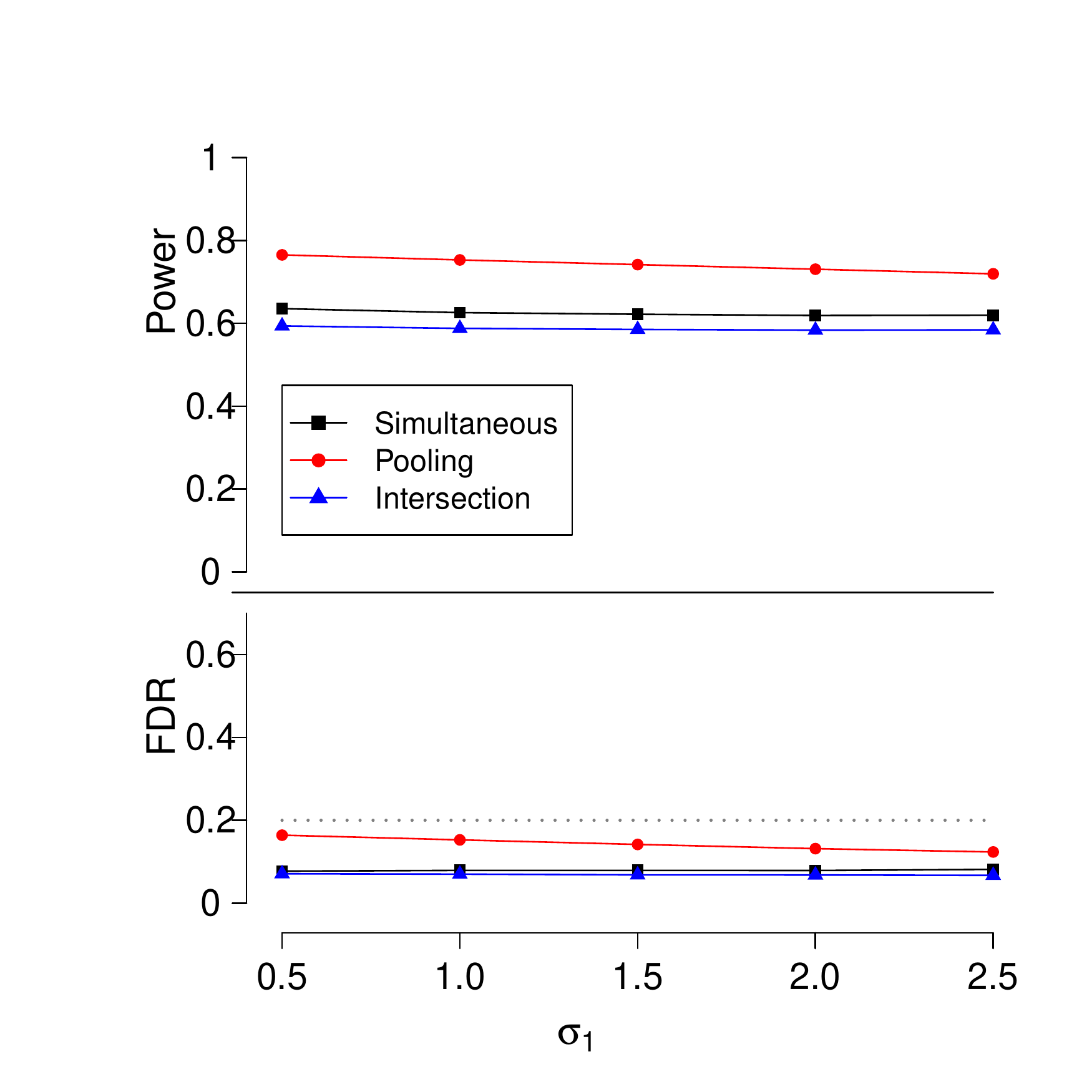}
    \caption{The power and the FDR for simulations with Setting 3 (mixed $K=2$) and Scenario 2 data. Column 1 includes the settings with $s_1=s_2=0$, column 2 includes the settings with $s_1=s_2\neq 0$, column 3 includes the settings with $s_1\neq 0, s_2=0$. Row 1 shows the experiments varying $s_0, s_1, s_2$, row 2 shows experiments varying $\rho_1, \rho_2$, row 3 shows experiments varying $\sigma_1, \sigma_2$.}
    \label{fig:mix_samesig=0}
\end{figure}

For the mixed setting, the results are similar to the continuous setting. The \textit{pooling} method can only control the FDR when $s_1=s_2=0$. The \textit{intersection} method can control the FDR when $s_1$ and $s_2$ are small. However, it does not control the FDR when $s_1$ and $s_2$ become large ($s_1=s_2 \geq 50$). The power of the \textit{simultaneous} method is slightly better than the \textit{intersection} method when $s_1=s_2$.

\subsubsection*{C.3.4 Additional simulation results for continuous outcomes with $K=3$\\}
In this section, we present simulation results for the $K=3$ cases when the outcomes are continuous.

From Figure \ref{fig:K=3}, we can see that when the only signals are mutual signals ($s_1=s_2=s_3=s_{12}=s_{13}=s_{23}=0$), all methods control the FDR and the \textit{pooling} method has the highest power as expected. However, when there are signals that only occur within some of the samples, the \textit{pooling} method fails and has very high FDR levels. The \textit{intersection} method works fine when the false signals are only shown in one of the three features, but when the false signals are shown in two of the three features, the \textit{intersection} method cannot control FDR. The \textit{simultaneous} method can always control FDR as expected and its power is similar to the \textit{intersection} method.

\begin{figure}[!p]
    \centering
    \includegraphics[scale=0.4]{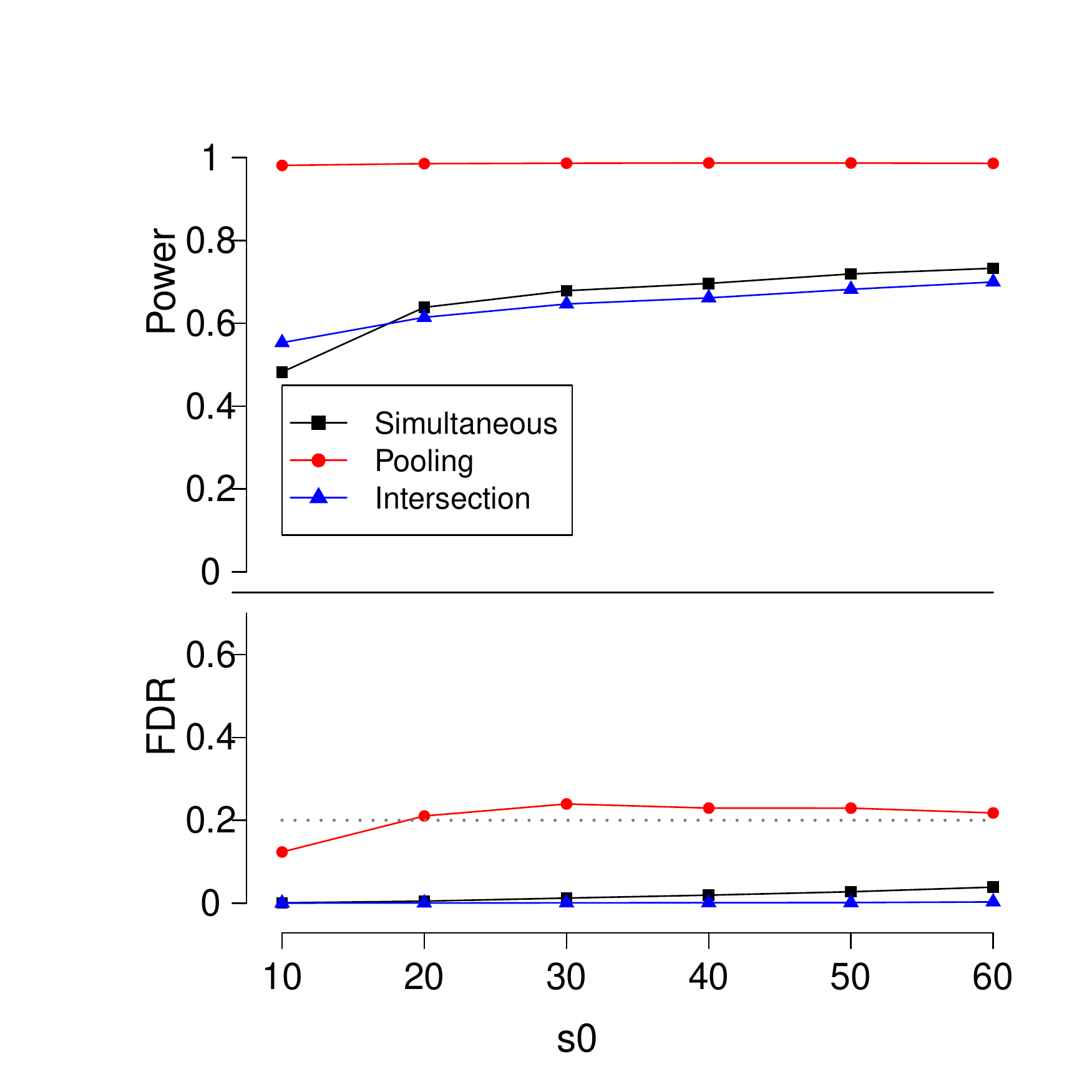}
    \includegraphics[scale=0.4]{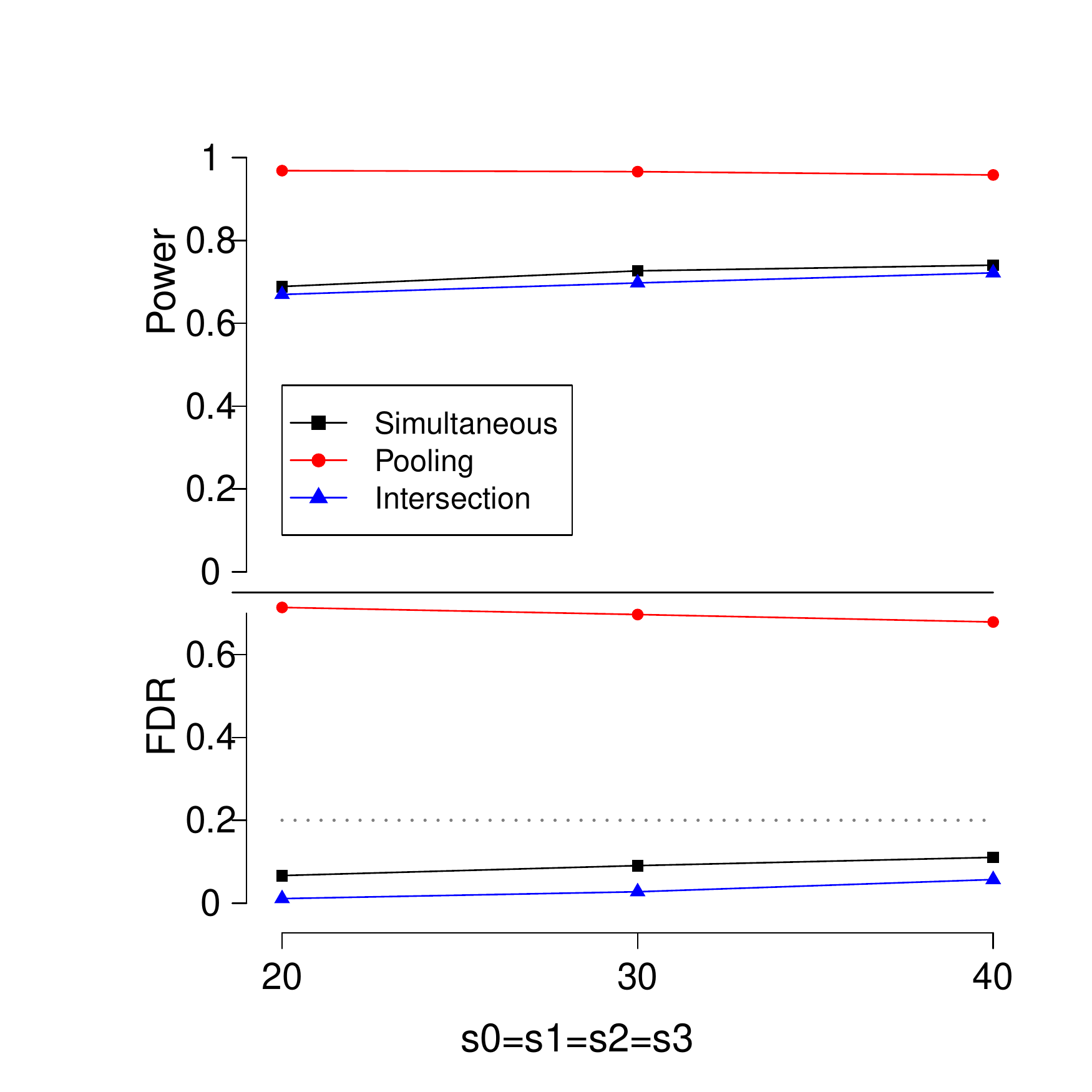}
    \includegraphics[scale=0.4]{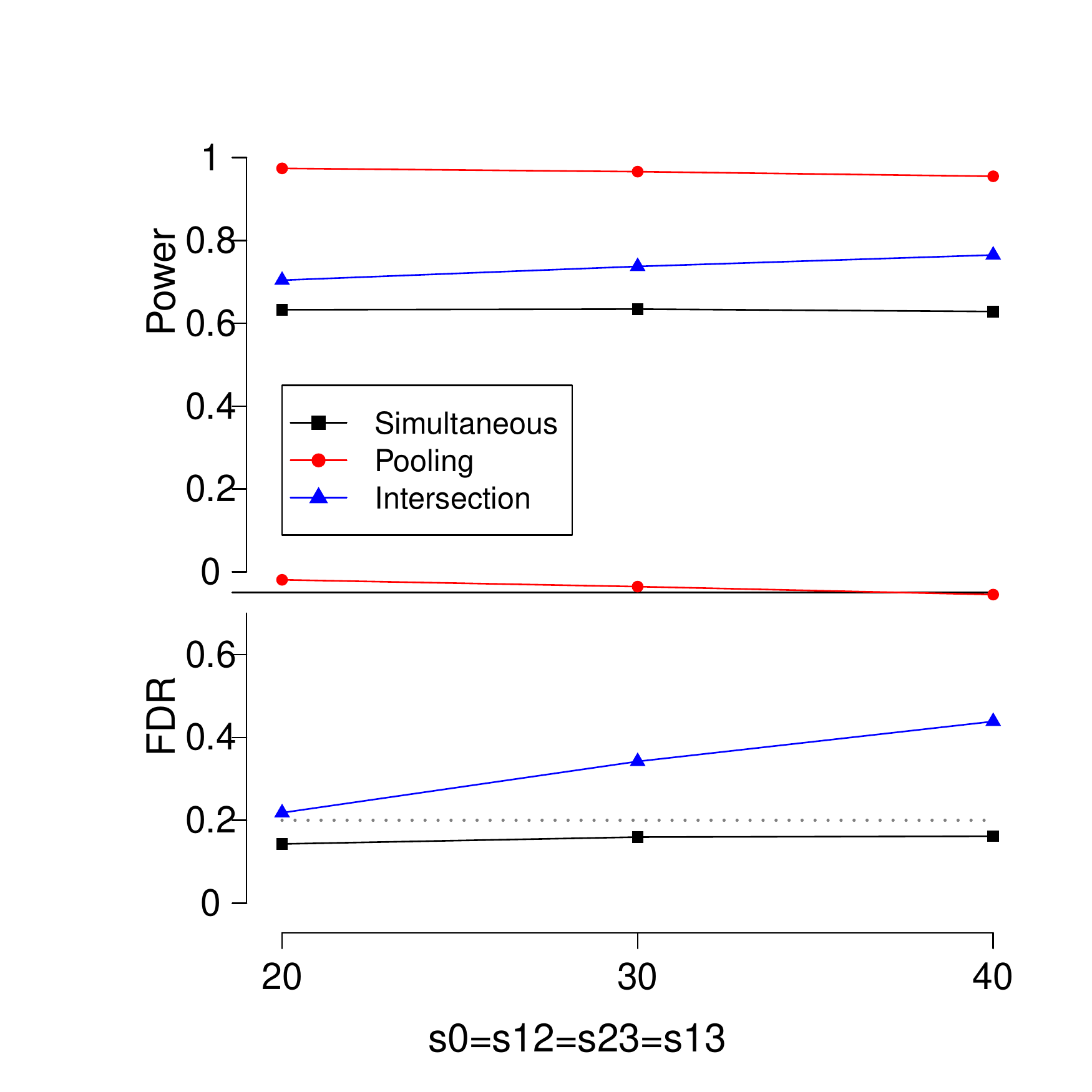}
     \includegraphics[scale=0.4]{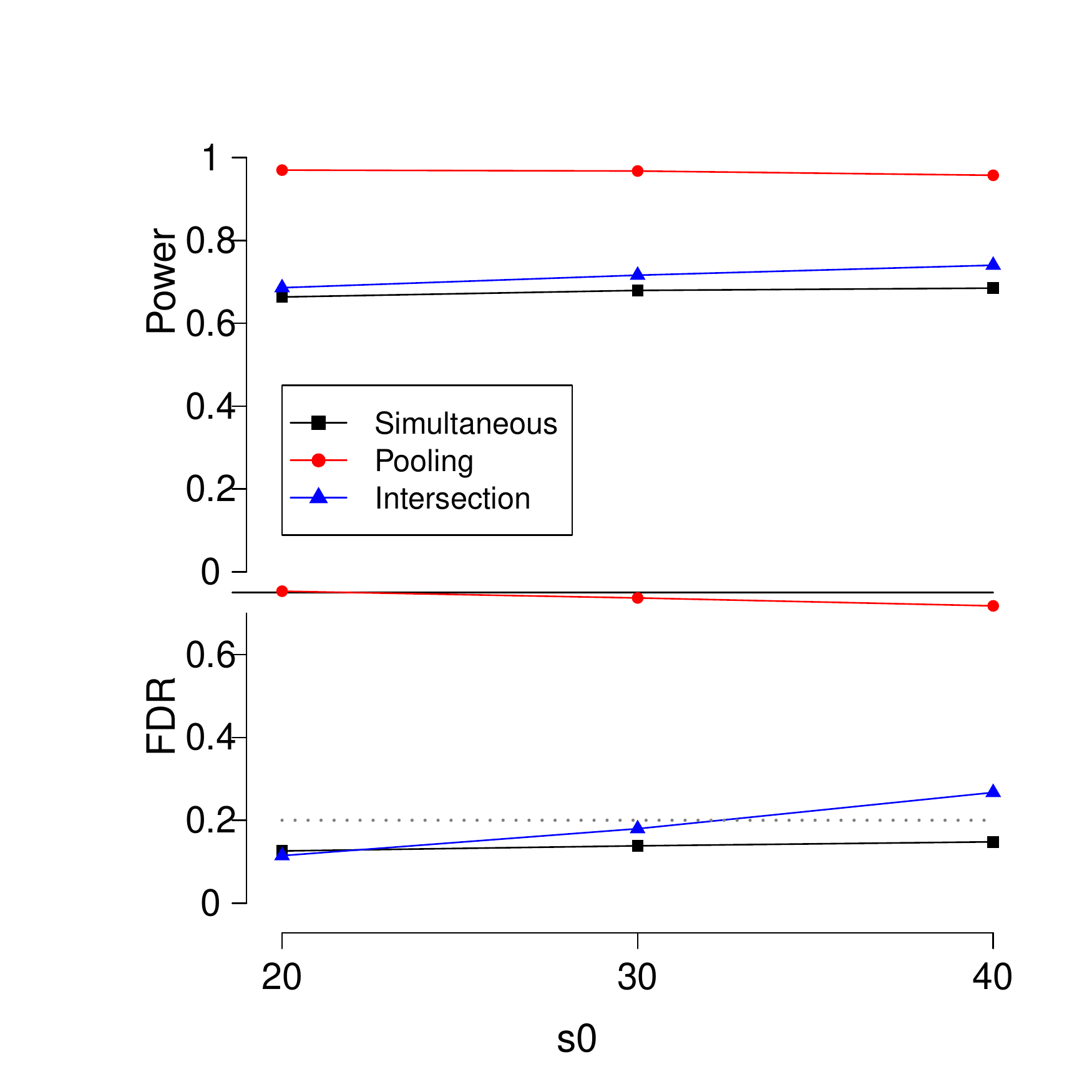}
    \caption{The power and the FDR for simulations with continuous data ($K=3$). Top left shows experiments varying $s_0$ while fix $s_1=s_2=s_3=s_{12}=s_{13}=s_{23}=0$, Top right shows experiments varying $s_0=s_1=s_2=s_3$ while fix $s_{12}=s_{13}=s_{23}=0$. Bottom left shows experiments varying $s_{12}=s_{13}=s_{23}=s_0$ while fix $s_1=s_2=s_3=0$. Bottom right shows experiments varying $s_{1}=s_{2}=s_{3}=s_{12}=s_{13}=s_{23}=s_0/2$.}
    \label{fig:K=3}
\end{figure}

\subsubsection*{C.3.5 Additional simulation results for power comparisons\\}
In this section, we plot the distributions of the filter statistics $W_j$ when assuming all underlying $Z_{kj}$, $k\in [K]$, and $j\in [p]$, are independent normal distributions. We assume $Z_{kj}=|Z^{\ast}_{kj}|$ where $Z^{\ast}_{kj}\sim N(0,1)$ when $j\in \mathcal{H}_k$ and $Z^{\ast}_{kj}\sim N(\mu,1)$ otherwise. We vary $\mu$ from 3 to 5 to denote weak and strong signals.

The symmetric distributions of $W_j$ in Figures \ref{fig:hist0}\ref{fig:hist1}\ref{fig:hist2}\ref{fig:hist1l}\ref{fig:hist2l} illustrate numerically that when the signals are not showing in all three datasets, the sign of $W_j$ is independent to its magnitude with $\PP{W_j>0}=1/2$ no matter how strong the signal is, which is just what we need for the FDR control theorem to work. The asymmetric distributions of $W_j$ in Figures \ref{fig:hist3}\ref{fig:hist3l} illustrate numerically that when the signals are showing in all three datasets, then we have a larger chance of observing positive $W_j$ and thus the test will have power (especially when the positive distribution of $W_j$ are far more toward right comparing to $W_j$ distribution for those features not having signals in all three data). Comparing these figures, we can see with the increase in the number or the strengths of signals in those null features that are only present in one or two groups will widen the distribution and thus lead to the lower power of the proposed method. 

\begin{figure}[!p]
    \centering
    \includegraphics[scale=0.9]{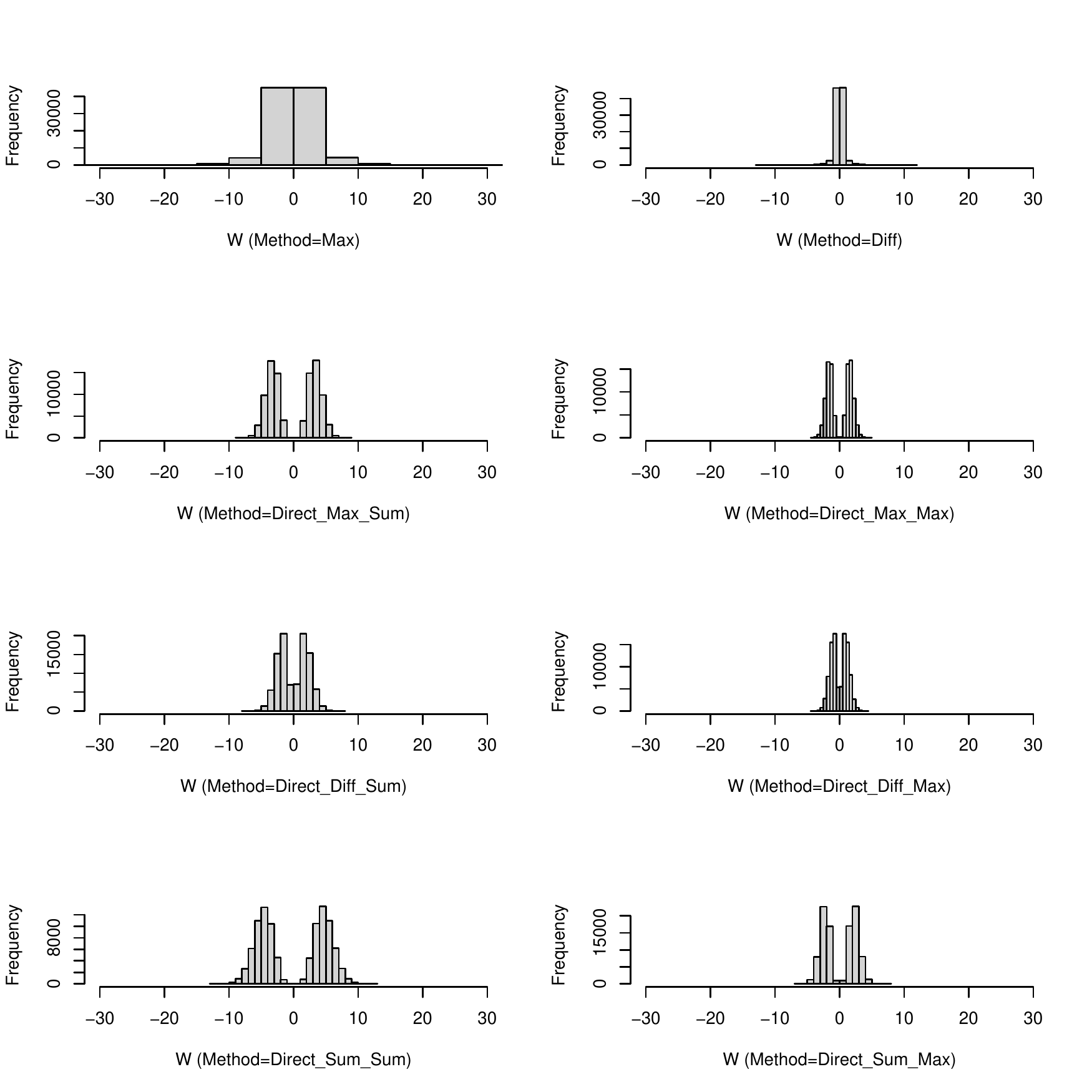}
    \caption{The distributions of the filter statistics $\W$ when the feature is not a signal in any of the three datasets.}
    \label{fig:hist0}
\end{figure}

\begin{figure}[!p]
    \centering
    \includegraphics[scale=0.9]{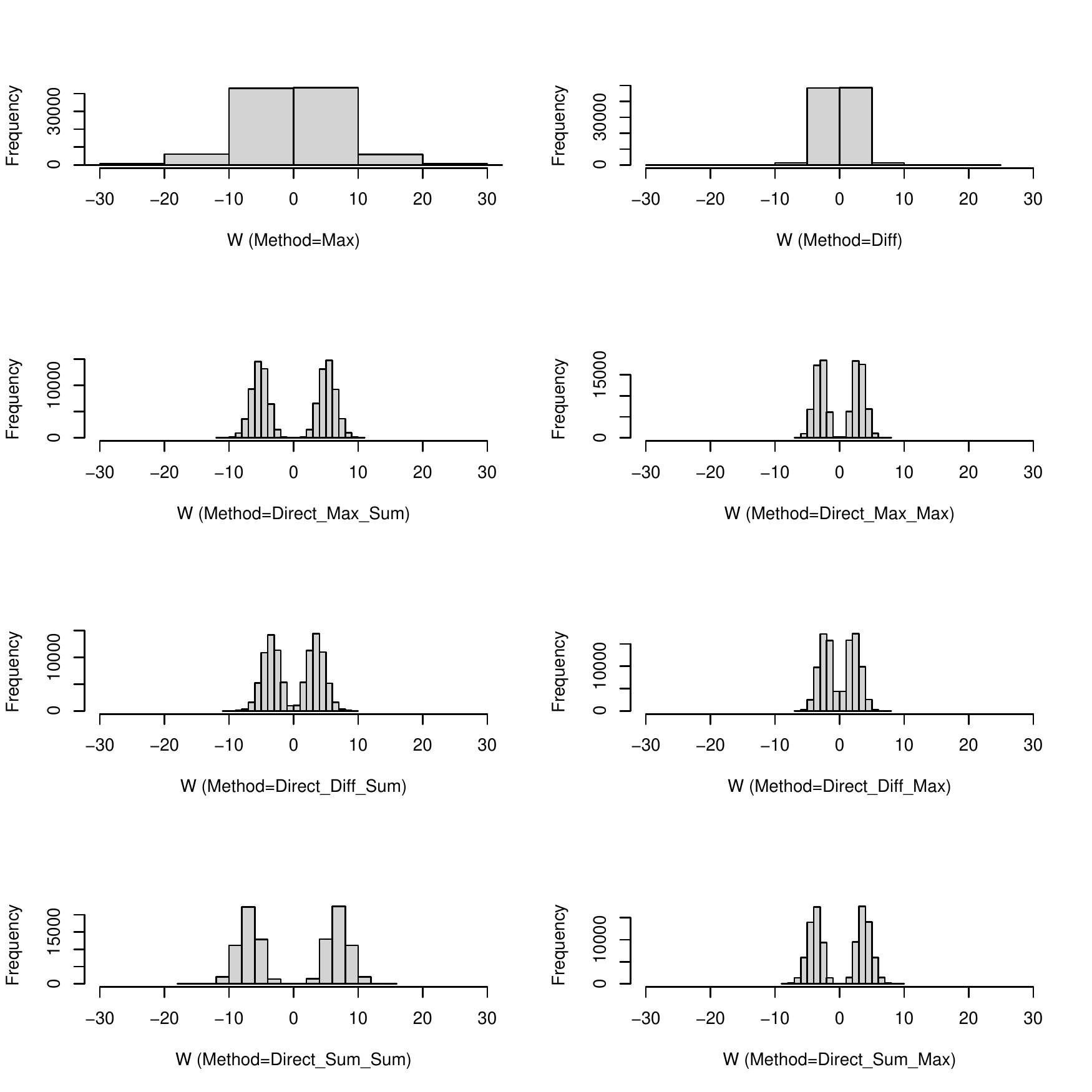}
    \caption{The distributions of the filter statistics $\W$ when the feature is a weak signal in only one of the three datasets.}
    \label{fig:hist1}
\end{figure}

\begin{figure}[!p]
    \centering
    \includegraphics[scale=0.9]{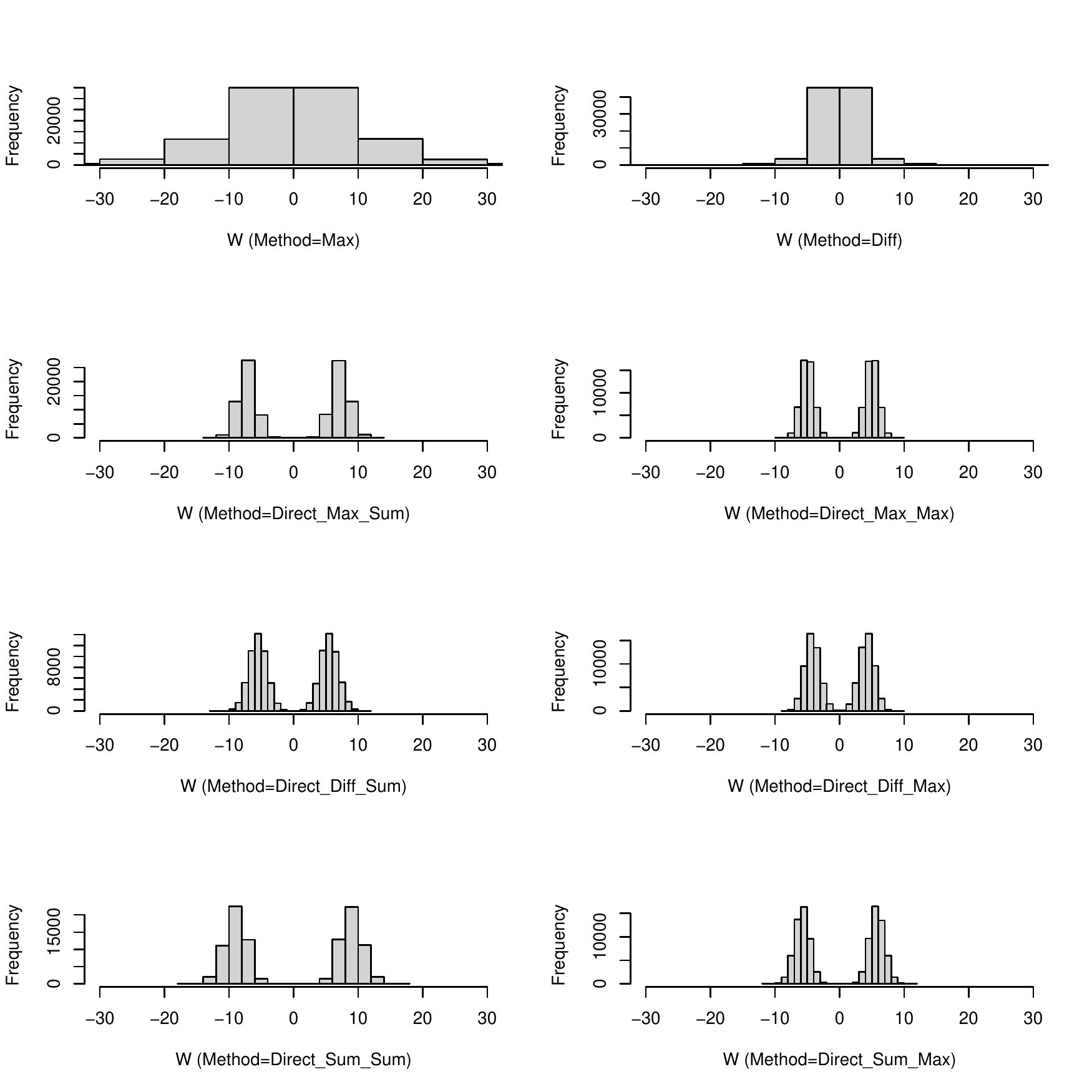}
    \caption{The distributions of the filter statistics $\W$ when the feature is a strong signal in only one of the three datasets.}
    \label{fig:hist1l}
\end{figure}

\begin{figure}[!p]
    \centering
     \includegraphics[scale=0.9]{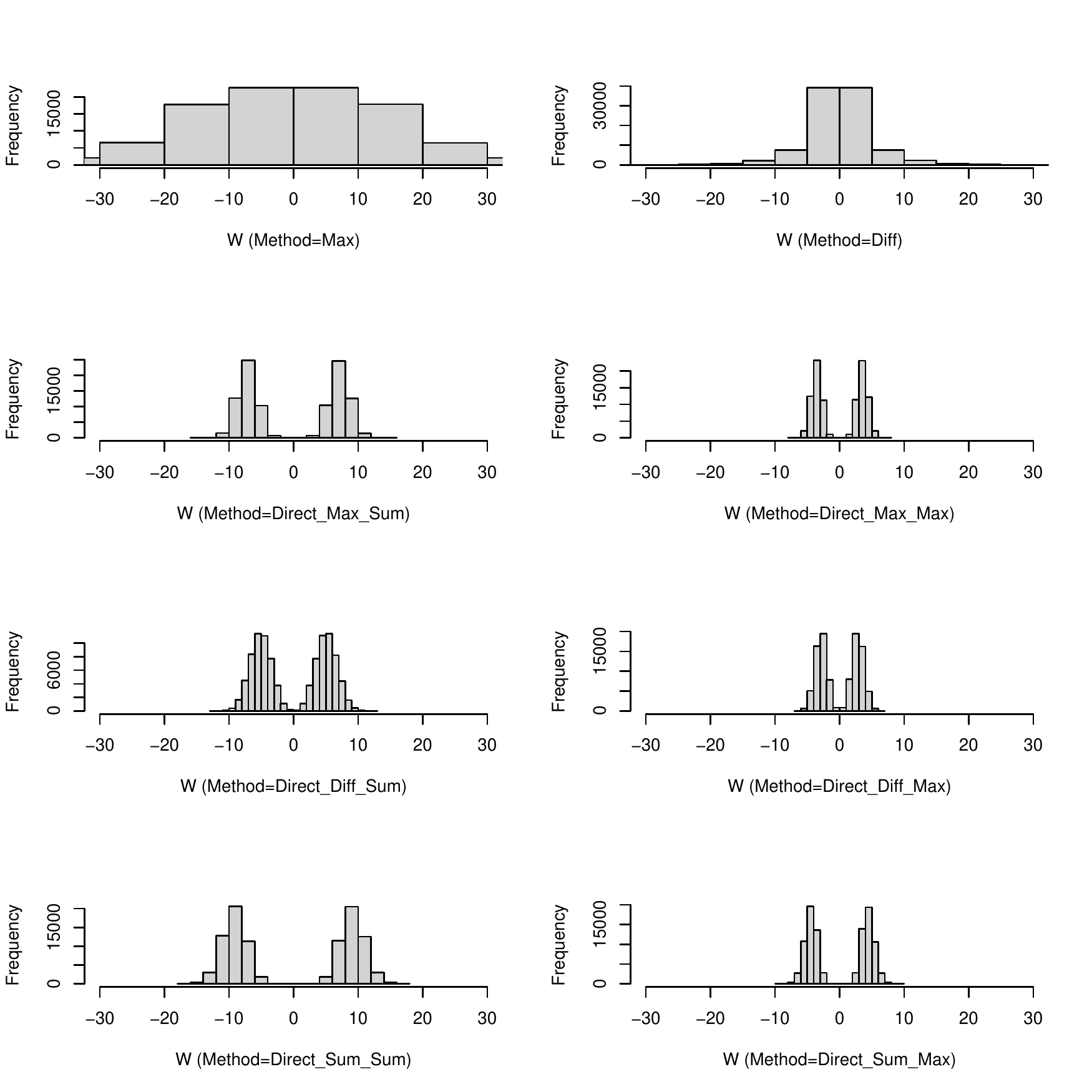}
    \caption{The distributions of the filter statistics $\W$ when the feature is a weak signal in only two of the three datasets.}
    \label{fig:hist2}
\end{figure}

\begin{figure}[!p]
    \centering
    \includegraphics[scale=0.9]{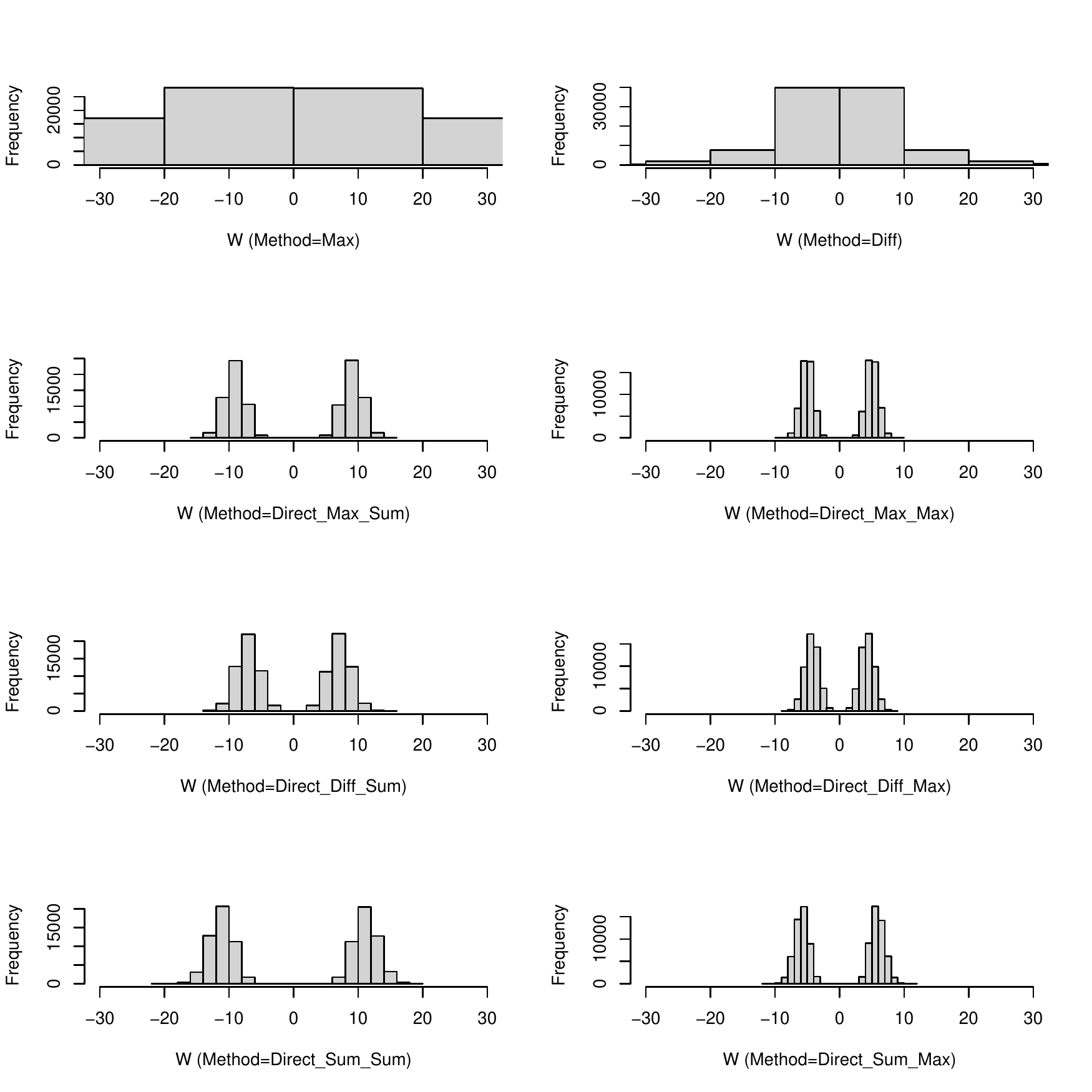}
    \caption{The distributions of the filter statistics $\W$ when the feature is a strong signal in only two of the three datasets.}
    \label{fig:hist2l}
\end{figure}

\begin{figure}[!p]
    \centering
    \includegraphics[scale=0.9]{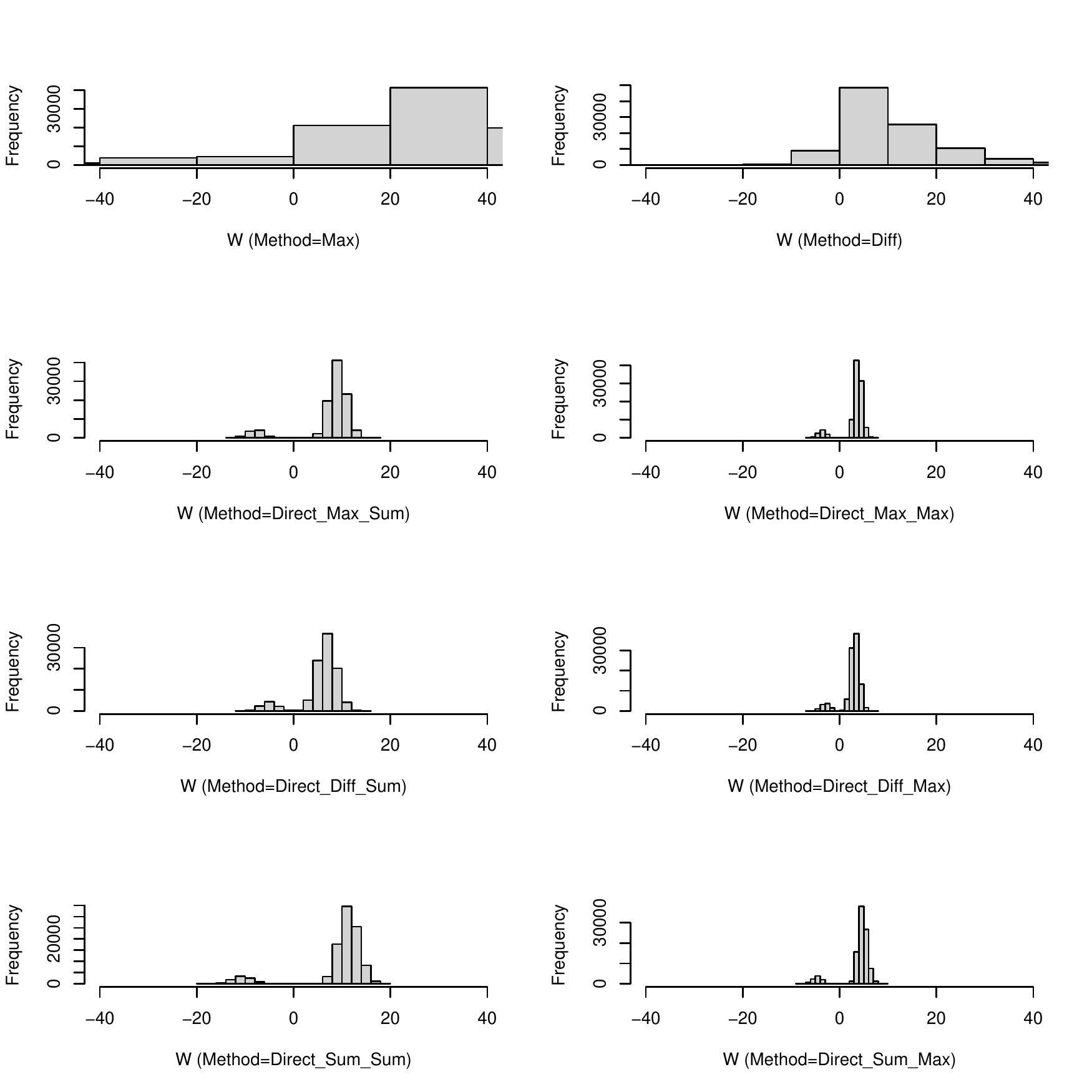}
    \caption{The distributions of the filter statistics $\W$ when the feature is a weak signal in all three datasets.}
    \label{fig:hist3}
\end{figure}

\begin{figure}[!p]
    \centering
     \includegraphics[scale=0.9]{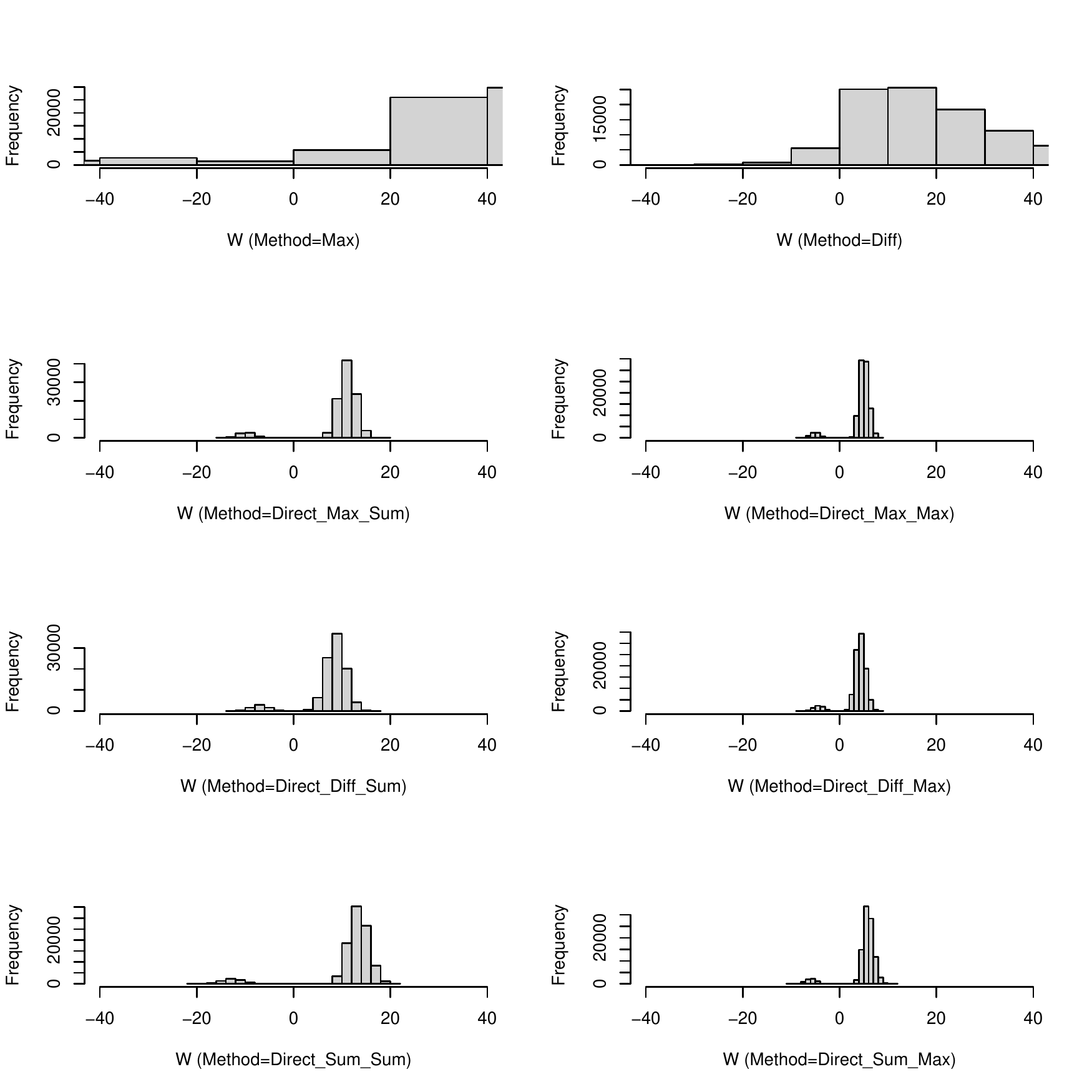}
    \caption{The distributions of the filter statistics $\W$ when the feature is a strong signal in all three datasets.}
    \label{fig:hist3l}
\end{figure}

The simulation results comparing the power of different filter statistics from the setting described in section C.1.3. are given in Table \ref{tab:res_comp}. From Table \ref{tab:res_comp} we can see that the signed max function and the difference function have the best performance among the functions we explored. 

\begin{table}[!p]
    \centering
    \begin{tabular}{|c|c|c|c|c|}
    \hline
        Scenario& $s_1=s_2$&  Method& FDR & Power\\
    \hline
         1&0&Max&0.08 &0.83\\
         1&0&Diff&0.08 &0.83\\
         1&0&Direct Max Sum&0.02 &0.82\\
         1&0&Direct Max Max&0.01 &0.82\\
         1&0&Direct Diff Sum&0.02 &0.82\\
        1&0&Direct Diff Max&0.01 &0.82\\
        1&0&Direct Sum Sum&0.01 &0.82\\
        1&0&Direct Sum Max&0.02 &0.82\\
    \hline
        1&40&Max&0.18 &0.80\\
         1&40&Diff&0.18 &0.80\\
         1&40&Direct Max Sum&0.06 &0.47\\
         1&40&Direct Max Max&0.02 &0.01\\
         1&40&Direct Diff Sum&0.05 &0.47\\
        1&40&Direct Diff Max&0.02 &0.01\\
        1&40&Direct Sum Sum&0.06 &0.48\\
        1&40&Direct Sum Max&0.02 &0.01\\
    \hline
        2&40&Max&0.15 &0.77\\
         2&40&Diff&0.15 &0.77\\
         2&40&Direct Max Sum&0.05 &0.43\\
         2&40&Direct Max Max&0.01 &0.01\\
         2&40&Direct Diff Sum&0.05 &0.42\\
        2&40&Direct Diff Max&0.01 &0.01\\
        2&40&Direct Sum Sum&0.05 &0.43\\
        2&40&Direct Sum Max&0.01 &0.01\\
    \hline
    \end{tabular}
    \vspace{0.2in}
    \caption{Empirical FDR and power comparisons among different choices of filter statistics $\W$ for $K=2$ and $q=0.2$ with the settings $s_1=s_2=0$ and $s_1=s_2=40$ with or without same signal strengths between two data sets (Scenarios 1 and 2).}
    \label{tab:res_comp}
\end{table}

The power and the FDR for different methods are shown in Figure \ref{fig:power} when varying the ratio of signal strength between mutual signals and non-mutual signals. We can see that for both the \textit{pooling} method and our proposed \textit{simultaneous} method, the power decreases a lot when the non-mutual signals are strong. But unlike the \textit{pooling} method, which has an increased FDR, our proposed \textit{simultaneous} method has a relatively stable FDR at the nominal level.
\begin{figure}[!p]
    \centering
    \includegraphics{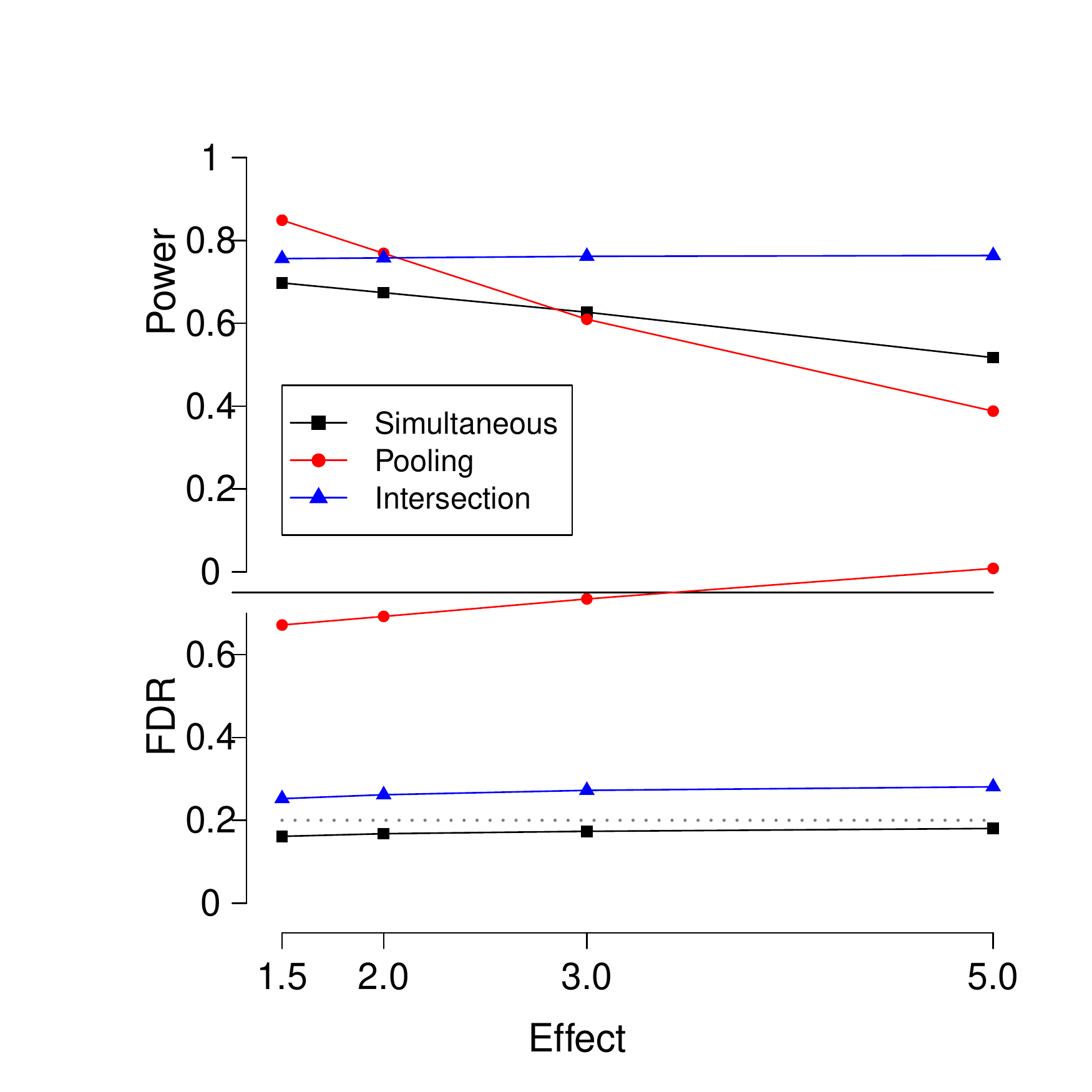}
    \caption{The power and the FDR comparisons among the three different methods when varying the ratio between the strengths of non-mutual signals and mutual signals.}
    \label{fig:power}
\end{figure}

Here we provide some insights into why our proposed filter statistics $\W$ will have power. Without loss of generality, we can assume all $Z_{jk}$'s and $\widetilde{Z}_{jk}$'s are non-negative (such as the absolute value of the standardized coefficient). Then we have for fixed $K$, as $n\rightarrow \infty$, $\widetilde{Z}_{jk}$ will converge to 0 for all $j,k$, while $Z_{jk}$ converge to 0 for $j\in H_0^k$ and a positive number otherwise. Thus we have that if $j\in \mathcal{H}$, then both $Z_j$ and $\widetilde{Z}_j$ will converge to 0 while if $j\in \mathcal{S}$, then $Z_j \approx \prod_{k=1}^K Z_{jk}$ converges to a positive number while $\widetilde{Z}_j$ converges to 0. Therefore, we have that the power will go to 1 as $n\rightarrow \infty$. The power is a monotonically decreasing function of $K$ and a monotonically increasing function of $n$. If $K$ grows with the sample size $n$, then the power might not reach 1. The rate of power growth and requirements (i.e. the growth condition of $K(n)$) to let the power reach 1 is complicated, beyond the scope of this paper, and worth future investigations.

\newpage
\section*{Web Appendix D: Additional details for real data analyses}
\subsection*{D.1: Implementation details for the analysis of the communities and crime data} 
In this section, we provide the details of how the four methods are implemented for our analysis of the community crime data.
\begin{itemize}
\item \textit{Simultaneous Knockoffs}: In Step 1, we use the Fixed-X knockoff construction method to create $\widetilde{\X}^k$, $k=1,2$ for the two samples. In Step 2, we run the Lasso regression of $Y^k$ on $[\X^k\ \widetilde{\X}^k]$ for each $k$ and use the absolute value of the regression coefficients $\widehat{\bmbeta}(\lambda)$ with $\lambda$ selected from the 5-fold cross-validation to calculate $[\Z_k,\widetilde{\Z}_k]$. In Step 3, we use the OSSF $\W=\odot_{k=1}^2(\Z^k-\widetilde{\Z}^k)$. In Step 4, we use the Knockoff filter as defined in equation (7) with $q=0.1$ to obtain $\widehat{S}$.
\item \textit{Pooling Knockoffs}: In Step 1, we use the Fixed-X knockoff construction method to create $\widetilde{\X}$, for the pooled samples. In Step 2, we run the Lasso regression of $Y$ on $[\X\ \widetilde{\X}]$ for the pooled sample and use the absolute value of the regression coefficients $\widehat{\bmbeta}(\lambda)$ with $\lambda$ selected from the 5-fold cross-validation to calculate $[\Z,\widetilde{\Z}]$. In Step 3, we use the antisymmetric function $\W=\Z-\widetilde{\Z}$. In Step 4, we use the Knockoff filter as defined in equation (7) with $q=0.1$ to obtain $\widehat{S}$.
\item \textit{Intersection Knockoffs}: In Step 1, we use the Fixed-X knockoff construction method to create $\widetilde{\X}^k$, $k=1,2$ for the two samples. In Step 2, we run the Lasso regression of $Y^k$ on $[\X^k\ \widetilde{\X}^k]$ for each $k$ and use the absolute value of the regression coefficients $\widehat{\bmbeta}(\lambda)$ with $\lambda$ selected from the 5-fold cross-validation to calculate $[\Z^k,\widetilde{\Z}^k]$. In Step 3, we use the antisymmetric function $\W^k=\Z^k-\widetilde{\Z}^k$ for $k=1,2$. In Step 4, we use the Knockoff filter as defined in equation (7) with $q=0.1$ to obtain $\widehat{S}_k$, $k=1,2$, and compute $\widehat{S}=\cap_{k=1}^2\widehat{S}_k$.
\item \textit{repfdr}: We run linear regressions of outcome $Y^k$ on $\X^k$ for $k=1,2$ and use the value of the $z$-statistics, (i.e., $\widehat{\beta}_j/\widehat{SE}(\widehat{\beta}_j)$) as test statistics. We consider the test allowing signals to have different signs among the two samples by setting \textit{n.association.status} as 2 and we set the number of bins in the discretization of the z-score as 100. The default setting of the natural spline is used and no trimming is performed on the z-score. We use the \textit{replication} option (which is equivalent to testing of union nulls when $K=2$) in the R function \textit{repfdr} from the R package \textit{repfdr}. 
\end{itemize}

\subsection*{D.2: Implementation details for the analysis of the TCGA data} 
In this section, we provided the details of how data is analyzed. The main analysis is based on a complete case analysis removing observations with missing values in any of the variables that passed the pre-screening step. The sensitivity analysis is based on an analysis with missing data imputed by mean. The primary screening is for top $d=n/log(n)=79$ genes and the sensitivity analysis tries using a fixed threshold $p<0.0002$ which leads to 111 genes passing the pre-screening step. Below are the details on how the four methods are implemented for this data.
\begin{itemize}
\item \textit{Simultaneous Knockoffs}: In Step 1, we use the Model-X knockoff second-order construction method to create $\widetilde{\X}^k$, $k=1,2$ for the two samples. In Step 2, we run the $\ell_1$-penalized Cox regression of $Y^k$ on $[\X^k\ \widetilde{\X}^k]$ for each $k$ and use the absolute value of the regression coefficients $\widehat{\bmbeta}(\lambda)$ with $\lambda$ selected from the 5-fold cross-validation to calculate $[\Z^k,\widetilde{\Z}^k]$. In Step 3, we use the OSSF $\W=\odot_{k=1}^2(\Z^k-\widetilde{\Z}^k)$. In Step 4, we use the Knockoff filter as defined in equation (7) with $q=0.1$ to obtain $\widehat{S}$.
\item \textit{Pooling Knockoffs}: In Step 1, we use the Model-X knockoff second-order construction method to create $\widetilde{\X}$, for the pooled samples. In Step 2, we run $\ell_1$-penalized Cox regression of $Y$ on $[\X\ \widetilde{\X}]$ for the pooled sample and use the absolute value of the regression coefficients $\widehat{\bmbeta}(\lambda)$ with $\lambda$ selected from the 5-fold cross-validation to calculate $[\Z,\widetilde{\Z}]$. In Step 3, we use the antisymmetric function $\W=\Z-\widetilde{\Z}$. In Step 4, we use the Knockoff filter as defined in equation (7) with $q=0.1$ to obtain $\widehat{S}$.
\item \textit{Intersection Knockoffs}: In Step 1, we use the Model-X knockoff second-order construction method to create $\widetilde{\X}_k$, $k=1,2$ for the two samples. In Step 2, we run $\ell_1$-penalized Cox regression of $Y^k$ on $[\X^k,\widetilde{\X}^k]$ for each $k$ and use the absolute value of the regression coefficients $\widehat{\bmbeta}(\lambda)$ with $\lambda$ selected from the 5-fold cross-validation to calculate $[\Z^k,\widetilde{\Z}^k]$. In Step 3, we use the antisymmetric function $\W^k=\Z^k-\widetilde{\Z}^k$ for $k=1,2$. In Step 4, we use the Knockoff filter as defined in equation (7) with $q=0.1$ to obtain $\widehat{S}_k$, $k=1,2$, and compute $\widehat{S}=\cap_{k=1}^2\widehat{S}_k$.
\item \textit{repfdr}: We run the Cox regressions of the survival outcome $Y^k$ on $\X^k$ for $k=1,2$ and use the value of the $z$-statistics, (i.e., $\widehat{\beta}_j/\widehat{SE}(\widehat{\beta}_j)$) as test statistics. We consider the test allowing signals to have different signs among the two samples by setting \textit{n.association.status} as 2 and we set the number of bins in the discretization of the z-score as 100. The default setting of the natural spline is used and no trimming is performed on the z-score. We use the \textit{replication} option (which is equivalent to testing of union nulls when $K=2$) in the R function \textit{repfdr} from the R package \textit{repfdr}. 
\end{itemize}

\end{document}